\documentclass[11pt]{article}
\usepackage{graphicx}
\usepackage{latexsym}
\usepackage{amssymb}
\usepackage{amsmath}
\usepackage{amsthm}
\usepackage{forloop}
\usepackage[paper=letterpaper,margin=1in]{geometry}
\usepackage[ruled,vlined,linesnumbered]{algorithm2e}\DecMargin{2mm}
\usepackage{paralist}
\usepackage{comment} 
\usepackage{bbm}
\usepackage{framed} 
\usepackage{booktabs}
\usepackage{thmtools} 
\usepackage{thm-restate}
\usepackage{nicefrac}
\usepackage{calc}
\usepackage{scalerel}
\usepackage[usenames,dvipsnames]{xcolor}
\usepackage{mleftright}\mleftright
\usepackage{url} 
\usepackage{hyperref}

\makeatletter
\patchcmd{\@algocf@start}
  {-1.5em}
  {0pt}
  {}{}
\makeatother

\definecolor{darkgreen}{rgb}{0,0.5,0}
\hypersetup{
    unicode=false,            
    colorlinks=true,          
    linkcolor=blue!70!black,  
    citecolor=darkgreen,      
    filecolor=magenta,        
    urlcolor=blue!70!black    
}

\newtheorem{theorem}{Theorem}[section]
\newtheorem{lemma}[theorem]{Lemma}

\newtheorem{corollary}[theorem]{Corollary}
\newtheorem{definition}{Definition}[section]
\newtheorem{proposition}[theorem]{Proposition}
\newtheorem{observation}[theorem]{Observation}

\newcommand{\defcal}[1]{\expandafter\newcommand\csname c#1\endcsname{{\mathcal{#1}}}}
\newcommand{\defbb}[1]{\expandafter\newcommand\csname b#1\endcsname{{\mathbb{#1}}}}
\newcommand{\defvec}[1]{\expandafter\newcommand\csname v#1\endcsname{{\mathbf{#1}}}}
\newcounter{calBbCounter}
\forLoop{1}{26}{calBbCounter}{
    \edef\Letter{\Alph{calBbCounter}}
		\edef\letter{\alph{calBbCounter}}
    \expandafter\defcal\Letter
		\expandafter\defbb\Letter
		\expandafter\defvec\letter
}

\newcommand{\eps}{\varepsilon}

\newcommand{\nnR}{{\bR_{\geq 0}}}
\newcommand{\characteristic}{{\mathbf{1}}}

\newcommand{\inner}[2]{\left<#1, #2\right>}
\newcommand{\email}[1]{{\href{mailto:#1}{#1}}}
\newcommand{\vone}{\mathbf{1}}
\newcommand{\vzero}{\mathbf{0}}

\newcommand{\FWMCGFull}{{\textbf{Frank-Wolfe Guided Measured Continuous Greedy}}}
\newcommand{\FWMCG}{{\textbf{\upshape{FW-guided-MCG}}}\xspace}
\newcommand{\BoxMaximization}{{\textbf{\upshape{Box-Maximization}}}}
\newcommand{\psum}{\mathbin{\oplus}}
\newcommand{\hprod}{\mathbin{\odot}}
\DeclareMathOperator*{\PSum}{\scalerel*{\oplus}{\sum}}
\DeclareMathOperator*{\HProd}{\scalerel*{\hprod}{\sum}}
\newcommand{\RSet}{{\texttt{\upshape{R}}}}
\newcommand{\titleRef}[1]{{\texorpdfstring{\ref{#1}}{\ref*{#1}}}}
\newcommand{\XOR}{\mathbin{\text{XOR}}}
\newcommand{\io}{{i_o}}
\let\oldnl\nl
\newcommand{\nonl}{\renewcommand{\nl}{\let\nl\oldnl}}

\newcommand{\dotbigcup}{\charfusion[\mathop]{\bigcup}{\cdot}}
\DeclareMathOperator*{\Poly}{Poly}

\makeatletter
\def\moverlay{\mathpalette\mov@rlay}
\def\mov@rlay#1#2{\leavevmode\vtop{%
   \baselineskip\z@skip \lineskiplimit-\maxdimen
   \ialign{\hfil$\m@th#1##$\hfil\cr#2\crcr}}}
\newcommand{\charfusion}[3][\mathord]{
    #1{\ifx#1\mathop\vphantom{#2}\fi
        \mathpalette\mov@rlay{#2\cr#3}
      }
    \ifx#1\mathop\expandafter\displaylimits\fi}
\makeatother

\begin{document}

\pagenumbering{arabic}

\title{Constrained Submodular Maximization via New Bounds for DR-Submodular Functions}

\author{Niv Buchbinder\thanks{Department of Statistics and Operations Research, Tel Aviv University. E-mail: \email{niv.buchbinder@gmail.com}} \and
				Moran Feldman\thanks{Department of Computer Science, University of Haifa. E-mail: \email{moranfe@cs.haifa.ac.il}}}

\maketitle
\thispagestyle{empty}
\pagenumbering{Alph}
\begin{abstract}
Submodular maximization under various constraints is a fundamental problem studied continuously, in both computer science and operations research, since the late $1970$'s. A central technique in this field is to approximately optimize the multilinear extension of the submodular objective, and then round the solution. The use of this technique requires a solver able to approximately maximize multilinear extensions. Following a long line of work, Buchbinder and Feldman~\cite{buchbinder2019constrained} described such a solver guaranteeing $0.385$-approximation for down-closed constraints, while Oveis Gharan and Vondr{\'{a}}k~\cite{gharan2011submodular} showed that no solver can guarantee better than $0.478$-approximation. In this paper, we present a solver guaranteeing $0.401$-approximation, which significantly reduces the gap between the best known solver and the inapproximability result. The design and analysis of our solver are based on a novel bound that we prove for DR-submodular functions. This bound improves over a previous bound due to Feldman et al.~\cite{feldman2011unified} that is used by essentially all state-of-the-art results for constrained maximization of general submodular/DR-submodular functions. Hence, we believe that our new bound is likely to find many additional applications in related problems, and to be a key component for further improvement.
\end{abstract}
\newpage
\pagenumbering{arabic}

\newtoggle{appendix}
\togglefalse{appendix}
\section{Introduction} \label{sec:introduction}

Submodular maximization under various constraints is a fundamental problem, and has been studied continuously, in both computer science and operations research, since the late $1970$'s~\cite{conforti1984submodular,fisher1978analysis,hausmann1978greedy,hausmann1980worst,jenkyns1976efficacy,korte1978analysis,nemhauser1978best,nemhauser1978analysis}. Given a submodular set function $f\colon 2^\cN \to \nnR$ over a ground set $\cN$ and a family of feasible subsets $\cI \subseteq 2^\cN$, the problem is to find a subset $A\in \cI$ maximizing $f(A)$.\footnote{We refer the reader to Section \ref{sec:preliminaries} for formal definitions.} 
Submodular functions are a rich class of functions that includes many functions of interest, such as cuts functions of graphs and directed graphs, the mutual information function, matroid weighted rank functions and log-determinants. Hence, many well-known problems in combinatorial optimization can be cast as special cases of submodular maximization problems. A few examples are Max-Cut~\cite{goemans1995improved,hastad2001optimal,karp1972reducibility,khot2007optimal,trevisan2000gadgets}, Max-DiCut~\cite{feige1995approximating,goemans1995improved,halperin2001combinatorial}, Generalized Assignment~\cite{chekuri2005polynomial,cohen2006efficient,feige2006approximation,fleischer2006tight}, Max-$k$-Coverage~\cite{feige1998threshold,khuller1999budgeted}, Max-Bisection~\cite{austrin2016better,frieze1997improved} and Facility Location~\cite{ageev1999approximation,cornuejols1977location,cornuejols1977uncapacitated}. From a more practical perspective, submodular maximization problems have found uses in social networks~\cite{hartline2008optimal,kempe2015maximizing}, vision~\cite{boykov2001interactive,jegelka2011submodularity}, machine learning~\cite{krause2005near,krause2008efficient,krause2008near,lin2010multidocument,lin2011class} and many other areas (the reader is referred, for example, to a comprehensive survey by Bach~\cite{bach2013foundations}). 
Submodular maximization has also been studied in various computational models, such as online and secretary (random arrival) settings~\cite{bateni2013submodular,buchbinder2020online,buchbinder2019online,feldman2018goes}, streaming~\cite{alaluf2022optimal,chekuri2015streaming,feldman2022streaming,feldman2020oneway,kazemi2019submodular,levin2021streaming}, parallel computation \cite{balkanski2022adaptive2,balkanski2018adaptive1,chekuri2019parallel1,chekuri2019parallel2} and distributed computing~\cite{mirzasoleiman2016distributed,pontebarbosa2016new}. Some works also took a game theoretic perspective on submodular optimization~\cite{assadi2020improved,dobzinski2021gross,mualem2023submodular}.

To tackle submodular maximization problems, earlier work used direct combinatorial approaches such as local search or greedy variants~\cite{buchbinder2015tight,feige11maximizing,feldman2011improved,lee2010maximizing,lee2010submodular,maxim2004note}. However, such algorithms tend to be highly tailored for the specific details of the problem at hand, making it difficult to get generally applicable results.
A second approach that emerged is the following. First, the feasible set $\cI$ is relaxed to a convex body $P\subseteq [0,1]^\cN$, and the objective function $f$ is extended to a function $F\colon [0,1]^\cN \to \nnR$. Defining the ``right'' extension function $F$ to use is tricky, as, unlike in the linear case, there is no single natural candidate. The extension that turned out to be the most useful is known as the \emph{multilinear extension}, first introduced by~\cite{calinescu2011maximzing}. Once $P$ and $F$ are defined, a two steps procedure is implemented. First, one finds a fractional solution $\vx \in P$ that approximately maximizes $F$. Then, the fractional solution $\vx$ is rounded to obtain an integral solution, while incurring a bounded loss in terms of the objective. This approach has been amazingly successful in resolving many open problems, and led to the introduction of many fundamental techniques. 

The rounding step of the above approach turned out to be successfully implementable for many kinds of constraints. For example, when $P$ is a matroid polytope, rounding can be done with no loss~\cite{calinescu2011maximzing,chekuri2010dependent}, and when $P$ is defined by a constant number of knapsack constraints, rounding results in only an arbitrarily small constant loss~\cite{kulik2013approximations} (see also~\cite{bruggmann2022optimal,chekuri2014submodular,feldman13maximization,qiu2022submodular} for other examples of rounding in submodular maximization problems). Hence, the main bottleneck and challenge in most cases is finding a good fractional solution to round. The first suggested (and prime example of) algorithm for this purpose is the Continuous Greedy algorithm designed by C{\u{a}}linescu et al.~\cite{calinescu2011maximzing}. Their algorithm outputs a $(1 - \nicefrac{1}{e})$-approximate solution when the submodular function is monotone.\footnote{A set function $f\colon 2^\cN \to \nnR$ is monotone if $f(A) \leq f(B)$ for every $A \subseteq B \subseteq \cN$.} Accompanied with the above mentioned rounding for matroid polytopes, this yielded the best possible approximation for maximizing a monotone submodular function subject to a matroid constraint, resolving a central problem that was open for $30$ years.

Finding a good fractional solution when $f$ is not guaranteed to be monotone proved to be a more challenging task. In fact, the only tight result currently known for such functions is a $\nicefrac{1}{2}$-approximation for the unconstrained case (in which every subset of $\cN$ is a feasible solution)~\cite{buchbinder2015tight,feige11maximizing,mualem2022using}. 
Nevertheless, there is a line work that aims to find a good fractional solution for maximizing a non-negative non-monotone submodular function $f$ subject to a general down-closed solvable polytope.\footnote{A polytope $P \subseteq [0, 1]^\cN$ is down-closed if $\vy \in P$ implies that every non-negative vector $\vx \leq \vy$ belongs to $P$ as well. The polytope $P$ is solvable if one can optimize linear functions subject to it.} The first work in this line was an algorithm of Chekuri et al.~\cite{chekuri2014submodular} guaranteeing $0.325$-approximation, which was soon improved over by an algorithm of Feldman et al.~\cite{feldman2011unified} termed Measured Continuous Greedy that guarantees an approximation ratio of $\nicefrac{1}{e} \approx 0.367$. The simplicity and the natural guarantee of Measured Continuous Greedy, combined with the fact that it was not improved for a few years, made some people suspect that it is optimal. However, Ene and Nguyen~\cite{ene2016constrained} shuttered this conjecture. They designed a $0.372$-approximation algorithm based on techniques used earlier to improved upon the $\nicefrac{1}{e}$-approximation in the special case of a cardinality constraint~\cite{buchbinder2014submodular}. The current best approximation is due to Buchbinder and Feldman~\cite{buchbinder2019constrained}, who improved the approximation ratio further to $0.385$. On the inapproximability side, Oveis Gharan and Vondr\'{a}k~\cite{gharan2011submodular} proved that no algorithm can achieve approximation better than $0.478$ even when $P$ is the matroid polytope of a partition matroid, and recently Qi~\cite{qi2022maximizing} showed that the same inapproximability applies also to a cardinality constraint. Closing the gap between the above algorithmic and inapproximability results, and understanding better this fundamental optimization problem, remains among the most important open problems in the area of submodular maximization.

We note that the restriction to down-closed polytopes in the above line of work is necessary, as otherwise the problem does not admit any constant factor approximation even for matroid base polytopes~\cite{vondrak2013symmetry}. However, fortunately, many natural constraints such as matroid constraints (or more generally, any set of packing constraints) are down-closed. Some works have also studied the possibility to bypass the inapproximability result of~\cite{vondrak2013symmetry} using an appropriate parametrization~\cite{du2022lyapunov,durr2020nonmonotone,mualem2023resolving}.

\subsection{Our Contribution}

Recent research has identified a useful class of continues functions, termed DR-submodular functions, that captures the diminishing return property of submodular set functions. A function $F\colon [0, 1]^\cN \to \bR$ is \emph{DR-submodular} if
$F(\vx + z \cdot \characteristic_u)-F(\vx)\geq F(\vy+z \cdot \characteristic_u)-F(\vy)$ for every element $u \in \cN$, two vectors $\vx \leq\vy\in[0, 1]^\cN$ and value $z \in [0, 1 - y_u]$, where $\characteristic_u$ is the standard basis vector having $1$ in the coordinate of $u$ and $0$ in all other coordinates. 
DR-submodular functions have many applications of their own (see, for example,~\cite{bian2017nonmonotone,mitra2021submodular,mualem2023resolving,niazadeh2020optimal}). However, it is also known that multilinear extensions of submodular set functions are DR-submodular~\cite{bian2017guaranteed}, and thus, optimization results for DR-submodular functions translate into optimization results for multilinear extensions.

Our main contribution is the design of a new algorithm for maximizing non-negative DR-submodular functions over a downed-closed convex body. This result is formally stated in Theorem~\ref{thm:main}. We refer the reader to Section~\ref{sec:preliminaries} for the definitions of the standard terms used in this theorem. However, we note that a convex body $P$ is meta-solvable if one can efficiently optimize linear functions over $P\cap Q$ for any polytope $Q$ defined by a polynomial number of explicit linear constraints.
\begin{theorem} \label{thm:main}
There exists an algorithm that given a non-negative $L$-smooth DR-submodular function $F$, a meta-solvable down-closed convex body $P \subseteq [0, 1]^\cN$ of diameter $D$ and a value $\delta \in (0, 1)$, returns a vector $\vx \in P$ such that $F(\vx) \geq 0.401 \cdot \max_{\vx\in P}F(\vx) - O(\delta D^2 L)$ and runs in $\Poly(|\cN|, \delta^{-1})$ time.
\end{theorem} 

As mentioned above, since multilinear extensions of submodular set functions are DR-sub\-modular, Theorem~\ref{thm:main} yields an improved approximation algorithm for such functions as well. Formally, we get the following theorem (see Appendix~\ref{app:technical_discrete} for a discussion of a few relevant technical issues in the derivation of Theorem~\ref{thm:discrete} from Theorem~\ref{thm:main}).
\begin{theorem} \label{thm:discrete}
There exists a polynomial time algorithm that, given a non-negative submodular set function $f \colon 2^\cN \to \nnR$ and a meta-solvable down-closed convex body $P \subseteq [0, 1]^\cN$, outputs a vector $\vx \in P$ such that
\[
	F(\vx)
	\geq
	0.401 \cdot \max_{\substack{S \subseteq \cN \\ \characteristic_S \in P}} f(S)
	\enspace,
\]
where $\characteristic_S$ is the characteristic vector of $S$, i.e., a vector that takes the value $1$ in the coordinates corresponding to the elements of $S$, and the value $0$ in the other coordinates.
\end{theorem}

Theorem~\ref{thm:discrete} can be used in a black box manner to improve various state-of-the-art submodular maximization results that are based on the rounding approach discussed above. However, here we only mention the following corollary, which follows by passing the output of the algorithm from Theorem~\ref{thm:discrete} to an appropriate rounding procedure (either Pipage Rounding~\cite{calinescu2011maximzing} or Swap Rounding~\cite{chekuri2010dependent} for a matroid constraint, and the rounding of Kulik et al.~\cite{kulik2013approximations} for a constant number of knapsack constraints).
\begin{corollary}
There exists a polynomial time $0.401$-approximation algorithm for the problem of maximizing a non-negative submodular set function $f \colon 2^\cN \to \nnR$ subject to either a matroid constraint or a constant number of knapsack constraints.
\end{corollary}

\paragraph{Remark.} The well-known equivalence between solvability and separability shows that any solvable convex body $P$ is almost meta-solvable in the following sense. If $P$ is solvable, then given a polytope $Q$ defined by a polynomial number of explicit linear constraints, one can optimize linear functions over $P \cap Q$ and get a solution that is almost optimal and close to being feasible. It is possible to use this observation to extend Theorems~\ref{thm:main} and~\ref{thm:discrete} to general solvable convex bodies at the cost of only guaranteeing that the output solution is close to being feasible. However, to avoid the need to handle the extra complications of almost feasible solutions, we restrict our attention to convex bodies that are meta-solvable.

\subsection{Techniques and Paper Structure}

Our main algorithm appears in Section~\ref{sec:algorithm} and makes use of several components. Two of these components are variants of known algorithms (discussed in Section~\ref{ssc:previous_alg}). However, the third component is novel, and we term it {\FWMCGFull} (or {\FWMCG} for short). As is hinted by its name, {\FWMCG} is a continues greedy algorithm. Intuitively, any continuous greedy algorithm maintains a solution $\vy(\tau)$ that starts as the all zeros solution at time $\tau = 0$, and monotonically increases as the time $\tau$ progresses, until the final solution of the algorithm is reached at time $\tau = 1$. The main distinction between various continues greedy algorithms are in the way they choose the direction in which $\vy(\tau)$ is increased at each time point. The analysis of these algorithms, then, always uses the fact that the optimal solution $\vo$ is a feasible direction to prove that the algorithm is able to find a direction leading to a substantial improvement in the value of $F(\vy(\tau))$. Typically, the improvement that can be guaranteed is initially large, but it decreases with $\tau$ since the already picked solution $\vy(\tau)$ inflicts an increasingly large ``damage'' on the value that can be gained by following the direction of $\vo$.

The analysis of the Measured Continuous Greedy algorithm of Feldman et al.~\cite{feldman2011unified} was based on a lemma (Lemma 2.2 of~\cite{feldman2011unified}) bounding the maximum possible damage at time $\tau$ using an expression based on $\|\vy(\tau)\|_\infty$. Since $\|\vy(\tau)\|_\infty$ provides very little information about $\vy(\tau)$, the bound given by the last lemma is necessarily weak. However, no better bound was known previously, and thus, \emph{essentially all} state-of-the-art works on non-monotone constrained submodular/DR-submodular maximization are based on this weak bound of~\cite{feldman2011unified}. In this work (in Section~\ref{sec:bounds}), we present a new bound on the damage. Our new bound is always at least as strong as the old bound of~\cite{feldman2011unified}, and in most cases it is strictly stronger. More specifically, our bound on the damage at time $\tau$ takes into accounts the directions used to increase $\vy(\tau)$ at every point prior to this time, and it shows that going in the direction of $\vo$ is good, unless $\vy(\tau)$ is always increased in directions that individually inflict a large damage on $\vo$. Theoretically, this means that we could guarantee a large output value if we could avoid increasing $\vy(\tau)$ in directions that inflict a large damage on $\vo$. Unfortunately, since we do not know $\vo$, we have to find an alternative way to take advantage of our new bound, which is described by the next paragraphs. Given the many uses of Lemma 2.2 of~\cite{feldman2011unified}, we believe that our stronger bound is likely to find additional applications in related problems, and to be a key component for further improvements.

Buchbinder and Feldman~\cite{buchbinder2019constrained} designed a version of Measured Continuous Greedy that gets a vector $\vz$, and produces a good output if $\vz$ inflicts a large damage on $\vo$ and also has a small intersection with $\vo$. Our novel component {\FWMCG} incorporates this feature, which means that if it ever increases $\vy(\tau)$ in a direction $\vx$ that inflicts a large damage on $\vo$ and also has a small intersection with $\vo$, then it is possible to get a good value by passing this direction $\vx$ as the vector $\vz$ to another copy of {\FWMCG}. However, this still does not allow us to handle the case in which the directions $\vx$ in which {\FWMCG} increases $\vy(\tau)$ always have a large intersection with $\vo$ while also inflicting a large damage on $\vo$. A first attempt to handle such directions is to use an algorithm for unconstrained submodular maximization to find a high value solution dominated by $\vx$. Since the intersection of $\vx$ with $\vo$ is large, we know that such a solution exist, and it is feasible by the down-closeness of $P$. However, the best possible algorithm for unconstrained submodular maximization cannot guarantee in general more than half of the value of this solution, which is not enough for our purposes. Thus, we need a way to guarantee that $\vx$ has extra properties that allow an unconstrained submodular maximization algorithm to find a solution dominated by $\vx$ with a better value.

All existing continuous greedy algorithms choose the direction in which $\vy(\tau)$ is increased by finding a vector $\vx \in P$ whose inner product with an appropriately chosen weight vector $\vw$ is large (Ene and Nguyen~\cite{ene2016constrained} uses a slightly different method, but Buchbinder and Feldman~\cite{buchbinder2019constrained} explain how \cite{ene2016constrained}'s method can be simulated with the usual one). The weakness of this method of choosing the direction $\vx$ is that it does not give the algorithm any control over the properties of $\vx$ (beyond its membership in $P$). Thus, our novel component {\FWMCG} uses a different method to choose directions. Specifically, it considers at every time $\tau$ the convex body $Q(\tau)$ obtained by restricting $P$ to vectors $\vx$ for which the inner product of $\vx$ with $\vw$ is large; specifically, at least as large as the lower bound we have on the value of the inner product of $\vo$ and $\vw$. The definition of $Q(\tau)$ guarantees that every vector $\vx \in Q(\tau)$ leads a significant increase in the value of $\vy(\tau)$ when used as the direction, and therefore, we are free to choose a vector from $Q(\tau)$ that has the properties that we need. In our case, it turns out that by using a Frank-Wolfe algorithm to choose $\vx$ as an (approximate) local maximum of $Q(\tau)$, we get properties allowing an unconstrained submodular maximization algorithm to produce a good value solution when $\vx$ both has a large intersection with $\vo$ and inflicts a lot of damage on $\vo$. Section~\ref{sec:main-component} presents and analyses two versions of {\FWMCG}: a continues time version that cannot be implemented, but is cleaner and easier to understand, and a discrete time version that is part of our final algorithm.

Our main algorithm (in Section~\ref{sec:algorithm}) is in charge of applying unconstrained submodular maximization to the directions $\vx$ used by {\FWMCG} as well as passing these directions to other copies of {\FWMCG} to handle the case that they inflict a large damage on $\vo$, but do not intersect it significantly. A technical difficulty is that a direction $\vx$ is considered to inflict a lot of damage on $\vo$ if it inflicts significantly more damage than the vector $\vz$ passed as input to {\FWMCG}. This means that if the main algorithm passes a direction $\vx$ that inflicts a lot of damage as the input vector $\vz$ to another copy of {\FWMCG}, it might get another direction $\vx'$ that inflicts even more damage, and thus, has to be passes to yet another copy of {\FWMCG}. Fortunately, the number of times this phenomenon can repeat is bounded because the damage inflicted to $\vo$ can never be larger than the entire value of $F(\vo)$. Thus, to handle this phenomenon, our main algorithm is designed as a recursive procedure, with the depth of the recursion bounded by the number of possible repetitions of the phenomenon.

\newtoggle{readingRecommendation}\togglefalse{readingRecommendation}
\iftoggle{readingRecommendation}{\subsection{Reading Recommendation}
For a short read showcasing our main new ideas, we recommend the familiar reader to start with the definitions in Section~\ref{sec:preliminaries} (excluding Subsection~\ref{ssc:previous_alg}), and then  read Sections~\ref{ssc:bounds_continuous} and~\ref{ssc:continuous_algorithm} that contain
our new bounds for DR-submodular functions and the unimplementable continuous time version of {\FWMCG} that uses these bounds, respectively. 
}{}

\section{Preliminaries} \label{sec:preliminaries}

This section introduces the notation and definitions that we use, and presents some previous results that we employ in this paper. We use $[n]$ to denote the set $\{1,2, \ldots, n\}$ and $[0, \tau]^i$ to denote the set of $i$-dimensional vectors whose individual coordinates all belong to the range $[0, \tau]$.

\paragraph{Vector operations:}
Throughout the paper, we use $\vzero$ and $\vone$ to represent the all zeros and all ones vectors, respectively.
We also employ various vector operators. First, given two vectors $\va, \vb \in [0,1]^\cN$, we use the following standard operators notation.
\begin{itemize}
	\item	The coordinate-wise maximum of $\va$ and $\vb$ is denoted by $\va \vee \vb$ (formally, $(\va \vee \vb)_u = \max\{a_u,b_u\}$ for every $u \in \cN$).
	\item The coordinate-wise minimum of $\va$ and $\vb$ is denoted by $\va \wedge \vb$ (formally, $(\va \wedge \vb)_u = \min\{a_u,b_u\}$ for every $u \in \cN$).
	\item The coordinate-wise product of $\va$ and $\vb$ (also known as the Hadamard product) is denoted by $\va \hprod \vb$ (formally, $(\va \hprod \vb)_u = a_u \cdot b_u$ for every $u \in \cN$).
	\item The inner product of $\va$ and $\vb$ is denoted by $\inner{\va}{\vb}$.
\end{itemize}

We also introduce a new operator $\psum$ for the coordinate-wise probabilistic sum of two vectors. Formally, $\va \psum \vb = \vone - (\vone -\va) \hprod (\vone-\vb)$. Observe that $\psum$ is a symmetric associative operator, and therefore, it makes sense to apply it also to sets of vectors. Formally, given vectors $\va^{(1)}, \va^{(2)}, \dotsc, \va^{(m)}$, we define
\[
	\PSum_{i=1}^{m}\va^{(i)}
	\triangleq
	\va^{(1)} \psum \va^{(2)} \psum \dotso \psum \va^{(m)}
	=
	\vone - \HProd_{i=1}^{m} (\vone - \va^{(i)})
	\enspace.
\]

To avoid using too many parentheses, we assume all the above vector operations have higher precedence compared to normal vector addition and subtraction. Additionally, whenever we have an inequality between vectors, this inequality should be understood to hold coordinate-wise. We occasionally also use vectors in the exponent. Such exponents should be understood as coordinate-wise operations. For example, given a vector $\vx \in [0, 1]^\cN$, $e^\vx$ is a vector such that $(e^\vx)_u = e^{x_u}$ for every $u \in \cN$.

%
%

\paragraph{The objective function:}

The focus of this work is on maximizing a \emph{DR-submodular} objective function.
%
As is common in the literature, we assume that the DR-submodular functions we consider are differentiable, and that one is able to evaluate both the function and its derivatives. For a differentiable function $F$, DR-submodularity is equivalent to the gradient $\nabla F$ being an antitone mapping, i.e., $\nabla F(\vx) \geq \nabla F(\vy)$ for every two vectors $\vx \leq \vy$ in $[0, 1]^\cN$. We also note that, for twice differentiable functions, DR-submodularity is equivalent to the Hessian being non-negative~\cite{bian2017guaranteed}. 

It is well known that every DR-submodular function is {\em continuous submodular},\footnote{The reverse is not always true~\cite{bian2017guaranteed}.} meaning that it satisfies
\[F(\vx)+F(\vy)\geq F(\vx \vee  \vy)+F(\vx \wedge \vy)\]
for every two vectors $\vx,\vy\in[0, 1]^\cN$.
The following lemma summarizes additional useful properties of DR-submodular functions. 

\begin{lemma} \label{lem:DR_properties}
Let $F\colon [0, 1]^\cN \to \nnR$ be a non-negative differentiable DR-submodular function. Then,
\begin{enumerate}
	\item $F(\lambda \cdot \vx+(1-\lambda) \cdot \vy) \geq \lambda \cdot F(\vx) +(1-\lambda) \cdot F(\vy)$ for every two vectors $\vx,\vy \in [0, 1]^\cN$ such that $\vx \leq \vy$ and value $\lambda\in[0,1]$. \label{prop:dr_bound1}
	\item $F(\vx+\lambda \cdot \vy) -F(\vx) \geq \lambda \cdot \left(F(\vx+\vy) -F(\vx)\right)$ for every two vectors $\vx \in [0, 1]^\cN$ and $\vy \geq \vzero$ such that $\vx+\vy \leq \vone$ and $\lambda\in[0,1]$. \label{prop:dr_bound11}
	\item $\inner{\nabla F(\vx)}{\vy} \geq F(\vx+\vy) - F(\vx)$ for every $\vx\in [0, 1]^\cN$ and $\vy\geq \vzero$ such that $\vx+\vy \leq \vone$. \label{prop:dr_bound2_up}
	\item $\inner{\nabla F(\vx)}{\vy} \leq F(\vx) - F(\vx-\vy)$ for every $\vx\in [0, 1]^\cN$ and $\vy\geq \vzero$ such that $\vx-\vy \geq \vzero$. \label{prop:dr_bound2_down}
	\item $F(\vx \vee \vy) + F(\vx \wedge \vy) \geq F(\vx \psum \vy) + F(\vx \hprod \vy)$ for every two vectors $\vx, \vy \in [0, 1]^\cN$. \label{prop:old_new_notation}
\end{enumerate}
\end{lemma}
\begin{proof}
The first two properties follow from the fact that DR-submodular functions are concave along non-negative directions (see Proposition~4 of~\cite{bian2017guaranteed}). The next two properties are immediate consequences of $\nabla F$ being an antitone mapping. To see why the last property holds, note that
\[
	F(\vx \vee \vy) + F(\vx \wedge \vy)
	\geq
	F(\vx \vee \vy + (\vx \wedge \vy - \vx \odot \vy))+ F(\vx \wedge \vy - (\vx \wedge \vy - \vx \odot \vy))]
	=
	F(\vx \psum \vy) + F(\vx \hprod \vy)
	\enspace,
\]
where the inequality holds by the DR-submodularity of $F$ since $\vx \odot \vy \leq \vx \wedge \vy \leq \vx \vee \vy$. 
\end{proof}

In addition, we use the closure properties stated by the following lemma.

\begin{lemma} \label{lem:DR_submodular}
Given a non-negative DR-submdoualr function $F\colon [0, 1]^\cN \to \nnR$ and vector $\vy \in [0, 1]^\cN$, the functions $G_{\psum}(\vx) \triangleq F(\vx \psum \vy)$ and $G_{\odot}(\vx) \triangleq F(\vx \odot \vy)$ are both non-negative DR-submodular functions.
\end{lemma}
\begin{proof}
The non-negativity of both functions follows immediately from the non-negativity of $F$.
Fix now two vectors $\vx^{(1)}, \vx^{(2)} \in [0, 1]^\cN$ such that $\vx^{(1)} \leq \vx^{(2)}$. Then, for every element $u \in \cN$ and value $p \in [0, 1 - x^{(2)}_u]$,
\begin{align*}
    G_{\psum}(\vx^{(1)} + p \cdot \characteristic_u&) - G_{\psum}(\vx^{(1)})
    =
    F((\vx^{(1)} + p \cdot \characteristic_u) \psum \vy) - F(\vx^{(1)} \psum \vy)\\
    ={} &
    F(\vx^{(1)} \psum \vy + p(1 - y_u) \cdot \characteristic_u) - F(\vx^{(1)} \psum \vy)\\
    \geq{} &
    F(\vx^{(2)} \psum \vy + p(1 - y_u) \cdot \characteristic_u) - F(\vx^{(2)} \psum \vy)\\
    ={} &
    F((\vx^{(2)} + p \cdot \characteristic_u) \psum \vy) - F(\vx^{(2)} \psum \vy)
    =
    G_{\psum}(\vx^{(2)} + p \cdot \characteristic_u) - G_{\psum}(\vx^{(2)})
    \enspace,
\end{align*}
where the inequality follows from the DR-submodularity of $F$ since $\vx^{(2)} \psum \vy + p(1 - y_u) \cdot \characteristic_u \leq \vone$. Similarly,
\begin{align*}
    G_{\odot}(\vx^{(1)} + p \cdot \characteristic_u&) - G_{\odot}(\vx^{(1)})
    =
    F((\vx^{(1)} + p \cdot \characteristic_u) \odot \vy) - F(\vx^{(1)} \odot \vy)\\
    ={} &
    F(\vx^{(1)} \odot \vy + p \cdot y_u \cdot \characteristic_u) - F(\vx^{(1)} \odot \vy)\\
    \geq{} &
    F(\vx^{(2)} \odot \vy + p \cdot y_u \cdot \characteristic_u) - F(\vx^{(2)} \odot \vy)\\
    ={} &
    F((\vx^{(2)} + p \cdot \characteristic_u) \odot \vy) - F(\vx^{(2)} \odot \vy)
    =
    G_{\odot}(\vx^{(2)} + p \cdot \characteristic_u) - G_{\odot}(\vx^{(2)})
    \enspace,
\end{align*}
where the inequality follows from the DR-submodularity of $F$ since $\vx^{(2)} \odot \vy + p \cdot y_u \cdot \characteristic_u \leq \vy \leq \vone$.
\end{proof}

We say that a differentiable function $F \colon [0, 1]^\cN \to \bR$ is \emph{$L$-smooth} for a given value $L \geq 0$ if $\|\nabla F(\vx) - \nabla F(\vy)\|_2 \leq L \cdot \|\vx - \vy\|_2$ for every two vectors $\vx, \vy \in [0, 1]^\cN$.

\paragraph{Problem definition:}
The problem we study is maximizing a non-negative differentiable DR-submodular function $F\colon [0, 1]^\cN \to \nnR$ subject to a feasible region $P \subseteq [0, 1]^\cN$ which is a meta-solvable down-closed convex body. We denote by $\vo\in P$ an (arbitrary) optimal solution maximizing $F$, and we assume that $F(\vo)$ is strictly positive since otherwise the problem is trivial (any feasible solution is optimal). We also say that the diameter of $P$ is $D$ if $\|\vx - \vy\|_2 \leq D$ for every two vectors $\vx, \vy \in P$.

\subsection{Previous Algorithmic Components} \label{ssc:previous_alg}

This section presents (small variants of) known algorithms that are used as procedures by our new algorithms.

\paragraph{Frank-Wolfe Variant.}
A vector $\vx$ is a \emph{local maximum with respect to a vector $\vy$} if it holds that $\inner{\vy - \vx}{\nabla F(\vx)} \leq 0$. While this condition does not suffice to guarantee that $F(\vx)$ is large compared to $F(\vy)$ when $F$ is a DR-submodular function, it does imply the next lemma.
\begin{lemma} \label{lem:optimal_local_maximum}
    Given a DR-submodular function $F \colon [0, 1]^\cN \to \bR$, if $\vx$ is a local optimum with respect to vector $\vy$, then
    \[
			F(\vx)
			\geq
			\frac{1}{2}\cdot[F(\vx \vee \vy)+ F(\vx \wedge \vy)]
			\enspace.
		\]
\end{lemma}
\begin{proof}
By Properties~\ref{prop:dr_bound2_up} and~\ref{prop:dr_bound2_down} of Lemma~\ref{lem:DR_properties},
\begin{align*}
	0
	\geq{} &
	\inner{\vy - \vx}{\nabla F(\vx)}
	=
	\inner{(\vy \vee \vx) - \vx}{\nabla F(\vx)} - \inner{\vx - (\vy \wedge \vx)}{\nabla F(\vx)}\\
	\geq{} &
	[F(\vy \vee \vx) - F(\vx)] - [F(\vx) - F(\vy \wedge \vx)]
	=
	F(\vy \vee \vx) + F(\vy \wedge \vx) - 2 \cdot F(\vx)
	\enspace.
\end{align*}
The lemma now follows by rearranging this inequality.
\end{proof}

To use the guarantee of the last lemma with respect to the vector $\vy$ that is optimal in some convex body $P$, it is often necessary to find an (approximate) local maximum within this convex body. Formally, a vector $\vx$ is a \emph{local maximum with respect to $P$} if for all $\vy\in P$, $(\vy-\vx)\cdot \nabla F(\vx)\leq 0$. Chekuri~\cite{chekuri2014submodular} described a local search algorithm that finds an approximate local maximum when $F$ is the multilinear extension of a submodular set function. Extending their algorithm to general DR-submodular functions is not trivial as they implicitly use a bound on the smoothness of the multilinear extension. Motivated by the work of~\cite{lacostejulien16convergance}, Bian~\cite{bian2017nonmonotone} suggested using a version of Frank-Wolfe\footnote{It should be noted that local search and Frank-Wolfe refer basically to the same algorithm in this context. However, the first name is usually used by the submodular set function optimization community, while the latter name is usually used by the continuous submodular optimization community.} for that purpose. Unfortunately, Bian~\cite{bian2017nonmonotone}'s algorithm assumes knowledge of the curvature of $F$ with respect to $P$, which can be costly to determine. Therefore, we use here a somewhat different variant of Frank-Wolfe that has the guarantee stated in the next theorem. We note that the algorithm whose existence is guaranteed by the Theorem~\ref{thm:local_search} does not assume knowledge of $L$, and the proof of this theorem can be found in Appendix~\ref{app:local_search}.

\begin{restatable}{theorem}{thmLocalSearch} \label{thm:local_search}
There exists a polynomial time algorithm that given a non-negative $L$-smooth function $F \colon 2^\cN \to \nnR$, a solvable convex body $P$ of diameter $D$ and value $\delta \in (0, 1)$,\iftoggle{appendix}{}{\footnote{Throughout the paper, we use both $\eps$ and $\delta$ to denote error control parameters. However, we reserve $\eps$ for error parameters that are considered to be constant, and $\delta$ for error parameters that can be polynomially small.}} outputs a vector $\vx \in P$ such that $\inner{\vy - \vx}{\nabla F(\vx)} \leq \delta \cdot [\max_{\vy' \in P} F(\vy') + D^2 L / 2]$ for every vector $\vy \in P$. If $F$ is DR-submodular, this implies 
\[
F(\vx)
	\geq
	\frac{1}{2}\cdot \max_{\vy \in P} [F(\vx \vee \vy)+ F(\vx \wedge \vy)]
	-
	\frac{\delta \cdot [2 \cdot \max_{\vy' \in P} F(\vy') + D^2 L]}{4} \enspace.
\]
\end{restatable}

\paragraph{Unbalanced Unconstrained Maximization.}

Given a non-negative submodular function $f$ over a ground set $\cN$, the unconstrained submodular maximization problem asks for a subset of $\cN$ maximizing $f$ among all subsets of $\cN$. Buchbiner et al.~\cite{buchbinder2015tight} designed an algorithm termed \textbf{Double-Greedy} obtaining $\nicefrac{1}{2}$-approximation for this problem, matching an inapproximability result due to Feige et al.~\cite{feige11maximizing}.
Later, it was shown by Mualem and Feldman~\cite{mualem2022using} (based on observations made by~\cite{buchbinder2014} and the pre-print version~\cite{qi2022maximizingArxiv} of~\cite{qi2022maximizing}) that the value of the output set of \textbf{Double-Greedy} is in fact lower bounded by
\begin{equation} \label{eq:general}
	\max_{r \geq 0} \left[\frac{2r}{(r + 1)^2} \cdot f(OPT) + \frac{1}{(r + 1)^2} \cdot f(\varnothing) + \frac{r^2}{(r + 1)^2} \cdot f(\cN)\right]
	\enspace,
\end{equation}
where $OPT$ is the subset of $\cN$ maximizing $f$. By plugging $r=1$, this expression is always at least $\nicefrac{1}{2} \cdot f(OPT)$, which recovers the original guarantee of~\cite{buchbinder2015tight}. However, in many cases, improved results can be obtained by using this more general expression. 

The \textbf{Double-Greedy} algorithm was generalized to DR-submodular objective functions by multiple works~\cite{bian2019bian,chen2019unconstrained,niazadeh2020optimal}. However, these works only proved $\nicefrac{1}{2}$-approximation, and did not generalize the lower bound stated in Equation~\eqref{eq:general} (because they temporally preceded the introduction of this lower bound). Since we require in this paper a lower bound for DR-submodular functions similar to Equation~\eqref{eq:general}, we prove Theorem~\ref{thm:unconstrained} in Appendix~\ref{app:unconstrained} by describing and analyzing another generalization of \textbf{Double-Greedy} incorporating ideas from all the above mentioned previous works. In addition to having a stronger approximation guarantee, Theorem~\ref{thm:unconstrained} improves over the previous works of~\cite{chen2019unconstrained,niazadeh2020optimal} in the sense that it requires as few as $O(|\cN| \log (|\cN|\eps^{-1}))$ function evaluations and does not need to assume Lipschitz continuity.\footnote{Niazadeh et al.~\cite{niazadeh2020optimal} require the same number of function evaluations, but assume Lipschitz continuity; while Chen et al.~\cite{chen2019unconstrained} avoid assuming Lipschitz continuity, but require more function evaluations. We do not compare with the algorithm of Bian et al.~\cite{bian2019bian} since an implementation is provided for it only in the context of particular applications.}
\begin{restatable}{theorem}{thmUnconstrained} \label{thm:unconstrained}
There exists a polynomial time algorithm that given a non-negative DR-submodular function $F \colon \allowbreak [0, 1]^\cN \to \nnR$ and parameter $\eps \in (0, 1)$, accesses the function $F$ in $O(n \log (n\eps^{-1}))$ locations, and then outputs a vector $\vx \in [0, 1]^\cN$ such that
\[
	F(\vx)
	\geq
	\max_{r \geq 0} \left[\left(\frac{2r}{(r + 1)^2} - O(\eps)\right) \cdot F(\vo) + \frac{1}{(r + 1)^2} \cdot F(\vzero) + \frac{r^2}{(r + 1)^2} \cdot F(\vone)\right]
	\enspace,
\]
where $n = |\cN|$ and $\vo$ is any fixed vector in $[0, 1]^\cN$.
\end{restatable}

\begin{corollary} \label{cor:unconstrained}
There exists a polynomial time algorithm, termed {\BoxMaximization}, that given a non-negative DR-submodular function $F \colon \allowbreak [0, 1]^\cN \to \nnR$, a vector $\vx \in [0, 1]^\cN$ and a parameter $\eps \in (0, 1)$, outputs a vector $\vy \leq \vx$ of value at least 
\begin{align} \label{eq:USM}
	\max_{r \geq 0} \left[\left(\frac{2r}{(r + 1)^2} - O(\eps)\right) \cdot F(\vx \hprod \vo) + \frac{1}{(r + 1)^2} \cdot F(\vzero) + \frac{r^2}{(r + 1)^2} \cdot F(\vx)\right]
	\enspace,
\end{align}
where $\vo$ is any fixed vector in $[0, 1]^\cN$. 
\end{corollary}
\begin{proof}
For the (fixed) vector $\vx$, define $G(\va) \triangleq F(\vx \hprod \va)$ for every vector $\va$.
By Lemma~\ref{lem:DR_submodular}, $G$ is a non-negative DR-submodular function.
%
Therefore, one can apply the algorithm from Theorem~\ref{thm:unconstrained} to $G$, and get a vector $\vy'$ satisfying
\begin{align*}
	F(\vx \hprod \vy') ={}& G(\vy')
	\geq{} 
	\max_{r \geq 0} \left[\left(\frac{2r}{(r + 1)^2} - O(\eps)\right) \cdot G(\vo) + \frac{1}{(r + 1)^2} \cdot G(\vzero) + \frac{r^2}{(r + 1)^2} \cdot G(\vone)\right]\\
	={} &
	\max_{r \geq 0} \left[\left(\frac{2r}{(r + 1)^2} - O(\eps)\right) \cdot F(\vx \hprod \vo) + \frac{1}{(r + 1)^2} \cdot F(\vzero) + \frac{r^2}{(r + 1)^2} \cdot F(\vx)\right]
	\enspace.
\end{align*}
Hence, the vector $\vx \hprod \vy' \leq \vx$ has the property guaranteed by the corollary. 
\end{proof}

%

\section{Our Algorithm} \label{sec:algorithm}

In this section, we describe and analyze the algorithm used to prove our main theorem (Theorem~\ref{thm:main}).
Recall that the main component used by this algorithm is our novel procedure named {\FWMCGFull} (or {\FWMCG} for short), which is described in Section~\ref{sec:main-component}. The guarantee of {\FWMCG} is given by the next theorem.

\begin{restatable}{theorem}{thmMainComponentDiscrete} \label{thm:main-component-discrete}
{\FWMCG} is an algorithm that obtains as input a non-negative $L$-smooth DR-submodular function $F \colon [0, 1]^\cN \to \nnR$, a meta-solvable down-closed convex body $P \subseteq [0, 1]^\cN$ of diameter $D$, a vector $\vz \in P$ and parameters $t_s \in (0,1)$, $\eps\in (0, 1/2)$ and $\delta \in (0, 1)$. Given this input, {\FWMCG} outputs a vector $\vy\in P$ and vectors $\vx(1), \ldots, \vx(m)\in P$, for some $m = O(\delta^{-1} + \eps^{-1})$, such that at least one of the following is always true.
 \begin{itemize}
     \item The vector $\vy$ obeys
\begin{align*}
    F(\vy) \geq{} & e^{-1} \cdot \left[((2 - t_s)e^{t_s} - 1 - O(\eps)) \cdot F(\vo) - (e^{t_s} - 1) \cdot F(\vz \hprod \vo)\vphantom{\frac{1}{2}}\right. \\&\mspace{200mu}- \left.\left((2 - t_s)e^{t_s} + t_s - \frac{t_s^2 + 5}{2}\right) \cdot F(\vz \psum \vo) \right] - O(\delta D^2L)
    \enspace.
\end{align*}
\item There exists $i \in [m]$ such that $F(\vx(i) \psum \vo) \leq F(\vz \psum \vo) - \eps \cdot F(\vo)$ and the vector $\vx(i)$ satisifies $F(\vx(i)) \geq \frac{1}{2} \cdot [F(\vx(i) \vee \vo) + F(\vx(i) \wedge \vo)] - O(\eps) \cdot F(\vo) - O(\delta D^2L)$ (i.e., it is an approximate local maximum with respect to $\vo$).
    \end{itemize}
Moreover, if $t_s$ and $\eps$ are considered to be constants, then the time complexity of the above algorithm is $\Poly(|\cN|, \delta^{-1})$.
\end{restatable}

If the first option happens (i.e., $F(\vy)$ is large) in a given execution of {\FWMCG}, we call this execution \emph{successful}. Otherwise, we say that the execution is \emph{unsuccessful}. Using this terminology, we can now describe the main algorithm, given as Algorithm~\ref{alg:main-rec}, that we use to prove Theorem~\ref{thm:main}. In addition to the parameters listed in Theorem~\ref{thm:main}, Algorithm~\ref{alg:main-rec} gets two additional parameteres: $\eps \in (0, 1/2)$ and $t_s \in (0, 1)$. Both parameters are simply passed forward to the executions of {\FWMCG} used, and they are assigned constant values to be determined later.

Algorithm~\ref{alg:main-rec} uses $1+\lceil 2/\eps\rceil$ levels of recursion (note that the parameter $i$ tracks the level of recursion of the current recursive call). The single recursive call on the first level of the recursion gets some approximate local maximum $\vz^{(0)}$ of $P$ as the parameter $\vz$, and then executes on this $\vz$ two algorithms: {\BoxMaximization} and {\FWMCG}, producing vectors $\vz'$, $\vy$ and $\vx(1), \vx(2), \dotsc, \vx(m)$. Each recursive call on the second level of the recursion executes the same two algorithms on a parameter $\vz$ set to be one of the vectors $\vx(\cdot)$ produced by the first level of the recursion, yielding a new set of vectors $\vz'$, $\vy$ and $\vx(1), \vx(2), \dotsc, \vx(m)$. More generally, each recursive call on level $i > 1$ of the recursion executes {\BoxMaximization} and {\FWMCG} on a parameter $\vz$ set to be one of the vectors $\vx(\cdot)$ produced by the $i - 1$ level of the recursion, yielding yet another set of vectors $\vz'$, $\vy$ and $\vx(1), \vx(2), \dotsc, \vx(m)$. The final output of Algorithm~\ref{alg:main-rec} is the best solution among all vectors $\vz'$ and $\vy$ produced by any level of the recursion.

\SetKwProg{Function}{Function}{:}{}
\begin{algorithm}
\DontPrintSemicolon
\caption{\texttt{Main Algorithm}$(F, P, t_s, \eps, \delta)$} \label{alg:main-rec}
Let $\vz^{(0)}$ be an (approximate) local maximum in $P$ obtained via the algorithm from Theorem~\ref{thm:local_search} by setting the error parameter $\delta$ of the last algorithm to $\min\{\eps, \delta\}$.\\
Execute \textbf{Main-Recursive}$(F, P, t_s, \vz^{(0)}, \eps, \delta, 1$).

\nonl
\noindent\rule{0.95\textwidth}{0.3pt}
\BlankLine

\Function{\textbf{\upshape{Main-Recursive}}$(F, P, t_s, \vz, \eps, \delta, i)$}{
$\vz'= \BoxMaximization(\vz)$ (Algorithm from Corollary~\ref{cor:unconstrained}).\label{line:unconstrained}\\
Let $(\vy, \vx(1), \ldots, \vx(m))= \FWMCG(F,P, \vz, t_s, \eps)$ (Algorithm from Theorem~\ref{thm:main-component-discrete}).\\
\If{$i<1+\lceil 2/\eps\rceil$}{
	\lFor{$j = 1$ \KwTo $m$}{Recursively execute \textbf{Main-Recursive}$(F, P, \vx(j), t_s, \eps, \delta, i+1)$, and let $\vy(j)$ be its output.}
}
\Return the vector maximizing $F$ among $\vz'$, $\vy$ and the vectors in $\{\vy(j) \mid j \in [m]\}$.
}
\end{algorithm}

We say that a recursive call of Algorithm~\ref{alg:main-rec} is \emph{successful} if its execution of {\FWMCG} is successful. The analysis of Algorithm~\ref{alg:main-rec}, which appears in Section~\ref{ssc:main_analysis}, shows that Algorithm~\ref{alg:main-rec} makes enough recursive calls to guarantee that at least one of them is successful and also gets a vector $\vz$ that is an approximate local maximum with respect to $\vo$. The analysis then shows that either the vector $\vz'$ or the vector $\vy$ produced by such a recursive call satisfies the guarantee of Theorem~\ref{thm:main}.

\subsection{Proof of Theorem~\ref{thm:main}} \label{ssc:main_analysis}

Note that the number of recursive calls executed by Algorithm~\ref{alg:main-rec} is $O\left(m^{1/\eps}\right) = (2 + \delta^{-1} + \eps^{-1})^{O(1/\eps)}$. For any constant $\eps$, this expression is polynomial in $\delta^{-1}$. Since every individual recursive call runs in time polynomial in $\delta^{-1}$ and $|\cN|$ (when excluding the time required for the recursive calls it invokes), we get that the full time complexity of Algorithm~\ref{alg:main-rec} is polynomial in these two values, as is guaranteed by Theorem~\ref{thm:main}.


From now on, we concentrate on proving that the value of the output solution of Algorithm~\ref{alg:main-rec} is large enough. We say that a recursive call $C$ of Algorithm~\ref{alg:main-rec} is an \emph{heir} if
\begin{itemize}
	\item it is the single recursive call on level $1$ of the recursion, or
	\item $C$ was invoked by another heir recursive call $C_p$ that is unsuccessful, and the input vector $\vz$ of $C$ is a vector $\vx(j)$ produced by $C_p$ that satisfies the second bullet in Theorem~\ref{thm:main-component-discrete} (such a vector must exist since $C_p$ is unsuccessful).
\end{itemize}
\begin{observation} \label{obs:heir}
An heir recursive call always gets a vector $\vz$ obeying
\[
	F(\vz)
	\geq
	\frac{1}{2} \cdot [F(\vz \vee \vo) + F(\vz \wedge \vo)] - O(\eps) \cdot F(\vo) - O(\delta D^2L)
	\enspace.
\]
\end{observation}
\begin{proof}
If the recursive call is not on the first level of the recursion, then it gets a vector $\vz$ that is equal to $\vx(j)$ for some vector $\vx(j)$ obeying the second bullet in Theorem~\ref{thm:main-component-discrete}, and therefore, this vector $\vz$ obeys the guarantee of the observation. Otherwise, $\vz$ is the output of executing the algorithm from Theorem~\ref{thm:local_search} on $P$ with the error parameter $\min\{\eps, \delta\}$, and thus, obeys
\begin{align*}
	F(\vz)
	\geq{} &
	\frac{1}{2}\cdot \max_{\vy \in P} [F(\vz \vee \vy)+ F(\vz \wedge \vy)]
	-
	\frac{\min\{\eps, \delta\} \cdot [2 \cdot \max_{\vy' \in P} F(\vy') + D^2 L]}{4}\\
	\geq{} &
	\frac{1}{2}\cdot [F(\vz \vee \vo)+ F(\vz \wedge \vo)]
	-
	\min\{\eps, \delta\} \cdot \frac{2 \cdot F(\vo) + D^2 L}{4}
	\enspace.
	\qedhere
\end{align*}
\end{proof}

Corollary~\ref{cor:successful} below shows that there is always some recursive call of Algorithm~\ref{alg:main-rec} which is a successful heir. To prove this corollary, we first need the following observation.
\begin{observation} \label{obs:if_all_fail}
If all heir recursive calls of Algorithm~\ref{alg:main-rec} are unsuccessful, then every level $i$ of the recursion must include an heir recursive call getting a vector $\vz$ such that $F(\vz \psum \vo) \leq F(\vz^{(0)} \psum \vo) - \eps(i - 1) \cdot F(\vo)$.
\end{observation}
\begin{proof}
The proof of the observation is done by induction. For $i = 1$ the observation is trivial; thus, let us prove the observation for level $i \geq 2$ assuming it holds for level $i - 1$. By the induction hypothesis, some heir recursive call on level $i - 1$ of the recursion gets a vector $\vz$ such that $F(\vz \psum \vo) \leq F(\vz^{(0)} \psum \vo) - \eps(i - 2) \cdot F(\vo)$. Since this heir recursive call is unsuccessful (as we assume all heir recursive calls are), there exists $j \in [m]$ such that the vector $\vx(j)$ produced by this recursive call satisfies the second bullet in Theorem~\ref{thm:main-component-discrete}, and in particular, obeys $F(\vx(j) \psum \vo) \leq F(\vz \psum \vo) - \eps \cdot F(\vo) \leq F(\vz^{(0)} \psum \vo) - \eps(i - 1) \cdot F(\vo)$. The vector $\vx(j)$ then becomes the vector $\vz$ of some recursive call on level $i$ of the recursion, and this recursive call is an heir by definition.
\end{proof}
\begin{corollary} \label{cor:successful}
Some recursive call of Algorithm~\ref{alg:main-rec} is a successful heir.
\end{corollary}
\begin{proof}
Assume the corollary is false. Then, by Observation~\ref{obs:if_all_fail}, level $1 + \lceil 2/\eps \rceil$ of the recursion includes an heir recursive call getting a vector $\vz$ such that $F(\vz \psum \vo) \leq F(\vz^{(0)} \psum \vo) - 2 \cdot F(\vo)$. Since this execution is also unsuccessful, it produces a vector $\vx(j)$ such that $F(\vx(j) \psum \vo) \leq F(\vz \psum \vo) - \eps \cdot F(\vo) \leq F(\vz^{(0)} \psum \vo) - (2 + \eps) \cdot F(\vo)$ for some $j \in [m]$. By the non-negativity of $F$ and the strict positivity of $F(\vo)$, the last inequality implies
\[
	F(\vo)
	<
	F(\vz^{(0)} \psum \vo) - F(\vo)
	\leq
	F(\vz^{(0)} \hprod (\vone - \vo)) - F(\vzero)
	\leq
	F(\vz^{(0)} \hprod (\vone - \vo))
	\enspace,
\]
where the second inequality holds by the DR-submodularity of $F$. However, this contradicts the definition of $\vo$ as the down-closeness of $P$ implies that $\vz^{(0)} \hprod (\vone - \vo) \in P$.
\end{proof}

Consider a successful heir recursive call of Algorithm~\ref{alg:main-rec}, and let us denote by $\vz^*$, $\vy^*$ and $\vz'^{*}$ the input vector $\vz$ of this call and the vectors $\vy$ and $\vz'$ produced during it, respectively. Then, by the guarantee of Theorem~\ref{thm:main-component-discrete},
\begin{align*}
    F(\vy^*) \geq{} & e^{-1} \cdot \left[((2 - t_s)e^{t_s} - 1 - O(\eps)) \cdot F(\vo) - (e^{t_s} - 1) \cdot F(\vz \hprod \vo)\vphantom{\frac{1}{2}}\right. \\&\mspace{200mu}- \left.\left((2 - t_s)e^{t_s} - \frac{t_s^2-2t_s + 5}{2}\right) \cdot F(\vz \psum \vo) \right] - O(\delta D^2L)
    \enspace,
\end{align*}
We also need to lower bound $F(\vz'^{*})$, which is obtained by the next lemma.

\begin{lemma}
It is guaranteed that
\[
	F(\vz'^{*})
	\geq
	\max_{r\geq 0}\left[\frac{4r + r^2}{2(r+1)^2} \cdot F(\vz^{*} \hprod \vo) + \frac{r^2}{2(r+1)^2} \cdot F(\vz^{*} \psum \vo)\right] - O(\eps)\cdot F(\vo) - O(\delta D^2L)
	\enspace.
\]
\end{lemma}
\begin{proof}
Since $\vz'^*$ is the output of the algorithm from Corollary~\ref{cor:unconstrained},
\begin{align*}
	F(\vz'^{*}) \geq{} &
		\max_{r \geq 0} \left[\left(\frac{2r}{(r + 1)^2} - O(\eps)\right) \cdot F(\vz^{*} \hprod \vo) + \frac{1}{(r + 1)^2} \cdot F(\vzero) + \frac{r^2}{(r + 1)^2} \cdot F(\vz^{*})\right]\\
  \geq{} & 
	\max_{r \geq 0} \left[\frac{2r}{(r + 1)^2} \cdot F(\vz^{*} \hprod \vo) + \frac{r^2}{(r + 1)^2} \cdot F(\vz^{*})\right] - O(\eps)\cdot F(\vo)
	\enspace,
\end{align*}
where the second inequality holds by (i) the non-negativity of $F$ and (ii) the observation that the down-closeness of $P$ implies that $\vz^{*} \hprod \vo \in P$, which yields $F(\vz^{*} \hprod \vo) \leq F(\vo)$ by the definition of $\vo$.
Note also that
\[
	F(\vz^{*})
	\geq
	\frac{1}{2}\cdot[F(\vz^{*} \vee \vo)+ F(\vz^{*} \wedge \vo)] - O(\delta D^2L)\\
	\geq
	\frac{1}{2}\cdot[F(\vz^{*} \psum \vo) + F(\vz^{*} \hprod \vo)] - O(\delta D^2L)
	\enspace,
\]
where the first inequality holds by Observation~\ref{obs:heir} since $\vz^*$ is the input for an heir recursive call, and the second inequality holds by Property~\ref{prop:old_new_notation} of Lemma~\ref{lem:DR_properties}. The lemma now follows by plugging this inequality into the previous one. 
\end{proof}

We have so far proved lower bounds on $F(\vy^*)$ and $F(\vz'^{*})$. Both these lower bounds immediately apply also to the value of the output of Algorithm~\ref{alg:main-rec}, and thus, any convex combination of these lower bounds is also a lower bound on the value of this output. Hence, for every $\alpha \in [0, 1]$ and $r \geq 0$, the value of the output of Algorithm~\ref{alg:main-rec} is lower bounded by
\begin{align*}
  &
		\left[\frac{(1 - \alpha)[(2 - t_s)e^{t_s} - 1]}{e} - O(\eps)\right] \cdot F(\vo) + \left[\frac{\alpha(4r + r^2)}{2(r + 1)^2} - \frac{(1 - \alpha)(e^{t_s} - 1)}{e}\right] \cdot F(\vz^{*} \hprod \vo)\\
    &\mspace{130mu}+
    \left[\frac{\alpha r^2}{2(r + 1)^2} - \frac{(1 - \alpha)[(2 - t_s)e^{t_s} + t_s - (t_s^2 + 5)/2]}{e}\right] \cdot F(\vz^{*} \psum \vo) - O(\delta D^2 L)
    \enspace,
\end{align*}
At this point, we choose $\alpha = 0.1974$ and $r = 2.22$. Additionally, let us set the parameter $t_s$ of Algorithm~\ref{alg:main-rec} to be $0.3682$. One can verify that this choice of values makes the coefficients of $F(\vz^{*} \hprod \vo)$ and $F(\vz^{*} \psum \vo)$ in the last expression non-negative, and thus, dropping the terms involving these two values can only decrease the total value of the expression. Hence, the value of the output of Algorithm~\ref{alg:main-rec} is at least
\begin{align*}
	\left[\frac{(1 - \alpha)[(2 - t_s)e^{t_s} - 1]}{e} - O(\eps)\right] \cdot F(\vo) - O(\delta D^2 L)
	\geq{} &
	[0.40101 - O(\eps)] \cdot F(\vo) - O(\delta D^2 L)\\
	\geq{} &
	0.401 \cdot F(\vo) - O(\delta D^2 L)
	\enspace,
\end{align*}
where the first inequality holds by plugging in the values we have chosen for $\alpha$ and $t_s$, and the second inequality holds when we set $\eps$ to be a small enough constant. 

\section{New Bounds for DR-Submodular Functions} \label{sec:bounds}

In this section, we prove our novel bounds for DR-submodular functions. We present two sets of such bounds. The bounds of the first set, given in Section~\ref{ssc:bounds_continuous}, have a more stylized form, and are useful in the analysis of the continuous time version of {\FWMCG}. In principle, these bounds can also be used to analyze the discrete time version of {\FWMCG} by setting their continuous function parameters to be step functions. However, we chose instead to give, in Section~\ref{ssc:bounds_discrete}, a second set of bounds that are more natural in the context of discrete time algorithms.

\subsection{Bounds Designed for Continuous Time Algorithms} \label{ssc:bounds_continuous}

The basic bound that we prove in this section is given by the next lemma.
\begin{lemma} \label{lem:mathematical}
Given a non-negative DR-submodular function $F \colon [0, 1]^\cN \allowbreak \to \nnR$, a value $t \geq 0$, an integrable function $\vx \colon [0, t] \to [0, 1]^\cN$ and a vector $\va\in [0, 1]^\cN$,
\[
    F\left( \vone - \va \odot e^{-\int_0^t \vx(\tau) d\tau} \right)
    \geq
    e^{-t} \cdot \left[F\left(\vone-\va\right) + \sum_{i = 1}^{\infty} \frac{1}{i!} \cdot \int_{\tau \in [0, t]^i} F\left((\vone - \va) \psum \PSum_{j=1}^{i} \vx(\tau_j)\right)d\tau \right]
    \enspace.
\]
\end{lemma}
\begin{proof}
We prove by induction that the lemma holds when the upper limit of the sum on the right hand side of the inequality of the lemma is replaced with any non-negative integer $k$. This implies the lemma by the monotone convergence theorem since every term of the infinite sum in the original right hand side is non-negative. For convenience, we use below the notation $\vy(t) \triangleq \int_0^t \vx(\tau) d\tau$. 

We begin with the base of the induction, i.e., the case of $k=0$. In this case,
\begin{align*}
    F\left( \vone - \va \odot e^{-\vy(t)} \right)
    & = F\left( e^{-t} \cdot (\vone - \va ) + (1-e^{-t})\cdot \left(\vone - \va \hprod \frac{e^{-\vy(t)}- e^{-t} \cdot \vone}{1-e^{-t}} \right)\right)\\
    & \geq e^{-t} \cdot F(\vone-\va) + (1 - e^{-t}) \cdot F\left(\vone - \va \hprod \frac{e^{-\vy(t)}- e^{-t} \cdot \vone}{1-e^{-t}} \right)
    \geq
    e^{-t} \cdot F\left(\vone - \va\right)
		\enspace,
\end{align*}
where the second inequality hold by the non-negativity of $F$, and the first inequality follows from Property~\ref{prop:dr_bound1} of Lemma~\ref{lem:DR_properties} because the fact that $\vzero \leq \vy(t)\leq t \cdot \vone$ implies $\vzero \leq \frac{e^{-\vy(t)}- e^{-t} \cdot \vone}{1-e^{-t}} \leq \vone$, and therefore, $\vzero \leq \vone - \va \leq \vone - \va \hprod \frac{e^{-\vy(t)}- e^{-t}}{1-e^{-t}}  \leq \vone$.

We now get to the induction step. In other words, we need to prove the bound for some integer $k\geq 1$ assuming it holds for $k - 1$. By the chain rule and the definition of $\vy(t)$,
\begin{align*}
		&
		\frac{dF\left(\vone - \va \odot e^{-\vy(t)}\right)}{dt}
    =
		\inner{\nabla F\left(\vone - \va \odot e^{-\vy(t)}\right)}{\va \odot \vx(t) \odot e^{-\vy(t)}}\\
    \geq{} &
    F\left(\vone - \va \odot e^{-\vy(t)}+ \va \odot \vx(t) \odot e^{-\vy(t)}\right) - F\left(\vone - \va \odot e^{-\vy(t)}\right) \\
    ={}& F\left(\vone - \va \odot (\vone-\vx(t) )\odot e^{-\vy(t)}\right) - F\left(\vone - \va \odot e^{-\vy(t)}\right) \\  
    \geq{} &
    e^{-t} \cdot \left[F\left((\vone - \va) \psum \vx(t)\right) + \sum_{i = 1}^{k - 1} \frac{1}{i!} \cdot \int_{\tau \in [0, t]^i} F\left((\vone - \va) \psum \vx(t) \psum \PSum_{j=1}^{i} \vx(\tau_j)\right) d\tau \right] \\
    & \mspace{520mu}- F\left(\vone - \va \odot e^{-\vy(t)}\right)
    \enspace,
\end{align*}
where the first inequality holds by Property~\ref{prop:dr_bound2_up} of Lemma~\ref{lem:DR_properties}, and the second inequality follows by applying the induction hypothesis to the function $F$ with $\va'= \va\odot(\vone-\vx(t)) = \vone - (\vone - \va) \psum \vx(t) \in [0, 1]^\cN$. Let us now define, for convenience, $g(t)\triangleq F\left(\vone - \va \odot e^{-\vy(t)}\right)$. Then, the previous inequality becomes
\[
    \frac{dg(t)}{dt}
    \geq
    e^{-t} \cdot \left[F\left((\vone - \va) \psum \vx(t)\right) + \sum_{i = 1}^{k - 1} \frac{1}{i!} \cdot \int_{\tau \in [0, t]^i} F\left((\vone - \va) \psum \vx(t) \psum \PSum_{j=1}^{i} \vx(\tau_j)\right) d\tau \right] - g(t)
		\enspace.
\]
Solving this differential inequality with the boundary condition $g(0)=F(\vone - \va)$, we get the promised inequality
\[F\left(\vone - \va \odot e^{-\vy(t)}\right) = g(t) \geq  e^{-t} \cdot \left[F(\vone - \va) + \sum_{i = 1}^{k} \frac{1}{i!} \cdot \int_{\tau \in [0, t]^i} F\left((\vone - \va) \psum \PSum_{j=1}^{i}\vx(\tau_j)\right) d\tau \right] \enspace.\]
To see that this is indeed the case, notice that, if we define $h(t)$ to be the right hand side of the last inequality, then, by the Leibniz integral rule,
\begin{align*}
\frac{dh(t)}{dt} & = - e^{-t} \cdot \left[F(\vone - \va) + \sum_{i = 1}^{k} \frac{1}{i!} \cdot \int_{\tau \in [0, t]^i} F\left((\vone - \va) \psum \PSum_{j=1}^{i}\vx(\tau_j)\right) d\tau \right]\\
& +  e^{-t} \cdot \left[F((\vone - \va) \psum \vx(t)) + \sum_{i = 2}^{k} \frac{1}{(i - 1)!} \cdot \int_{\tau \in [0, t]^{i - 1}} F\left((\vone - \va) \psum \vx(t) \psum \PSum_{j=1}^{i - 1}\vx(\tau_j)\right) d\tau \right]\\
={} &
 - h(t) + e^{-t} \cdot \left[F\left((\vone - \va) \psum \vx(t)\right) + \sum_{i = 1}^{k - 1} \frac{1}{i!} \cdot \int_{\tau \in [0, t]^i} F\left((\vone - \va) \psum \vx(t) \psum \PSum_{j=1}^{i} \vx(\tau_j)\right) d\tau \right]
\enspace.
\qedhere
\end{align*}
\end{proof}

Using the basic bound proved in the last lemma, it is possible to prove the more general bound given by the next lemma. Notice that Lemma~\ref{lem:mathematical} is the special case of Lemma~\ref{lem:mathematical_general} obtained by setting $h = 1$.

\begin{lemma} \label{lem:mathematical_general}
Consider a non-negative DR-submodular function $F \colon\allowbreak [0, 1]^\cN \to \nnR$, a value $t \geq 0$ and integer $h \geq 1$. Then, given an integrable function $\vx^{(i)} \colon [0, t] \to [0, 1]^\cN$ and two vectors $\va^{(i)}, \vb^{(i)} \in [0, 1]^\cN$ for every $i \in [h]$, such that $\sum_{i=1}^h \vb^{(i)} = \vone$, its holds that
\begin{align*}
    &   
    F\left( \vone - \sum_{i = 1}^h \vb^{(i)} \hprod \va^{(i)} \hprod e^{-\int_0^t \vx^{(i)}(\tau) d\tau}\right)\\
    \geq{} &
    e^{-t} \cdot \left[F\left(\vone- \sum_{i = 1}^h \vb^{(i)} \hprod \va^{(i)} \right) + \sum_{i = 1}^{\infty} \frac{1}{i!} \cdot \int_{\tau \in [0, t]^i} F\left(\sum_{i = 1}^h \vb^{(i)} \hprod \left((\vone - \va) \psum \PSum_{j=1}^{i}\vx^{(i)}(\tau_j)\right)\right)d\tau \right]
    \enspace.
\end{align*}
\end{lemma}
\begin{proof}
Let us define a new DR-submodular function $G$ over an extended ground set. Formally, let $\cN_1, \cN_2, \dotsc, \cN_h$ be $h$ disjoint copies of $\cN$, and let us define a function $G \colon [0, 1]^{\dotbigcup_{i = 1}^h \cN_i} \to \nnR$. It is notationally convenient to think of $G$ as a function getting $h$ arguments $\vc^{(1)}, \vc^{(2)}, \dotsc, \vc^{(h)} \in [0, 1]^{\cN}$. Using this convention, we define $G$ as follows.
\[
	G(\vc^{(1)}, \vc^{(2)}, \dotsc, \vc^{(h)})
	=
	F\left(\sum_{i = 1}^h \vb^{(i)} \hprod \vc^{(i)}\right)
	\enspace.
\]

One can verify that the guarantee of the current lemma follows by applying Lemma~\ref{lem:mathematical} to $G$. However, this can be done only if $G$ has the properties assumed by Lemma~\ref{lem:mathematical}, namely, DR-submodularity and non-negativity. Therefore, the rest of this proof is devoted to showing that $G$ indeed has these properties. The non-negativity of $G$ follows immediately from the non-negativity of $F$. To see that $G$ is DR-submodular, note that, for any $i \in [h]$ and element $u \in \cN$, we have
\[
	\frac{\partial G(\vc^{(1)}, \vc^{(2)}, \dotsc, \vc^{(h)})}{\partial c^{(i)}_u}
	=
	b_u^{(i)} \cdot \left.\frac{\partial F(\vx)}{\partial x_u} \right|_{\vx = \sum_{i = 1}^h \vb^{(i)} \hprod \vc^{(i)}}
	\enspace,
\]
Thus, if we have vectors $\vc^{(i)}, \bar{\vc}^{(i)} \in [0, 1]^{\cN}$ that obey $\vc^{(i)} \leq \bar{\vc}^{(i)}$ for every $i \in [h]$, then, by the DR-submodularity of $F$,
\begin{align*}
	\frac{\partial G(\vc^{(1)}, \vc^{(2)}, \dotsc, \vc^{(h)})}{\partial c^{(i)}_u}
	={} &
	b_u^{(i)} \cdot \left.\frac{\partial F(\vx)}{\partial x_u} \right|_{\vx = \sum_{i = 1}^h \vb^{(i)} \hprod \vc^{(i)}}\\
	\geq{} &
	b_u^{(i)} \cdot \left.\frac{\partial F(\vx)}{\partial x_u} \right|_{\vx = \sum_{i = 1}^h \vb^{(i)} \hprod \bar{\vc}^{(i)}}
	=
	\frac{\partial G(\bar{\vc}^{(1)}, \bar{\vc}^{(2)}, \dotsc, \bar{\vc}^{(h)})}{\partial \bar{c}^{(i)}_u}
	\enspace.
	\qedhere
\end{align*}
\end{proof}

\subsection{Bounds Designed for Discrete Time Algorithms} \label{ssc:bounds_discrete}

In this section, we prove bounds for DR-submodular functions that are analogous to the ones given in Section~\ref{ssc:bounds_continuous}, but are designed to be useful in the analysis of discrete time algorithms. The first bound that we give in this section is given by the next lemma. As mentioned above, this bound can be viewed as an analog of Lemma~\ref{lem:mathematical}. 
Interestingly, it can also be viewed as an extension to DR-submodular functions of Lemma~2.2 of~\cite{feige11maximizing}, and the proof we give for the bound resembles the proof of the last lemma.

\begin{lemma} \label{lem:basic_bound_discrete}
Given a DR-submodular function $F: [0, 1]^\cN \to \nnR$, integer value $r \geq 1$, vectors $\vx(1), \vx(2), \dotsc, \vx(r) \in [0, 1]^\cN$, and values $p_1, p_2, \dotsc, p_r \in [0, 1]$,
\[
    F\left(\PSum_{i = 1}^r (p_i \cdot \vx(i)) \right)
    \geq
    \sum_{S \subseteq [r]} \left[\prod_{i \in S} p_i \cdot \mspace{-9mu} \prod_{i \in [r] \setminus S} \mspace{-9mu} (1 - p_i) \cdot F\left(\PSum_{i \in S} \vx(i)\right)\right]
    \enspace.
\]
\end{lemma}
\begin{proof}
We prove the lemma by induction on $r$. For $r = 1$ the lemma reduces to
\[
    F(p_1 \cdot \vx(1))
    \geq
    p_1 \cdot F(\vx^{(1)}) + (1 - p_1) \cdot F(\vzero)
    \enspace,
\]
which holds by Property~\ref{prop:dr_bound1} of Lemma~\ref{lem:DR_properties}. Assume now that the lemma holds for $r - 1$, and let us prove it for $r \geq 2$. Define $G_1(\vx) = F(\vx \psum \PSum_{i = 1}^{r - 1} (p_i \cdot \vx(i)))$ and $G_2(\vx) = F(\vx \psum \vx(r))$. By Lemma~\ref{lem:DR_submodular}, $G_1$ and $G_2$ are both non-negative DR-submodular functions. Hence,
\begin{align*}
    F\left(\PSum_{i = 1}^r (p_i \cdot \vx(i)) \right)
    ={} &
    G_1(p_r \cdot \vx(r))
    \geq
    p_r \cdot G_1(\vx(r)) + (1 - p_r) \cdot G_1(\vzero)\\
    ={} &
    p_r \cdot G_2\left(\PSum_{i = 1}^{r - 1} (p_i \cdot \vx(i))\right) + (1 - p_r) \cdot F\left(\PSum_{i = 1}^{r - 1} (p_i \cdot \vx(i))\right)\\
    \geq{} &
    p_r \cdot \sum_{S \subseteq [r - 1]} \left[\prod_{i \in S} p_i \cdot \mspace{-18mu} \prod_{i \in [r - 1] \setminus S} \mspace{-18mu} (1 - p_i) \cdot G_2\left(\PSum_{i \in S} \vx(i)\right) \right] \\&\mspace{100mu}+ (1 - p_r) \sum_{S \subseteq [r - 1]} \left[\prod_{i \in S} p_i \cdot \mspace{-18mu} \prod_{i \in [r - 1] \setminus S} \mspace{-18mu} (1 - p_i) \cdot F\left(\PSum_{i \in S} \vx(i)\right) \right]\\
    ={} &
    \sum_{S \subseteq [r]} \left[\prod_{i \in S} p_i \cdot \mspace{-9mu} \prod_{i \in [r] \setminus S} \mspace{-9mu} (1 - p_i) \cdot F\left(\PSum_{i \in S} \vx(i)\right) \right]
    \enspace,
\end{align*}
where the first inequality holds by Property~\ref{prop:dr_bound1} of Lemma~\ref{lem:DR_properties} since $G_1$ is DR-submodular, and the second inequality follows by applying the induction hypothesis to $G_2$ and $F$.
\end{proof}

Like in Section~\ref{ssc:bounds_continuous}, we would like to generalize the basic bound given by the last lemma. The generalized bound is given by the next corollary. 
\begin{corollary} \label{cor:general_bound_discrete}
Consider a non-negative DR-submodular function $F \colon\allowbreak [0, 1]^\cN \to \nnR$ and integer values $r, h \geq 1$. Then, given vectors $\vx^{(i)}(1), \vx^{(i)}(2), \dotsc, \vx^{(i)}(r), \vb^{(i)} \in [0, 1]^\cN$ for every $i \in [h]$, such that $\sum_{i=1}^h \vb^{(i)} = \vone$, and value $p_j \in [0, 1]$ for every $j \in [r]$, it holds that
\begin{align*}
    \mspace{100mu}&\mspace{-100mu}
    F\left(\sum_{i = 1}^h \left(\vb^{(i)} \hprod \PSum_{j = 1}^h (p_j \cdot \vx^{(i)}(j)) \right) \right)\\
    \geq{} &
    \sum_{S \subseteq [r]} \left[\left(\prod_{j \in S} p_j \cdot \mspace{-9mu} \prod_{j \in [r] \setminus S} \mspace{-9mu} (1 - p_j)\right) \cdot F\left(\sum_{i = 1}^h \left(\vb^{(i)} \hprod \PSum_{j \in S} \vx^{(i)}(j)\right)\right)\right]
    \enspace.
\end{align*}
\end{corollary}
\begin{proof}
Let us define a new DR-submodular function $G$ over an extended ground set. Formally, let $\cN_1, \cN_2, \dotsc, \cN_h$ be $h$ disjoint copies of $\cN$, and let us define a function $G \colon [0, 1]^{\dotbigcup_{i = 1}^h \cN_i} \to \nnR$. It is notationally convenient to think of $G$ as a function getting $h$ arguments $\vc^{(1)}, \vc^{(2)}, \dotsc, \vc^{(h)} \in [0, 1]^{\cN}$. Using this convention, we define $G$ as follows.
\[
	G(\vc^{(1)}, \vc^{(2)}, \dotsc, \vc^{(h)})
	=
	F\left(\sum_{i = 1}^h \vb^{(i)} \hprod \vc^{(i)}\right)
	\enspace.
\]

By repeating the arguments from the proof of Lemma~\ref{lem:mathematical_general}, we get that $G$ is a DR-submodular function. Thus, we can apply Lemma~\ref{lem:basic_bound_discrete} to $G$ with $\vx(i) = (\vx^{(1)}(i), \dotsc, \vx^{(h)}(i)) \in [0, 1]^{\dotbigcup_{i = 1}^h \cN_i}$ for every $i \in [r]$, which proves the corollary.
\end{proof}

\section{Our Main Algorithmic Component ({\FWMCG})}\label{sec:main-component}

In this section, we describe and analyze {\FWMCG}, and thereby prove Theorem~\ref{thm:main-component-discrete}. The formal description of the algorithm and its analysis appear in Section~\ref{ssc:component_discrete}. However, the exact description of the algorithm and the complete proof of Theorem~\ref{thm:main-component-discrete} involve many technical details. Hence, to convey better the main ideas of the algorithm and the proof, we find it useful to start by presenting a simplified unimplementable continuous time version of {\FWMCG}. This algorithm, along with a corresponding variant of Theorem~\ref{thm:main-component-discrete}, is presented in Section \ref{ssc:continuous_algorithm}.

\subsection{Continuous Time Version} \label{ssc:continuous_algorithm}

In this section, we present a simplified unimplementable continuous time version of {\FWMCG}. To avoid less important details, we make the following assumptions in the description and analysis of this version.

\begin{itemize}
    \item It is possible to efficiently find exact local maxima of convex bodies (rather than approximate ones as guaranteed by Theorem~\ref{thm:local_search}).
    \item The values of $F(\vo)$, $F(\vz \hprod \vo)$ and $F(\vz \psum \vo)$ are known to the algorithm.
		\item All integrals appearing in either the algorithm or its analysis are well-defined.
\end{itemize}

 Using our simplified continuous time version of {\FWMCG}, we prove the following continuous time version of Theorem~\ref{thm:main-component-discrete}. 

\begin{theorem} \label{thm:main-component}
There exists an (unimplementable) continuous time algorithm that, given a non-negative DR-submodular function $F \colon [0, 1]^\cN \to \nnR$, a meta-solvable down-closed convex body $P \subseteq [0, 1]^\cN$, a vector $\vz \in P$ and constant parameters $t_s\in(0,1)$ and $\eps \in (0, 1/2)$, outputs a vector $\vy\in P$ and a function $\vx(\tau) \colon [0, 1) \to P$ such that at least one of the following is always true.
 \begin{itemize}
     \item The vector $\vy$ obeys
\begin{align*}
    F(\vy) \geq{} & e^{-1} \cdot \left[((2 - t_s)e^{t_s} - 1 - O(\eps)) \cdot F(\vo) - (e^{t_s} - 1) \cdot F(\vz \hprod \vo)\vphantom{\frac{1}{2}}\right. \\&\mspace{200mu}- \left.\left((2 - t_s)e^{t_s} - \frac{t_s^2-2t_s + 5}{2}\right) \cdot F(\vz \psum \vo) \right]
    \enspace.
\end{align*}
\item There exists $\tau \in [t_s, 1)$ such that $F(\vx(\tau) \psum \vo) \leq F(\vz \psum \vo) - \eps \cdot F(\vo)$ and the vector $\vx(\tau)$ is a local maximum with respect to $\vo$.
\end{itemize}
\end{theorem}


The algorithm that we use to prove Theorem~\ref{thm:main-component} appears as Algorithm~\ref{alg:aided_mcgreedy_continuous}. This algorithm grows continuously a solution $\vy(\tau)$. The direction of the growth at any time $\tau\in [0,1]$ is determined by a vector $\vx(\tau)$ chosen as a local maximum inside some convex body $Q(\tau) \subseteq P$. The definition of $Q(\tau)$ guarantees that every vector $\vx(\tau) \in Q(\tau)$ leads to an increase in the value of $\vy(\tau)$ at a rate of at least $V(\tau) - F(\vy(t))$.\footnote{A technical issues is that the convex body $Q(\tau)$ might end up being empty. We show below that whenever this happens, there is a value $\tau' \in [t_s, \tau)$ for which $\vx(\tau')$ satisfy the second case in the guarantee of Theorem~\ref{thm:main-component}. However, to guarantee that the function $\vx$ is defined for its entire domain, Algorithm~\ref{alg:aided_mcgreedy_continuous} still selects $\vx(\tau)$ as an arbitrary vector in $P$ in this case.}

\begin{algorithm}
\caption{{\FWMCGFull}\textsf{ -- Continuous Time Version}($F, P, \vz, t_s, \eps$)} \label{alg:aided_mcgreedy_continuous}
\DontPrintSemicolon
Let $V(\tau) \triangleq \left\{\begin{array}{ll}F(\vo) - F(\vz \hprod \vo) - (1- e^{-\tau})\cdot F(\vz \psum \vo)  & \tau\in [0,t_s)\\
e^{-\tau}\left[(e^{t_s} - 2\eps)\cdot F(\vo) -\left(e^{t_s}-1-(\tau-t_s)\right) \cdot F(\vz \psum \vo)\right] & \tau \in [t_s,1) \enspace.
\end{array}\right.$ \\
Let $\vz(\tau) \triangleq \left\{\begin{array}{ll}\vz  & \tau\in [0,t_s) \\
\vzero & \tau \in [t_s,1) \enspace.
\end{array}\right.$\\

\BlankLine

Let $\vy(0) \gets \vzero$.\\
\ForEach{$\tau \in [0, 1)$}
{
  Let $\vw(\tau) \gets (\vone - \vy(\tau)-\vz(\tau)) \hprod \nabla F(\vy(\tau))$.\label{line:weight}\\
  Let $Q(\tau) \triangleq  \{\vx \in \cP \mid  \inner{\vw(\tau)}{\vx} \geq V(\tau) - F(\vy(\tau))\}$.\\ 
  Find a local maximum of $\vx(\tau) \in Q(\tau)$ (if $Q(\tau)=\varnothing$, set $\vx(\tau)$ to be an arbitrary vector in $P$).\\
	Increase $\vy(\tau)$ at a rate of $\frac{d\vy(\tau)}{d\tau} = (\vone - \vy(\tau)-\vz(\tau)) \odot \vx(\tau)$.\label{line:rate}
}
\Return the vector $\vy(1)$ and the function $\vx(\cdot)$.
\end{algorithm}

We begin the analysis of Algorithm~\ref{alg:aided_mcgreedy_continuous} by deriving a formula for $\vy(\tau)$, and then using this formula to show that $\vy(1) \in P$.
\begin{lemma} \label{lem:y_value}
For every $\tau \in [0, 1]$,
\[
	\vy(\tau)
	=\left\{
	\begin{array}{ll}
		\vone - \vz \psum e^{-\int_0^\tau \vx(s) ds} & \tau\in [0,t_s]\\
		\vone - \left(\vz \psum e^{-\int_0^{t_s} \vx(s) ds}\right) \hprod e^{-\int_{t_s}^\tau \vx(s) ds} & \tau \in [t_s, 1] \enspace.
	\end{array}\right.
\]
\end{lemma}
\begin{proof}
In the range $[0, t_s)$, the vector $\vy(\tau)$ grows at a rate of $(\vone - \vy(\tau)-\vz(\tau)) \odot \vx(\tau) = (\vone - \vy(\tau)-\vz) \odot \vx(\tau)$. One can verify that, together with the boundary condition $\vy(0) = \vzero$, the solution for this differential equation is
\[
	\vy(\tau)
	=
	(\vone - \vz) \hprod \left(1 - e^{-\int_0^\tau \vx(s) ds}\right)
	=
	\vone - \vz \psum e^{-\int_0^\tau \vx(s) ds}
	\enspace.
\]

In the range $[t_s, 1)$, the vector $\vy(\tau)$ grows at a rate of $(\vone - \vy(\tau)-\vz(\tau)) \odot \vx(\tau) = (\vone - \vy(\tau)) \odot \vx(\tau)$. From the previous part of the proof, we get that the boundary condition for this range is $\vy(t_s) = \vone - \vz \psum e^{-\int_0^{t_s} \vx(s) ds}$, and therefore, the solution for the last differential equation is
\[
	\vone - \left(\vz \psum e^{-\int_0^{t_s} \vx(s) ds}\right) \hprod e^{-\int_{t_s}^\tau \vx(s) ds}
	\enspace.
	\qedhere
\]





\end{proof}
\begin{corollary} \label{cor:belonging}
The vector $\vy(1)$ belongs to $P$.
\end{corollary}
\begin{proof}
Observe that
\begin{align*}
	\vy(1)
	={} &
	\vone - \left(\vz \psum e^{-\int_0^{t_s} \vx(s) ds}\right) \hprod e^{-\int_{t_s}^1 \vx(s) ds}\\
	\leq{} &
	\vone - \left(\vone - \int_0^{t_s} \vx(s) ds\right) \hprod \left(\vone -\int_{t_s}^1 \vx(s) ds\right)
	\leq
	\int_0^{t_s} \vx(s) ds + \int_{t_s}^1 \vx(s) ds
	=
	\int_0^1 \vx(s) ds
	\enspace.
\end{align*}
Since $P$ is convex, the rightmost hand side $\int_0^1 \vx(s) ds$ of the last inequality belongs to $P$, and thus, $\vy(1)$ is upper bounded by a point in $P$. By the down-closeness of $P$, this implies $\vy(1) \in P$ because the fact that $\vz$, $e^{-\int_0^{t_s} \vx(s) ds}$ and $e^{-\int_{t_s}^1 \vx(s) ds}$ are all vectors in $[0, 1]^\cN$ implies $\vy(1) = \vone - \left(\vz \psum e^{-\int_0^{t_s} \vx(s) ds}\right) \hprod e^{-\int_{t_s}^1 \vx(s) ds} \geq \vzero$.
\end{proof}

In the rest of the analysis of Algorithm~\ref{alg:aided_mcgreedy_continuous}, we need to consider two cases. The first case is when $Q(\tau)$ is non-empty for every $\tau \in [0, 1)$.
The following lemma completes the proof of Theorem~\ref{thm:main-component} in this first case.

\begin{lemma} \label{cor:y_values}
If $Q(\tau)$ is non-empty for every $\tau \in [0, 1)$, then
\begin{align*}
	F(\vy(1)) \geq{} & e^{-1} \cdot \left[((2 - t_s)e^{t_s} - 1 - 2\eps) \cdot F(\vo) - (e^{t_s} - 1) \cdot F(\vz \hprod \vo)\vphantom{\frac{1}{2}}\right. \\&\mspace{200mu}- \left.\left((2 - t_s)e^{t_s} + t_s - \frac{t_s^2 + 5}{2}\right) \cdot F(\vz \psum \vo) \right]
  \enspace.
\end{align*}
and therefore, Theorem~\ref{thm:main-component} holds in this case.
\end{lemma}
\begin{proof}
By the chain rule,
\begin{align} \label{eq:differential_equation}
	\frac{dF(\vy(\tau))}{d\tau}
	={} &
	\inner{\frac{d \vy(\tau)}{d\tau}}{\nabla F(\vy(\tau))}\\\nonumber
	={} &
	\inner{(\vone - \vy(\tau)-\vz(\tau)) \hprod \vx(\tau)}{\nabla F(\vy(\tau))}
	=
	\inner{\vw(\tau)}{\vx(\tau)}
	\geq
	V(\tau) - F(\vy(\tau))
	\enspace,
\end{align}
where the last inequality holds since we assume that $Q(\tau)$ is non-empty for every $\tau \in [0, 1)$, and therefore, $\vx(\tau)$ is a vector in $Q(\tau)$. This implies a differential inequality. Solving it for $\tau \in [0, t_s]$ yields
\[
    F(\vy(\tau)) \geq F(\vo) - F(\vz \hprod \vo) - F(\vz \psum \vo) + e^{-\tau}[\tau \cdot F(\vz \psum \vo) + c]
\]
for some constant $c$. Due to the limit condition $F(\vy(0)) \geq 0$, which hold by the non-negativity of $F$, we can set $c = F(\vz \hprod \vo) + F(\vz \psum \vo) - F(\vo)$ in the last inequality, yielding
\[
    F(\vy(\tau)) \geq (1 - e^{-\tau}) \cdot F(\vo) - (1 - e^{-\tau}) \cdot F(\vz \hprod \vo) - (1 - (1 + \tau)e^{-\tau}) \cdot F(\vz \psum \vo) \enspace.
\]
Recall that the last inequality holds for $\tau \in [0, t_s]$. In particular, for $\tau = t_s$, we get
\begin{equation} \label{eq:t_s_value}
    F(\vy(t_s)) \geq 
    (1 - e^{-t_s}) \cdot F(\vo) - (1 - e^{-t_s}) \cdot F(\vz \hprod \vo) - (1 - (1 + t_s)e^{-t_s}) \cdot F(\vz \psum \vo)
    \enspace.
\end{equation}
For $\tau \in [t_s, 1]$, the solution of the differential inequality \eqref{eq:differential_equation} is
\[
    F(\vy(\tau)) \geq e^{-\tau}[\tau ((e^{t_s} - 2\eps) \cdot F(\vo) - (e^{t_s} - 1 + t_s) \cdot F(\vz \psum \vo)) + (\tau^2/2) \cdot F(\vz \psum \vo) + c] \enspace,
\]
where $c$ is again some constant. This time we use Inequality~\eqref{eq:t_s_value} as a limit condition. To determine the value of $c$ that is justified by this limit condition, we need to solve the following equation.
\begin{align*}
    &
    e^{-t_s}[t_s ((e^{t_s} - 2\eps) \cdot F(\vo) - (e^{t_s} - 1 + t_s) \cdot F(\vz \psum \vo)) + (t_s^2/2) \cdot F(\vz \psum \vo) + c]\\
    ={} &
    (1 - e^{-t_s}) \cdot F(\vo) - (1 - e^{-t_s}) \cdot F(\vz \hprod \vo) - (1 - (1 + t_s)e^{-t_s}) \cdot F(\vz \psum \vo)
    \enspace.
\end{align*}
One can verify that the solution of this equation is
\[
    c = ((1 - t_s)e^{t_s} + 2\eps t_s - 1) \cdot F(\vo) - (e^{t_s} - 1) \cdot F(\vz \hprod \vo) - ((1 - t_s)e^{t_s} - 1 - t_s^2/2) \cdot F(\vz \psum \vo)
    \enspace.
\]
Plugging this value of $c$ into the above stated solution of the differential inequality for $\tau \in [t_s, 1]$, we get that, for this range of $\tau$ values,
\begin{align*}
    F(\vy(\tau)) \geq{}& e^{-\tau}[((1 + \tau - t_s) e^{t_s} - 2\eps(\tau - t_s) - 1) \cdot F(\vo) - (e^{t_s} - 1) \cdot F(\vz \hprod \vo) \\&- ((1 - t_s)e^{t_s} - 1 - (t_s^2 + \tau^2)/2 + \tau(e^{t_s} - 1 + t_s)) \cdot F(\vz \psum \vz)]\\
		\geq{} &
		e^{-\tau}[((1 + \tau - t_s) e^{t_s} - 1 - 2\eps) \cdot F(\vo) - (e^{t_s} - 1) \cdot F(\vz \hprod \vo) \\&- ((1 - t_s)e^{t_s} - 1 - (t_s^2 + \tau^2)/2 + \tau(e^{t_s} - 1 + t_s)) \cdot F(\vz \psum \vo)]
    \enspace.
\end{align*}
In particular, for $\tau = 1$,
\begin{align*}
    F(\vy(1)) \geq{} & e^{-1} \cdot \left[((2 - t_s)e^{t_s} - 1 - 2\eps) \cdot F(\vo) - (e^{t_s} - 1) \cdot F(\vz \hprod \vo)\vphantom{\frac{1}{2}}\right. \\&\mspace{200mu}- \left.\left((2 - t_s)e^{t_s} + t_s - \frac{t_s^2 + 5}{2}\right) \cdot F(\vz \psum \vo) \right]
    \enspace. \qedhere
\end{align*}
\end{proof}

From this point on, we consider the case when $Q(\tau)$ is empty for some $\tau \in [0, 1)$. Note that, when $Q(\tau)$ is empty, it (in particular) does not include $\vo$. Therefore, it is useful to lower bound $\inner{\vw(\tau)}{\vo}$. The following lemma lower bounds this expression with the values of $F$ at appropriately chosen points.
\begin{lemma} \label{lem:product_lower_bound}
For every $\tau \in [0, 1)$, $\inner{\vw(\tau)}{\vo} \geq F(\vy(\tau) \psum \vo - \vz(\tau) \hprod \vo) - F(\vy(\tau))$.
\end{lemma}
\begin{proof}
Property~\ref{prop:dr_bound2_up} of Lemma~\ref{lem:DR_properties} implies that
\begin{align*}
	\inner{\vw(\tau)}{\vo}
	={} &
	\inner{(\vone - \vy(\tau) - \vz(\tau)) \hprod \nabla F(\vy(\tau))}{\vo}
	=
	\inner{\nabla F(\vy(\tau))}{(\vone - \vy(\tau) - \vz(\tau)) \hprod \vo}\\
	\geq{} &
	F(\vy(\tau) + (\vone - \vy(\tau) - \vz(\tau)) \hprod \vo) - F(\vy(\tau))
	=
	F(\vy(\tau) \psum \vo - \vz(\tau) \hprod \vo) - F(\vy(\tau))
	\enspace.
	\qedhere
\end{align*}
\end{proof}

Given the last lemma, our next goal is to lower bound $F(\vy(\tau) \psum \vo - \vz(\tau) \hprod \vo)$. To achieve this goal, we first need to prove the following technical lemma.
\begin{lemma} \label{lem:technical_less_ts}
For every vector $\va \in [0, 1]^\cN$ obeying $\va \leq \vone - \vo$ and vector $\vo' \in [0, 1]^\cN$ obeying $\vo' \leq \vo$,
\[
	F(\vo' + (\vone - \vz) \hprod \va)
	\geq
	F(\vo) + F(\vzero) - F(\vz \psum \vo) - F(\vo - \vo')
	\geq
	F(\vo) - F(\vz \psum \vo) - F(\vo - \vo')
	\enspace.
\]
\end{lemma}
\begin{proof}
By the DR-submodularity of $F$,
\begin{align*}
	F(\vo' + (\vone - \vz) \hprod \va)
	\geq{} &
	F(\vo + (\vone - \vz) \hprod \va) + F(\vzero) - F(\vo - \vo')\\
	\geq{} &
	F(\vo) + F(\vz \psum (\vo + \va)) - F(\vz \psum \vo) + F(\vzero) - F(\vo - \vo')\\
	\geq{} &
	F(\vo) + F(\vzero) - F(\vz \psum \vo) - F(\vo - \vo')
	\enspace,
\end{align*}
where the last inequality follows from the non-negativity of $F$. This proves the first inequality of the lemma. The second inequality of the lemma also follows from the non-negativity of $F$.
\end{proof}

We are now ready to lower bound $F(\vy(\tau) \psum \vo - \vz(\tau) \hprod \vo)$ in the range $[0, t_s)$. In Corollary~\ref{cor:low_Q_non-empty}, we use this lower bound to show that $Q(\tau)$ is never empty when $\tau \in [0, t_s)$.
\begin{lemma} \label{lem:first_o_guarantee}
For every $\tau \in [0, t_s)$,
\begin{align*}
	F(\vy(\tau) \psum \vo - \vz(\tau) \hprod \vo)
	={} &
	F(\vy(\tau) \psum \vo - \vz \hprod \vo)\\
	\geq{} &
	F(\vo) - F(\vz \hprod \vo) - (1 - e^{-\tau}) \cdot F(\vz \psum \vo)
	=
	V(\tau)
	\enspace.
\end{align*}
\end{lemma}
\begin{proof}
By Lemma~\ref{lem:y_value},
\begin{align*}
	F(\vy(\tau) \psum \vo - {}&\vz \hprod \vo)
	=
	F\left(\left(\vone - \vz \psum e^{-\int_0^\tau \vx(s) ds}\right) \psum \vo - \vz \hprod \vo\right)\\
	={} &
	F\left(\vone - \vz \hprod \vone \hprod e^{-\int_0^\tau \vzero ds} - (\vone - \vz) \hprod (\vone - \vo) \hprod e^{-\int_0^\tau \vx(s) ds}\right)\\
	\geq{} &
	e^{-\tau} \cdot \left[F(1 - \vz \hprod \vone - (\vone - \vz) \hprod (\vone - \vo)) \vphantom{\left(\PSum_{j = 1}^{i}\right)}\right.\\&+\left. \sum_{i = 1}^{\infty} \frac{1}{i!} \cdot \int_{\vs \in [0, \tau]^i} F\left(\vz \hprod \left((\vone - \vone) \psum \PSum_{j = 1}^i \vzero\right) + (\vone - \vz) \hprod \left(\vo \psum \PSum_{j = 1}^i \vx(s)\right)\right) d\vs \right]\\
	={} &
	e^{-\tau} \cdot \left[F((1 - \vz) \hprod \vo) + \sum_{i = 1}^{\infty} \frac{1}{i!} \cdot \int_{\vs \in [0, \tau]^i} F\left((1 - \vz) \hprod \left(\vo \psum \PSum_{j=1}^{i} \vx(s_j)\right)\right) d\vs \right]
	\enspace,
\end{align*}
where the inequality follows from Lemma~\ref{lem:mathematical_general} by setting $h = 2$, $\vb^{(1)} = \vz$ and $\vb^{(2)} = \vone - \vz$.
Note that the the expression $F\left((1 - \vz) \hprod \left(\vo \psum \PSum_{j=1}^{i} \vx(s_j)\right)\right)$ in the rightmost side of the last inequality can be lower bounded by $F(\vo) - F(\vz \psum \vo) - F(\vz \hprod \vo)$ due to Lemma~\ref{lem:technical_less_ts} (by setting $\vo' = (1 - \vz) \hprod \vo$). Additionally, the DR-submodularity and non-negativity of $F$ imply together $F((1 - \vz) \hprod \vo) \geq F(\vo) + F(\vzero) - F(\vz \hprod \vo) \geq F(\vo) - F(\vz \hprod \vo)$.

Combining all the above inequalities, we get
\begin{align*}
	F(\vy(\tau) \psum \vo - \vz \hprod \vo)
	\geq{} &
	e^{-\tau} \cdot \left[F(\vo) - F(\vz \hprod \vo) + \sum_{i = 1}^{\infty} \frac{1}{i!} \cdot \int_{\vs \in [0, \tau]^i} \mspace{-18mu} \{F(\vo) - F(\vz \psum \vo) - F(\vz \hprod \vo)\} d\tau \right]\\
	={} &
	e^{-\tau} \cdot \left[F(\vo) - F(\vz \hprod \vo) + \{F(\vo) - F(\vz \psum \vo) - F(\vz \hprod \vo)\} \cdot \sum_{i = 1}^{\infty} \frac{\tau^i}{i!} \right]\\
	={} &
	e^{-\tau} \cdot \left[F(\vo) - F(\vz \hprod \vo) + \{F(\vo) - F(\vz \psum \vo) - F(\vz \hprod \vo)\} \cdot (e^\tau - 1) \right]\\
	={} &
	F(\vo) - F(\vz \hprod \vo) - (1 - e^{-\tau}) \cdot F(\vz \psum \vo)
	\enspace.
	\qedhere
\end{align*}
\end{proof}

\begin{corollary} \label{cor:low_Q_non-empty}
For every $\tau \in [0, t_s)$, the vector $\vo$ belongs to $Q(\tau)$, and therefore, $Q(\tau)$ is non-empty.
\end{corollary}
\begin{proof}
Since $\vo \in P$, to prove the corollary we only need to show that $\inner{\vw(\tau)}{\vo} \geq V(\tau) - F(\vy(\tau))$, which follows by combining Lemmata~\ref{lem:product_lower_bound} and~\ref{lem:first_o_guarantee}.
\end{proof}

We now shift our attention to values of $\tau$ above $t_s$. The next lemma lower bounds $F(\vy(\tau) \psum \vo - \vz(\tau) \hprod \vo)$ within this range of $\tau$ values.
\begin{lemma} \label{lem:second_o_guarantee}
For every $\tau \in [t_s, 1)$,
\begin{align*}
	F(\vy(\tau) \psum \vo - \vz(\tau) \hprod \vo)
	={} &
	F(\vy(\tau) \psum \vo)\\
	\geq{} &
	e^{-\tau} \cdot \left[e^{t_s} \cdot F(\vo) - (e^{t_s} - 1) \cdot F(\vz \psum \vo) + \int_{t_s}^\tau F(\vo \psum \vx(s)) ds \right]
	\enspace.
\end{align*}
\end{lemma}
\begin{proof}
For notational convenience, let us define $\vx'(s_j)$ as equal to $\vx(s_j)$ for $s_j \geq t_s$ and equal to $\vzero$ otherwise. Then, by Lemma~\ref{lem:y_value},
\begin{align*}
	F(\vy(\tau) \psum \vo)
	={} &
	F\left(\left(\vone - \left(\vz \psum e^{-\int_0^{t_s} \vx(s) ds}\right) \hprod e^{-\int_{t_s}^\tau \vx(s) ds}\right) \psum \vo\right)\\
	={} &
	F\left(\vone - \vz \hprod (\vone - \vo) \hprod e^{-\int_{0}^\tau \vx'(s) ds}-  (\vone - \vz) \hprod (\vone - \vo) \hprod e^{-\int_0^\tau \vx(s) ds}\right)\\
	\geq{} &
	e^{-\tau} \cdot \left[F(\vone - \vz \hprod (\vone - \vo) + (\vone - \vz) \hprod (\vone - \vo)) \vphantom{\left(\PSum_{j = 1}^{i}\right)} \right.\\&\left.+ \sum_{i = 1}^{\infty} \frac{1}{i!} \cdot \int_{\vs \in [0, \tau]^i} F\left(\vz \hprod \left(\vo \psum \PSum_{j=1}^{i} \vx'(s_j)\right) + (\vone - \vz) \hprod \left( \vo \psum \PSum_{j=1}^{i} \vx(s_j)\right)\right) d\vs \right]
	\enspace,
\end{align*}
where the inequality follows from Lemma~\ref{lem:mathematical_general} by setting $h = 2$, $\vb^{(1)} = \vz$ and $\vb^{(2)} = \vone - \vz$.

At this point we need to lower bound the expression $F(\vz \hprod (\vo \psum \PSum_{j=1}^{i} \vx'(s_j)) + (1 - \vz) \hprod (\vo \psum \PSum_{j=1}^{i} \vx(s_j)))$ from the rightmost side of the last inequality. The lower bound we use for that purpose depends on the value of $i$ and the vector $\vs$. If $i = 1$ and the only coordinate of $\vs$ is at least $t_s$, then we observe that the above expression is equal to $F(\vo \psum \vx(s_1))$. Otherwise, if $i \geq 2$ and some coordinate of $\vs$ takes a value of at least $t_s$, then we simply use the non-negativity of $F$ to lower bound the above expression by $0$. It remains to lower bound the above expression in the case of $\vs \in [0, t_s)^i$. In this case,
\begin{align} \label{eq:basic_lower_bound_above_ts}
	F\left(\vz \hprod \left(\vo \psum \PSum_{j=1}^{i} \vx'(s_j)\right) + (1 - \vz) \hprod \left(\vo \psum \PSum_{j=1}^{i} \vx(s_j)\right)\right)\mspace{-300mu}&\mspace{300mu}\\\nonumber
	={} &
	F\left(\vo + (1 - \vz) \hprod (1 - \vo) \hprod \PSum_{j=1}^{i} \vx(s_j)\right)
	\geq
	F(\vo) - F(\vz \psum \vo)
	\enspace,
\end{align}
where the inequality follows from Lemma~\ref{lem:technical_less_ts} (by setting $\vo' = \vo$).

Plugging into Inequality~\eqref{eq:basic_lower_bound_above_ts} the lower bounds discussed above and the equality $\vone - \vz \hprod (\vone - \vo) + (\vone - \vz) \hprod (\vone - \vo) = \vo$, we obtain
{\allowdisplaybreaks\begin{align*}
	F(\vy(\tau) \psum \vo)
	\geq{} &
	e^{-\tau} \cdot \left[F(\vo) + \int_{t_s}^\tau F(\vo \psum \vx(s)) ds + \sum_{i = 1}^{\infty} \frac{1}{i!} \cdot \int_{\vs \in [0, t_s]^i} \mspace{-18mu} \{F(\vo) - F(\vz \psum \vo)\} d\tau \right]\\
	={} &
	e^{-\tau} \cdot \left[F(\vo) + \int_{t_s}^\tau F(\vo \psum \vx(s)) ds + \{F(\vo) - F(\vz \psum \vo)\} \cdot \sum_{i = 1}^{\infty} \frac{t_s^i}{i!} \right]\\
	={} &
	e^{-\tau} \cdot \left[F(\vo) + \int_{t_s}^\tau F(\vo \psum \vx(s)) ds + \{F(\vo) - F(\vz \psum \vo)\} \cdot (e^{t_s} - 1) \right]\\
	={} &
	e^{-\tau} \cdot \left[e^{t_s} \cdot F(\vo) - (e^{t_s} - 1) \cdot F(\vz \psum \vo) + \int_{t_s}^\tau F(\vo \psum \vx(s)) ds \right]
	\enspace.
	\qedhere
\end{align*}}%
\end{proof}

\begin{corollary} \label{cor:good_below_empty}
If $\vo \not \in Q(\tau)$ for some $\tau \in [t_s, 1)$, then there exists a value $\tau' \in [t_s, \tau - \eps / 2]$ such that $F(\vx(\tau') \psum \vo) \leq F(\vz \psum \vo) - \eps \cdot F(\vo)$.
\end{corollary}
\begin{proof}
Since $\vo$ does not belong to $Q(\tau)$ despite belonging to $P$, $\inner{\vw(\tau)}{\vo} < V(\tau) - F(\vy(\tau))$. By combining Lemmata~\ref{lem:product_lower_bound} and~\ref{lem:second_o_guarantee}, we get a lower bound on $\inner{\vw(\tau)}{\vo}$. Plugging this lower bound and the definition of $V(\tau)$ into the last inequality yields
\begin{align*}
	e^{-\tau} \cdot{}\mspace{180mu}&\mspace{-180mu} \left[e^{t_s} \cdot F(\vo) - (e^{t_s} - 1) \cdot F(\vz \psum \vo) + \int_{t_s}^\tau F(\vo \psum \vx(s)) ds \right] - F(\vy(\tau))\\
	\leq{} &
	e^{-\tau}\left[(e^{t_s} - 2\eps)\cdot F(\vo) -\left(e^{t_s}-1-(\tau-t_s)\right) \cdot f(\vz \psum \vo)\right] - F(\vy(\tau))
	\enspace,
\end{align*}
and by rearranging this inequality, we obtain
\begin{equation} \label{eq:violating_boils_down}
	\int_{t_s}^\tau F(\vo \psum \vx(s)) ds
	\leq
	(\tau-t_s) \cdot F(\vz \psum \vo) - 2\eps \cdot F(\vo)
	\enspace.
\end{equation}
Notice that, since the left hand side of the last inequality is non-negative, we immediately get $F(\vz \psum \vo) - \eps \cdot F(\vo) \geq (\tau-t_s) \cdot F(\vz \psum \vo) - 2\eps \cdot F(\vo) \geq 0$. This observation is used below.

Assume now towards a contradiction that the corollary is false. This implies that the integral on the left hand side of Inequality~\eqref{eq:violating_boils_down} can be lower bounded by $(\tau - \eps / 2 - t_s) \cdot [F(\vz \psum \vo) - \eps \cdot F(\vo)]$---notice that we have used here the fact that $F$ is non-negative, and when $\tau - t_s \leq \eps/2$, we also used the inequality $F(\vz \psum \vo) - \eps \cdot F(\vo) \geq 0$ proved above. Plugging this lower bound into Inequality~\eqref{eq:violating_boils_down} yields
\[
	(\tau - \eps / 2 - t_s) \cdot [F(\vz \psum \vo) - \eps \cdot F(\vo)]
	\leq
	(\tau-t_s) \cdot f(\vz \psum \vo) - 2\eps \cdot F(\vo)
	\enspace,
\]
and rearranging this inequality gives
\[
	(2 - \tau + \eps / 2 + t_s) \cdot F(\vo)
	\leq
	\tfrac{1}{2} \cdot F(\vz \psum \vo)
	\enspace.
\]
Since $(2 - \tau + \eps / 2 + t_s) \geq 1 + \eps / 2 > 1$, we get $F(\vz \psum \vo) > 2 \cdot F(\vo)$ (recall that we assume that $F(\vo)$ is strictly positive). However, this leads to a contradiction since the DR-submodularity of $F$ implies that
\begin{align*}
	F(\vz \psum \vo)
	={} &
	F(\vo + (1 - \vo) \hprod \vz)
	\leq
	F(\vo) + F((1 - \vo) \hprod \vz) - F(\vzero)\\
	\leq{} &
	F(\vo) + F((1 - \vo) \hprod \vz)
	\leq
	2 \cdot F(\vo)
	\enspace,
\end{align*}
where the second inequality holds by $F$'s non-negativity, and the last inequality holds by the definition of $\vo$ since the down-closeness of $P$ implies that $(1 - \vo) \hprod \vz \in P$ since $(1 - \vo) \hprod \vz \leq \vz$.
\end{proof}

We are now ready to complete the proof of Theorem~\ref{thm:main-component} for the case in which $Q(\tau)$ is empty for some $\tau \in [0, 1)$.
\begin{lemma}
If there exists a value $\tau \in [0, 1)$ for which $Q(\tau)$ is empty, then there exists a value $\tau' \in [t_s, 1)$ for which $\vx(\tau')$ is a local maximum with respect to $\vo$, and $F(\vx(\tau') \psum \vo) \leq F(\vz \psum \vo) - \eps \cdot F(\vo)$. Thus, Theorem~\ref{thm:main-component} holds in this case.
\end{lemma}
\begin{proof}
Let $\tau_I$ be the infimum over all values $\tau \in [0, 1)$ for which $\vo \not \in Q(\tau)$. Notice that there are such $\tau$ values since we assume that $Q(\tau)$ is empty for some value of $\tau$. Corollary~\ref{cor:low_Q_non-empty} guarantees that all the values of $\tau$ on which we take the infimum are at least $t_s$, and therefore, $\tau_I \geq t_s$.

By the definition of $\tau_I$, there must exist a value $s \in [\tau_I, \tau_I + \eps/4]$ such that $\vo \not \in Q(s)$; and thus, by Corollary~\ref{cor:good_below_empty}, there exists a value $\tau' \in [t_s, s - \eps / 2] \subseteq [t_s, \tau_I)$ such that $F(\vx(\tau') \psum \vo) \leq F(\vz \psum \vo) - \eps \cdot F(\vo)$.

Since $\tau' < \tau_I$, $\vo \in Q(\tau')$ by the definition of $\tau_I$, which in particular implies that $Q(\tau')$ is not empty. Thus, $\vx(\tau')$ is chosen by Algorithm~\ref{alg:aided_mcgreedy_continuous} as a local maximum of $Q(\tau')$, and in particular, it is a local maximum with respect to the vector $\vo$ included in this convex body.
\end{proof}

\subsection{Discrete Time Version} \label{ssc:component_discrete}

In this section, we prove Theorem~\ref{thm:main-component-discrete}, which we repeat here for convenience.

\thmMainComponentDiscrete*

The algorithm {\FWMCG} used to prove Theorem~\ref{thm:main-component-discrete} is given as Algorithm~\ref{alg:component-discrete}, which is a discrete time analog of  Algorithm~\ref{alg:aided_mcgreedy_continuous}. 
In the description of Algorithm~\ref{alg:component-discrete}, we assume for simplicity that $\delta^{-1}$ is integral and $\delta \leq \eps$ (which allows us to set $m = \delta^{-1}$). If this is not the case, then the value of $\delta$ can be reduced to $1 / \lceil 1 / \min\{\delta, \eps\} \rceil$.
Finally, in order to remove the assumption used by Algorithm~\ref{alg:aided_mcgreedy_continuous} that $F(\vo),F(\vz \hprod \vo)$ and $F(\vz \psum \vo)$ are known, we prove in Appendix~\ref{sec:knowledge} (using standard guessing arguments) the following lemma. 
\begin{restatable}{lemma}{lemGuess}
\label{lem:guess_tau}
It is possible to compute a constant size (depending only on $\eps$) set of triples $G$ of non-negative values that is guaranteed to contain a triple $(g, g_{\hprod}, g_{\psum})\in G$ such that
\begin{align*}
   (1 - \eps) \cdot F(\vo) & \leq g \leq F(\vo) \enspace,\\
   F(\vz \hprod \vo) - \eps g & \leq g_{\hprod} \leq F(\vz \hprod \vo) \enspace, \quad\text{and}\\
   F(\vz \psum \vo) - \eps g & \leq g_{\psum} \leq F(\vz \psum \vo) \enspace.
\end{align*}
\end{restatable}
Therefore, by enumerating all triples of $G$, we may assume Algorithm~\ref{alg:component-discrete} has access to values $g, g_{\hprod}$, and $g_{\psum}$ with the properties guaranteed by Lemma~\ref{lem:guess_tau}. 

\begin{algorithm}
\caption{{\FWMCGFull}($F, P, \vz, t_s, \eps, \delta$)} \label{alg:component-discrete}
\DontPrintSemicolon


Let $i_s = \lceil t_s / \delta \rceil$.\\
\For{$i = 0$ \KwTo $\delta^{-1} - 1$}
{
	Let $V(i) \triangleq \left\{\begin{array}{ll}(1 - 2\eps) \cdot g - g_{\hprod} - (1- (1 - \delta)^i)\cdot g_{\psum}  & i < i_s\\
	(1 - \delta)^i \cdot \left[((1 - \delta)^{-i_s} - 4\eps)\cdot g -\left((1 - \delta)^{-i_s}-1-\delta(i-i_s)\right) \cdot g_{\psum}\right] & i \geq i_s \enspace.
\end{array}\right.$ \\
	Let $\vz(i) \triangleq \left\{\begin{array}{ll}\vz  & i < i_s \\
\vzero & i \geq i_s \enspace.
\end{array}\right.$
}

\BlankLine

Let $\vy(0) \gets \vzero$.\\
\For{$i = 1$ \KwTo $\delta^{-1}$}
{
  Let $\vw(i) \gets (\vone - \vy(i - 1)-\vz(i - 1)) \hprod \nabla F(\vy(i - 1))$.\label{line:weight_discrete}\\
  Let $Q(i) \triangleq  \{\vx \in \cP \mid  \inner{\vw(i)}{\vx} \geq V(i - 1) - F(\vy(i - 1))\}$.\\ 
  Use the algorithm from Theorem~\ref{thm:local_search} to find an approximate local maximum $\vx(i)$ of $Q(i)$ (if $Q(i) = \varnothing$, set $\vx(i)$ to be an arbitrary vector in $P$).\\
	Let $\vy(i) \gets \vy(i - 1) + \delta (\vone - \vy(i - 1)-\vz(i - 1)) \odot \vx(i)$.\label{line:update}
}
\Return the vectors $\vy(\delta^{-1})$ and $\vx(1), \vx(2), \dotsc, \vx(\delta^{-1})$.
\end{algorithm}

One can verify that Algorithm~\ref{alg:component-discrete} can be implemented to run in the time complexity stated in Theorem~\ref{thm:main-component-discrete}. The rest of this section is devoted to the showing that the output vector $\vy(\delta^{-1})$ and vectors $\vx(1), \vx(2), \dotsc, \vx(\delta^{-1})$ of Algorithm~\ref{alg:component-discrete} obey the other guarantees of Theorem~\ref{thm:main-component-discrete}. This proof has the same general plan as the analysis of Algorithm~\ref{alg:aided_mcgreedy_continuous} in Section~\ref{ssc:continuous_algorithm}. In particular, we begin by deriving a formula for $\vy(i)$, which we later use to show that $\vy(\delta^{-1}) \in P$.
\begin{lemma} \label{lem:y_value_discrete}
For every integral $0 \leq i \leq \delta^{-1}$,
\[
	\vy(i)
	=\left\{
	\begin{array}{ll}
		(\vone - \vz) \hprod \PSum_{j = 1}^i (\delta \cdot \vx(j)) & i \leq i_s \\[1mm]
		(\vone - \vz) \hprod \PSum_{j = 1}^i (\delta \cdot \vx(j)) + \vz \hprod \PSum_{j = i_s + 1}^i (\delta \cdot \vx(j)) & i \geq i_s \enspace.
	\end{array}\right.
\]
For notational convenience, we assume in this lemma that $\PSum_{j = a}^{a - 1} (\delta \cdot \vx(j)) = \vzero$ for every value $a$.
\end{lemma}
\begin{proof}
We prove the lemma by induction on $i$. Clearly, the lemma holds for $i = 0$ since $\vy(0)$ is set to $\vzero$. Assume now that the lemma holds for $i - 1$, and let us prove it for an integer $1 \leq i \leq \delta^{-1}$. There are two cases to consider based on the relationship between $i$ and $i_s$. If $i \leq i_s$, then $i - 1 < i_s$, and therefore,
\begin{align*}
	\vy(i)
	={} &
	\vy(i - 1) + \delta (\vone - \vy(i - 1)-\vz(i - 1)) \odot \vx(i)
	=
	(\vone - \delta \cdot \vx(i)) \hprod \vy(i - 1) + \delta (\vone-\vz) \odot \vx(i)\\
	={} &
	(\vone - \delta \cdot \vx(i)) \hprod (\vone - \vz) \hprod \PSum_{j = 1}^{i - 1} (\delta \cdot \vx(j)) + \delta (\vone-\vz) \odot \vx(i)\\
	={} &
	(\vone - \vz) \hprod \left[(\vone - \delta \cdot \vx(i)) \hprod \PSum_{j = 1}^{i - 1} (\delta \cdot \vx(j)) + \delta \cdot \vx(i)\right]
	=
	(\vone - \vz) \hprod \PSum_{j = 1}^{i}(\delta \cdot \vx(j))
	\enspace,
\end{align*}
where the third equality holds by the induction hypothesis.

Consider now the case of $i > i_s$. Since $i$ and $i_s$ are both integral, we get in this case $i - 1 \geq i_s$, and thus,
{\allowdisplaybreaks\begin{align*}
	\vy(i)
	={} &
	\vy(i - 1) + \delta (\vone - \vy(i - 1)-\vz(i - 1)) \odot \vx(i)
	=
	(\vone - \delta \cdot \vx(i)) \hprod \vy(i - 1) + \delta \cdot \vx(i)\\
	={} &
	(\vone - \delta \cdot \vx(i)) \hprod \left[(\vone - \vz) \hprod \PSum_{j = 1}^{i - 1} (\delta \cdot \vx(j)) + \vz \hprod \PSum_{j = i_s + 1}^{i - 1} (\delta \cdot \vx(j))\right] + \delta \cdot \vx(i)\\
	={} &
	(\vone - \vz) \hprod \left[(\vone - \delta \cdot \vx(i)) \hprod \PSum_{j = 1}^{i - 1} (\delta \cdot \vx(j)) + \delta \cdot \vx(i)\right] \\&\mspace{200mu}+ \vz \hprod \left[(\vone - \delta \cdot \vx(i)) \hprod \PSum_{j = i_s + 1}^{i - 1} (\delta \cdot \vx(j)) + \delta \cdot \vx(i)\right]\\
	={} &
	(\vone - \vz) \hprod \PSum_{j = 1}^{i} (\delta \cdot \vx(j)) + \vz \hprod \PSum_{j = i_s + 1}^{i} (\delta \cdot \vx(j))
	\enspace,
\end{align*}}%
where the third equality holds by the induction hypothesis.
\end{proof}
\begin{corollary} \label{cor:belonging_discrete}
The vector $\vy(\delta^{-1})$ belongs to $P$.
\end{corollary}
\begin{proof}
%
Since $\delta^{-1} \geq i_s$, we get, by Lemma~\ref{lem:y_value_discrete},
\begin{align*}
	\vy(\delta^{-1})
	={} &
	(\vone - \vz) \hprod \PSum_{j = 1}^{\delta^{-1}} (\delta \cdot \vx(j)) + \vz \hprod \PSum_{j = i_s + 1}^{\delta^{-1}} \mspace{-6mu} (\delta \cdot \vx(j))\\
	\leq{} &
	(\vone - \vz) \hprod \PSum_{j = 1}^{\delta^{-1}} (\delta \cdot \vx(j)) + \vz \hprod \PSum_{j = 1}^{\delta^{-1}} (\delta \cdot \vx(j))
	=
	\PSum_{j = 1}^{\delta^{-1}} (\delta \cdot \vx(j))
	\leq
	\sum_{j = 1}^{\delta^{-1}} (\delta \cdot \vx(j))
	\in
	P
	\enspace,
\end{align*}
where the inclusion holds since $P$ is a convex body and $\vx(j) \in P$ for every $j \in [\delta^{-1}]$. The corollary now follows since $P$ is down-closed.
\end{proof}

It remains to show that at least one of the two bullets stated in Theorem~\ref{thm:main-component-discrete} always holds. We do that by considering two case. The first case is when $Q(i)$ is non-empty for every $i \in [\delta^{-1}]$. We show below that in this case $F(\vy(\delta^{-1}))$ is large enough to satisfy the guarantee given in the first bullet. Towards this goal, we begin with the following observation, which lower bounds the rate in which $F(\vy(i))$ increases as a function of $i$ when $Q(i)$ is non-empty.
\begin{observation} \label{obs:basic_increase}
If $Q(i)$ is non-empty for some $i \in [\delta^{-1}]$, then $F(\vy(i)) - F(\vy(i - 1)) \geq \delta[V(i - 1) - F(\vy(i - 1))] - \delta^2 D^2L / 2$.
\end{observation}
\begin{proof}
By the chain rule,
{\allowdisplaybreaks\begin{align*}
	F(\vy(i))&{} - F(\vy(i - 1))
	=
	F(\vy(i - 1) + \delta(\vone - \vy(i - 1) - \vz(i - 1)) \hprod \vx(i))- F(\vy(i - 1))\\
	={} &
	\int_0^\delta \inner{(\vone - \vy(i - 1) - \vz(i - 1)) \hprod \vx(i)}{\nabla F(\vy(i - 1) + s(\vone - \vy(i - 1) - \vz(i - 1))  \hprod \vx(i))} ds\\
	={} &
	\delta \cdot \inner{\vw(i)}{\vx(i)} + \int_0^\delta \langle(\vone - \vy(i - 1) - \vz(i - 1)) \hprod \vx(i), \\&\mspace{150mu}\nabla F(\vy(i - 1) + s(\vone - \vy(i - 1) - \vz(i - 1))  \hprod \vx(i)) - \nabla F(\vx(i))\rangle ds\\
	\geq{} &
	\delta \cdot \inner{\vw(i)}{\vx(i)} - \int_0^\delta \|\vone - \vy(i - 1) - \vz(i - 1)) \hprod \vx(i)\|_2 \\&\mspace{150mu} \cdot \|\nabla F(\vy(i - 1) + s(\vone - \vy(i - 1) - \vz(i - 1))  \hprod \vx(i)) - \nabla F(\vx(i))\|_2 ds\\
	\geq{} &
	\delta \cdot \inner{\vw(i)}{\vx(i)} - \int_0^\delta sL \cdot \|\vone - \vy(i - 1) - \vz(i - 1)) \hprod \vx(i)\|^2_2 ds\\
	\geq{} &
	\delta \cdot \inner{\vw(i)}{\vx(i)} - \delta^2 D^2L / 2
	\geq
	\delta[V(i - 1) - F(\vy(i - 1))] - \delta^2 D^2L / 2
	\enspace,
\end{align*}}%
where the first inequality holds by the Cauchy-Schwarz inequality, the second inequality follows from the $L$-smoothness of $F$, and the last inequality holds since the assumption of the lemma that $Q(i)$ is non-empty implies $\vx(i) \in Q(i)$.
\end{proof}

The last observation implies a lower bound on $F(\vy(i))$ given as a recursive formula. The following claims (up to Corollary~\ref{cor:final_value}) aim to obtain a closed form formula for this lower bound.

\begin{lemma} \label{lem:recursive_solve_first}
If $Q(i)$ is non-empty for every $i \in [i_s]$, then, for every integer $0 \leq i \leq i_s$, $F(\vy(i)) \geq (1 - (1 - \delta)^i) \cdot ((1 - 2\eps)g - g_{\hprod}) - (1 - (1 + \delta i) \cdot (1 - \delta)^i) \cdot g_{\psum} - i \cdot \delta^2 D^2 L$.
\end{lemma}
\begin{proof}
We prove the lemma by induction on $i$. For $i = 0$, the lemma follows immediately from the non-negativity of $F$. Assume now that the lemma holds for $i - 1$, and let us prove it for $i \geq 1$. By Observation~\ref{obs:basic_increase},
\begin{align*}
	F(\vy(i))
	\geq{} &
	F(\vy(i - 1)) + \delta[V(i - 1) - F(\vy(i - 1))] - \delta^2 D^2L / 2\\
	={} &
	(1 - \delta) \cdot F(\vy(i - 1)) + \delta[(1 - 2\eps) \cdot g - g_{\hprod} - (1- (1 - \delta)^i)\cdot g_{\psum}] - \delta^2 D^2L / 2\\
	\geq{} &
	(1 - \delta) \cdot [(1 - (1 - \delta)^{i - 1}) \cdot ((1 - 2\eps)g - g_{\hprod}) - (1 - (1 + \delta (i - 1)) \cdot (1 - \delta)^{i - 1}) \cdot g_{\psum}] \\&\mspace{150mu}+ \delta[(1 - 2\eps) \cdot g - g_{\hprod} - (1- (1 - \delta)^i)\cdot g_{\psum}] - i \cdot \delta^2 D^2 L\\
	={} &
	(1 - (1 - \delta)^i \cdot ((1 - 2\eps)g - g_{\hprod}) - (1 - (1 + \delta i) \cdot (1 - \delta)^i) \cdot g_{\psum} - i \cdot \delta^2 D^2 L
	\enspace,
\end{align*}
where the first equality holds since $i \leq i_s$ implies $i - 1 < i_s$, and the second inequality holds by the induction hypothesis.
\end{proof}

\begin{lemma} \label{lem:recursive_solve_second}
If $Q(i)$ is non-empty for every integer $i_s < i \leq \delta^{-1}$, then, for every integer $i_s \leq i \leq \delta^{-1}$,
\begin{align*}
	F(\vy(i))
	\geq{} &
	\delta (i - i_s) \cdot (1 - \delta)^i \cdot \left[((1 - \delta)^{-i_s} - 4\eps)\cdot g -\left((1 - \delta)^{-i_s}-1 - \frac{\delta (i + 1 - i_s)}{2}\right) \cdot g_{\psum}\right] \\&\mspace{300mu} + (1 - \delta)^{i - i_s} \cdot F(\vy(i_s)) - (i - i_s) \cdot \delta^2 D^2 L
	\enspace.
\end{align*}
\end{lemma}
\begin{proof}
We prove the lemma by induction on $i$. For $i = i_s$, the lemma trivially holds since the right hand side of the equality of the lemma reduces to $F(\vy(i_s))$ for this choice of $i$ value. Assume now that the lemma holds for $i - 1$, and let us prove it for an integer $i \in (i_s, \delta^{-1}]$.  By Observation~\ref{obs:basic_increase},
\begin{align*}
	F(\vy(i)&)
	\geq
	F(\vy(i - 1)) + \delta[V(i - 1) - F(\vy(i - 1))] - \delta^2 D^2L / 2\\
	={} &
	(1 - \delta) \cdot F(\vy(i - 1)) + \delta(1 - \delta)^i \cdot \left[(1 - \delta)^{-i_s} - 4\eps)\cdot g -\left((1 - \delta)^{-i_s}-1-\delta(i-i_s)\right) \cdot g_{\psum}\right]\\&\mspace{500mu}- \delta^2 D^2L / 2\\
	\geq{} &
	\delta (i - 1 - i_s) \cdot (1 - \delta)^i \cdot \left[((1 - \delta)^{-i_s} - 4\eps)\cdot g -\left((1 - \delta)^{-i_s}-1 - \frac{\delta (i - i_s)}{2}\right) \cdot g_{\psum}\right]\\&\mspace{150mu}+ \delta(1 - \delta)^i \cdot \left[(1 - \delta)^{-i_s} - 4\eps)\cdot g -\left((1 - \delta)^{-i_s}-1-\delta(i-i_s)\right) \cdot g_{\psum}\right] \\&\mspace{150mu}+ (1 - \delta)^{i - i_s} \cdot F(\vy(i_s)) - (i - i_s) \cdot \delta^2 D^2 L\\
	={} &
	\delta (i - i_s) \cdot (1 - \delta)^i \cdot \left[((1 - \delta)^{-i_s} - 4\eps)\cdot g -\left((1 - \delta)^{-i_s}-1 - \frac{\delta (i + 1 - i_s)}{2}\right) \cdot g_{\psum}\right]\\&\mspace{150mu}+ (1 - \delta)^{i - i_s} \cdot F(\vy(i_s)) - (i - i_s) \cdot \delta^2 D^2 L
	\enspace,
\end{align*}
where the first equality holds since the integrality of $i$ and $i_s$ implies $i - 1 \geq i_s$, and the second inequality holds by the induction hypothesis.
\end{proof}

Combining the last two lemmata, we get the following corollary, which completes the proof of Theorem~\ref{thm:main-component-discrete} in the case that $Q(i)$ is always non-empty.
\begin{corollary} \label{cor:final_value}
If $Q(i)$ is non-empty for every $i \in [\delta^{-1}]$, then
\begin{align*}
	F(\vy(\delta^{-1}))
	\geq{} &
	e^{-1} \cdot \left[((2 - t_s)e^{t_s} - 1 - O(\eps)) \cdot F(\vo) - (e^{t_s} - 1) \cdot F(\vz \hprod \vo) - \vphantom{\frac{1}{2}}\right.\\&\mspace{190mu} \left.\left((2 - t_s)e^{t_s} + t_s - \frac{t_s^2 + 5}{2}\right) \cdot F(\vz \psum \vo) \right] - \delta D^2 L
	\enspace.
\end{align*}
\end{corollary}
\begin{proof}
Plugging the lower bound given by Lemma~\ref{lem:recursive_solve_first} for $i = i_s$ into the lower bound given by Lemma~\ref{lem:recursive_solve_second} for $i = \delta^{-1}$, we get
\begin{align} \label{eq:basic_bound}
	F(\vy(\delta^{-1}))
	\geq{} &
	(1 - \delta i_s) \cdot (1 - \delta)^{1/\delta} \left[((1 - \delta)^{-i_s} - 4\eps) g -\left((1 - \delta)^{-i_s}-1 - \frac{1 + \delta (1 - i_s)}{2}\right) g_{\psum}\right] \\\nonumber&\mspace{200mu}+ (1 - \delta)^{1/\delta} \cdot ((1 - \delta)^{-i_s} - 1) \cdot ((1 - 2\eps)g - g_{\hprod}) \\\nonumber&\mspace{200mu}- ( (1 - \delta)^{1/\delta - i_s} - (1 + \delta i_s) \cdot (1 - \delta)^{1/\delta}) \cdot g_{\psum} - \delta D^2 L
	\enspace.
\end{align}

The coefficient of $g$ on the right hand side of Inequality~\eqref{eq:basic_bound} is
\begin{align*}
	(1 - \delta i_s) \cdot (&1 - \delta)^{1/\delta} \cdot ((1 - \delta)^{-i_s} - 4\eps) + (1 - \delta)^{1/\delta} \cdot ((1 - \delta)^{-i_s} - 1) \cdot (1 - 2\eps)\\
	={} &
	(1 - \delta i_s) \cdot (1 - \delta)^{1/\delta} \cdot (1 - \delta)^{-i_s} + (1 - \delta)^{1/\delta} \cdot ((1 - \delta)^{-i_s} - 1) - O(\eps)\\
	={} &
	(1 - \delta)^{1/\delta} \cdot [(1 - \delta)^{-i_s} \cdot (2 - \delta i_s) - 1] - O(\eps)
	\geq
	e^{-1}(1 - \delta) \cdot [e^{\delta i_s} \cdot (2 - \delta i_s) - 1] - O(\eps)\\
	\geq{} &
	e^{-1}(1 - \delta) \cdot [e^{t_s} \cdot (2 - t_s - \delta) - 1] - O(\eps)
	=
	e^{-1}\cdot [(2 - t_s)e^{t_s} - 1] - O(\eps)
	\enspace,
\end{align*}
where the last equality holds due to our assumption that $\delta \leq \eps$.

The coefficient of $g_{\hprod}$ on the right hand side of Inequality~\eqref{eq:basic_bound} is
\begin{align*}
	- (1 - \delta)^{1/\delta} \cdot ((1 - \delta)^{-i_s} - 1)
	={} &
	(1 - \delta)^{1/\delta} - (1 - \delta)^{1/\delta - i_s}
	\geq
	e^{-1}(1 - \delta) - e^{\delta i_s - 1}\\
	\geq{} &
	e^{-1}(1 - \delta) - (1 + 2\delta) \cdot e^{t_s - 1}
	=
	-e^{-1}(e^{t_s} - 1) - O(\eps)
	\enspace,
\end{align*}
where the last equality follows from our assumption that $\delta \leq \eps$.

Finally, the coefficient of $g_{\psum}$ on the right hand side of Inequality~\eqref{eq:basic_bound} is
\begin{align*}
	&
	- (1 - \delta i_s) \cdot (1 - \delta)^{1/\delta} \left((1 - \delta)^{-i_s}-1 - \frac{1 + \delta (1 - i_s)}{2}\right) - (1 - \delta)^{1/\delta - i_s} + (1 + \delta i_s) \cdot (1 - \delta)^{1/\delta}\\
	\geq{} &
	- (1 - \delta i_s) \cdot (1 - \delta)^{1/\delta} \left((1 - \delta)^{-i_s}-1 - \frac{1 - \delta i_s}{2}\right) - (1 - \delta)^{1/\delta - i_s} + (1 + \delta i_s) \cdot (1 - \delta)^{1/\delta}\\
	={} &
	- (1 - \delta)^{1/\delta} \left(2(1 - \delta)^{-i_s} + \delta i_s(1 - (1 - \delta)^{-i_s}) - \frac{5 + \delta^2i_s^2}{2} \right)\\
	\geq{} &
	- (1 - \delta)^{1/\delta} \left(2(1 - \delta)^{-i_s} + t_s(1 - (1 - \delta)^{-i_s}) + \delta - \frac{5 + t_s^2}{2} \right)\\
	={} &
	- (1 - \delta)^{1/\delta}\left(t_s + \delta - \frac{5 + t_s^2}{2}\right) - (1 - \delta)^{1/\delta - i_s}(2 - t_s)
	\geq
	- e^{-1}\left(t_s + \delta - \frac{5 + t_s^2}{2}\right) - e^{\delta i_s - 1}(2 - t_s)\\
	\geq{} &
	- e^{-1}\left(t_s + \delta - \frac{5 + t_s^2}{2}\right) - (1 + 2\delta) \cdot e^{t_s - 1}(2 - t_s)
	=
	-e^{-1}\left((2 - t_s)e^{t_s} + t_s - \frac{t_s^2 + 5}{2}\right) - O(\eps)
	\enspace,
\end{align*}
where the last equality again holds due to our assumption that $\delta \leq \eps$.

Combining all the above inequalities, we get
\begin{align*}
	F(\vy(\delta^{-1}))
	\geq{} &
	e^{-1} \cdot \left[((2 - t_s)e^{t_s} - 1) \cdot g - (e^{t_s} - 1) \cdot g_{\hprod} - \left((2 - t_s)e^{t_s} + t_s - \frac{t_s^2 + 5}{2}\right) \cdot g_{\psum} \right] \\&\mspace{400mu}- O(\eps) \cdot (g + g_{\hprod} + g_{\psum}) - \delta D^2 L\\
	\geq{} &
	e^{-1} \cdot \left[((2 - t_s)e^{t_s} - 1) \cdot F(\vo) - (e^{t_s} - 1) \cdot F(\vz \hprod \vo) - \vphantom{\frac{1}{2}}\right.\\&\mspace{30mu} \left.\left((2 - t_s)e^{t_s} + t_s - \frac{t_s^2 + 5}{2}\right) \cdot F(\vz \psum \vo) \right] - O(\eps) \cdot (F(\vo) + g + g_{\hprod} + g_{\psum}) - \delta D^2 L
	\enspace,
\end{align*}
where the second inequality holds since $(2 - t_s)e^{t_s} + t_s - \frac{5 + t_s^2}{2} \geq -1$. To complete the proof of the lemma, it now remains to observe that: $g \leq F(\vo)$ by the definition of $g$; $g_{\vz \hprod \vo} \leq F(\vz \hprod \vo) \leq F(\vo)$, where the second inequality holds by the definition of $\vo$ since the down-closeness of $P$ guarantees that $\vz \hprod \vo \in P$; and $g_{\psum} \leq F(\vz \psum \vo) \leq 2\cdot F(\vo)$ (see the proof of Corollary~\ref{cor:good_below_empty} for a justification of the last inequality).
\end{proof}

At this point, it remains to consider the case in which $Q(i)$ is empty for some $i \in [\delta^{-1}]$. Note that when $Q(i)$ is empty, we get, in particular, $\vo \not \in Q(i)$. Therefore, we can define $\io \triangleq \min\{i \in [\delta^{-1}] \mid \vo \not \in Q(i)\}$.
\begin{observation} \label{obs:approximate_local_search}
For every $i \in [\io - 1]$,
\[
	F(\vx(i))
	\geq
	\frac{1}{2} \cdot [F(\vx(i) \vee \vo) + F(\vx(i) \wedge \vo)] - \eps \cdot F(\vo) - \delta D^2L
	\enspace.
\]
\end{observation}
\begin{proof}
Since $\vo \in Q(i)$ by the definition of $\io$, Theorem~\ref{thm:local_search} guarantees that the vector $\vx(i)$ obeys
\begin{align*}
	F(\vx(i))
	\geq{} &
	\frac{1}{2}\cdot \max_{\vy \in Q(i)} [F(\vx(i) \vee \vy)+ F(\vx(i) \wedge \vy)]
	-
	\frac{\delta \cdot [2 \cdot \max_{\vy' \in Q(i)} F(\vy') + D^2 L]}{4}\\
	\geq{} &
	\frac{1}{2}\cdot [F(\vx(i) \vee \vo)+ F(\vx(i) \wedge \vo)]
	-
	\frac{\delta \cdot [2 \cdot F(\vo) + D^2 L]}{4}\\
	={} &
	\frac{1}{2}\cdot [F(\vx(i) \vee \vo)+ F(\vx(i) \wedge \vo)] - O(\eps) \cdot F(\vo) + O(\delta D^2 L)
	\enspace,
\end{align*}
where the last equality holds by our assumption that $\delta \leq \eps$.
\end{proof}

Observation~\ref{obs:approximate_local_search} immediately implies Theorem~\ref{thm:main-component-discrete} if there is some $i \in [\io - 1]$ such that $F(\vx(i) \psum \vo) \leq F(\vz \psum \vo) - \eps \cdot F(\vo)$. Hence, proving Theorem~\ref{thm:main-component-discrete} now boils down to showing that the last inequality must hold for at least one value $i \in [\io - 1]$. Lemmata~\ref{lem:o_in_first} and~\ref{lem:o_in_second} show together that this is indeed the case, which completes the proof of Theorem~\ref{thm:main-component-discrete}. To prove these lemmata, we first need to prove the following helper lemma.

\begin{lemma} \label{lem:helper}
It must hold that $F(\vy(\io - 1) \psum \vo - \vz(\io - 1) \hprod \vo) \leq V(\io - 1)$.
\end{lemma}
\begin{proof}
By the definition of $\io$, $\vo \not \in Q(\io)$. Thus, we must have
\[
	\inner{\vw(\io)}{\vo} \leq V(\io - 1) - F(\vy(\io - 1))
	\enspace.
\]
We now observe that the proof of Lemma~\ref{lem:product_lower_bound} can be repeated to get, for every $i \in [\delta^{-1}]$,
\[
	\inner{\vw(i)}{\vo} \geq F(\vy(i - 1) \psum \vo - \vz(i - 1) \hprod \vo) - F(\vy(i - 1))
	\enspace.
\]
The lemma now follows by plugging $i = \io$ into the last inequality, and then combining it with the previous inequality.
\end{proof}

\begin{lemma} \label{lem:o_in_first}
It must hold that $\io > i_s$.
\end{lemma}
\begin{proof}
Assume towards a contradiction that $\io \leq i_s$. Then, since $\io - 1 < i_s$, the inequality proved by Lemma~\ref{lem:helper} can be rewritten as
\begin{align} \label{eq:before_expansion}
	F(\vy(\io - 1) \psum \vo - \vz \hprod \vo)
	\leq{} &
	(1 - 2\eps)	\cdot g - g_{\hprod} - (1- (1 - \delta)^{\io - 1})\cdot g_{\psum}\\\nonumber
	\leq{} &
	g - F(\vz \hprod \vo) - (1- (1 - \delta)^{\io - 1})\cdot F(\vz \psum \vo)\\\nonumber
	\leq{} &
	F(\vo) - F(\vz \hprod \vo) - (1- (1 - \delta)^{\io - 1})\cdot F(\vz \psum \vo)
	\enspace.
\end{align}
Notice now that Lemma~\ref{lem:y_value_discrete} implies that
\begin{align*}
	F(\vy(\io - 1) \psum \vo - \vz \hprod \vo)
	={} &
	F\left(\left((\vone - \vz) \hprod \PSum_{j = 1}^{i_o - 1} (\delta \cdot \vx(j))\right) \psum \vo - \vz \hprod \vo\right)\\
	={} &
	F\left(\mspace{-1mu}(\vone - \vz) \hprod \left(\PSum_{j = 1}^{i_o - 1} (\delta \cdot \vx(j)) \psum (1 \cdot \vo)\right) \mspace{-2mu}+ \vz \hprod \left(\PSum_{j = 1}^{i_o - 1} (\delta \cdot \vzero) \psum (1 \cdot \vzero)\right)\mspace{-1mu}\right)\\
	\geq{} &
	\sum_{S \subseteq [\io - 1]} \delta^{|S|}(1 - \delta)^{\io - 1 - |S|} \cdot F\left((\vone - \vz) \hprod \left(\PSum_{j \in S} \vx(j) \psum \vo\right)\right)
	\enspace,
\end{align*}
where the inequality follows from Corollary~\ref{cor:general_bound_discrete}. The term on the rightmost side of the last inequality corresponding to $S = \varnothing$ is equal to
\begin{align*}
	(1 - \delta)^{\io - 1} \cdot F((\vone - \vz) \hprod \vo)
	\geq{} &
	(1 - \delta)^{\io - 1} \cdot [F(\vo) - F(\vz \hprod \vo) + F(\vzero)]\\
	\geq{} &
	(1 - \delta)^{\io - 1} \cdot [F(\vo) - F(\vz \hprod \vo)]
	\enspace,
\end{align*}
where the first inequality holds by $F$'s DR-submodularity, and the second inequality follows from $F$'s non-negativity. The term corresponding to every set $S \neq \varnothing$ can be lower bounded by $\delta^{|S|}(1 - \delta)^{\io - 1 - |S|}$ times
\begin{align*}
	\mspace{160mu}&\mspace{-160mu}
	F\left((\vone - \vz) \hprod \left(\PSum_{j \in S} \vx(j) \psum \vo\right)\right)
	=
	F\left(\left((\vone - \vz) \hprod \PSum_{j \in S} (\delta \cdot \vx(j))\right) \psum \vo - \vz \hprod \vo\right)\\
	\geq{} &
	F\left(\left((\vone - \vz) \hprod \PSum_{j \in S} (\delta \cdot \vx(j))\right) \psum \vo\right) - F(\vz \hprod \vo)\\
	\geq{} &
	F\left(\left((\vone - \vz) \hprod \PSum_{j \in S} (\delta \cdot \vx(j)) + \vz\right) \psum \vo\right) - [F(\vz \psum \vo) - F(\vo)] - F(\vz \hprod \vo)\\
	\geq{} &
	F(\vo) - F(\vz \psum \vo) - F(\vz \hprod \vo)
	\enspace,
\end{align*}
where the first inequality follows form $F$'s DR-submodularity and non-negativity, the second inequality holds by the DR-submodularity of $F$, and the last inequality holds by the non-negativity of $F$. Combining the above inequalities, we get
\begin{align*}
	F(\vy(\io - 1) &{}\psum \vo - \vz \hprod \vo)
	\geq
	(1 - \delta)^{\io - 1} \cdot [F(\vo) - F(\vz \hprod \vo)] \\&\mspace{150mu}+ \sum_{\varnothing \neq S \subseteq [\io - 1]} \mspace{-18mu} \delta^{|S|}(1 - \delta)^{\io - 1 - |S|} \cdot [F(\vo) - F(\vz \psum \vo) - F(\vz \hprod \vo)]\\
	={} &
	(1 - \delta)^{\io - 1} \cdot [F(\vo) - F(\vz \hprod \vo)] + [1 - (1 - \delta)^{\io - 1}] \cdot [F(\vo) - F(\vz \psum \vo) - F(\vz \hprod \vo)]\\
	={} &
	F(\vo) - F(\vz \hprod \vo) - [1 - (1 - \delta)^{\io - 1}] \cdot  F(\vz \psum \vo)
	\enspace,
\end{align*}
which contradicts Inequality~\eqref{eq:before_expansion}.
\end{proof}

\begin{lemma} \label{lem:o_in_second}
If $\io > i_s$, then there must exist some $i \in [\io - 1]$ such that $F(\vx(i) \psum \vo) \leq F(\vz \psum \vo) - \eps \cdot F(\vo)$.
\end{lemma}
\begin{proof}
Since $\io$ and $i_s$ are integral, $\io - 1 \geq i_s$. Thus, the inequality proved by Lemma~\ref{lem:helper} can be rewritten as
\begin{align} \label{eq:before_expansion_second}
	F(\vy(\io - 1&) \psum \vo)
	\leq
	(1 - \delta)^{\io - 1} \cdot \left[((1 - \delta)^{-i_s} - 4\eps)\cdot g -\left((1 - \delta)^{-i_s}-1-\delta(\io - 1 -i_s)\right) \cdot g_{\psum}\right]\\\nonumber
	\leq{} &
	(1 - \delta)^{\io - 1} \cdot \left[((1 - \delta)^{-i_s} - \eps)\cdot g -\left((1 - \delta)^{-i_s}-1-\delta(\io - 1 -i_s)\right) \cdot F(\vz \psum \vo)\right] \\\nonumber
	\leq{} &
	(1 - \delta)^{\io - 1} \cdot \left[((1 - \delta)^{-i_s} - \eps)\cdot F(\vo) -\left((1 - \delta)^{-i_s}-1-\delta(\io - 1 -i_s)\right) \cdot F(\vz \psum \vo)\right]
	\enspace,
\end{align}
where the first inequality uses the fact that $(1 - \delta)^{-i_s} \leq (1 - \delta)^{-1/\delta} \leq 4$ because we assume that $\delta \leq \eps \leq 1/2$.

Lemma~\ref{lem:y_value_discrete} now implies
\begin{align} \label{eq:after_ts_basic_bound_discrete}
	F(\vy(\io - 1) \psum \vo)
	={} &
	F\left(\left((\vone - \vz) \hprod \PSum_{j = 1}^{\io - 1} (\delta \cdot \vx(j)) + \vz \hprod \PSum_{j = i_s + 1}^{
	\io - 1}(\delta \cdot \vx(j))\right) \psum \vo\right)\\\nonumber
	={} &
	F\left((\vone - \vz) \hprod  \left(\PSum_{j = 1}^{\io - 1} (\delta \cdot \vx(j)) \psum (1 \cdot \vo)\right) + \vz \hprod \left(\PSum_{j = i_s + 1}^{
	\io - 1}(\delta \cdot \vx(j)) \psum (1 \cdot \vo)\right)\right)\\\nonumber
	\geq{} &
	\sum_{S \subseteq [\io - 1]} \delta^{|S|} (1 - \delta)^{\io - 1 - |S|} \cdot F\left((\vone - \vz) \hprod \left(\PSum_{j \in S} \vx(j) \psum \vo \right) \right.\\\nonumber&\mspace{380mu}\left.+ \vz \hprod \left(\PSum_{j \in S \setminus [i_s]} \vx(j) \psum \vo \right)\right)
	\enspace,
\end{align}
where the inequality follows from Corollary~\ref{cor:general_bound_discrete} by thinking about $\PSum_{j = i_s + 1}^{\io - 1}(\delta \cdot \vx(j))$ as $\PSum_{j = 1}^{i_s} (\delta \cdot 0) \psum \PSum_{j = i_s + 1}^{\io - 1}(\delta \cdot \vx(j))$. We now lower bound some of the terms on rightmost side of the last inequality. Notice that terms that we do not (explicitly) lower bound are trivially lower bounded by $0$ since $F$ is non-negative. The term corresponding to $S = \varnothing$ is equal to
\[
	(1 - \delta)^{\io - 1} \cdot F((\vone - \vz) \hprod \vo + \vz \hprod \vo)
	=
	(1 - \delta)^{\io - 1} \cdot F(\vo)
	\enspace.
\]
Next, we need to bound all the terms corresponding to $\varnothing \neq S \subseteq [i_s]$. Such terms can be lower bounded by $\delta^{|S|} (1 - \delta)^{\io - 1 - |S|}$ times
\begin{align*}
	F\left((\vone - \vz) \hprod \left(\PSum_{j \in S} \vx(j) \psum \vo \right) + \vz \hprod \vo\right)
	={} \mspace{-200mu}&\mspace{200mu}
	F\left(\vo \psum \left((\vone - \vz) \hprod \PSum_{j \in S} \vx(j) \right)\right)\\
	\geq{} &
	F\left(\vo \psum \left(\vz + (\vone - \vz) \hprod \PSum_{j \in S} \vx(j) \right)\right) - [F(\vz \psum \vo) - F(\vo)]
	\geq
	F(\vo) - F(\vz \psum \vo)
	\enspace,
\end{align*}
where the first inequality follows from the DR-submodularity of $F$, and the second inequality holds by $F$'s non-negativity. It remains to bound terms corresponding to sets $S$ that contain a single element of $j \in [\io - 1] \setminus [i_s]$. Such terms are given by
\[
	\delta(1 - \delta)^{\io - 2} \cdot F((\vone - \vz) \hprod (\vx(j) \psum \vo) + \vz \hprod (\vx(j) \psum \vo))
	=
	\delta(1 - \delta)^{\io - 2} \cdot F(\vx(j) \psum \vo)
	\enspace.
\]
Plugging all the above bounds on the terms corresponding to various sets $S$ into Inequality~\eqref{eq:after_ts_basic_bound_discrete}, we get
\begin{align*}
	F(\vy(\io - 1&) \psum \vo)
	\geq
	(1 - \delta)^{\io - 1} \cdot F(\vo) + (1 - \delta)^{\io - 1 - i_s} [1 - (1 - \delta)^{i_s}] \cdot [F(\vo) - F(\vz \psum \vo)] \\&\mspace{350mu}+ \sum_{i = i_s + 1}^{\io - 1} \delta(1 - \delta)^{\io - 2} \cdot F(\vx(j) \psum \vo)\\
	={} &
	(1 - \delta)^{\io - 1} \cdot [(1 - \delta)^{-i_s} \cdot F(\vo) - ((1 - \delta)^{-i_s} - 1) \cdot F(\vz \psum \vo) \\&\mspace{350mu}+ \delta(1 - \delta)^{-1} \cdot \sum_{i = i_s + 1}^{\io - 1} F(\vx(j) \psum \vo)]\\
	\geq{} &
	(1 - \delta)^{\io - 1} \cdot [(1 - \delta)^{-i_s} \cdot F(\vo) - ((1 - \delta)^{-i_s} - 1) \cdot F(\vz \psum \vo) + \delta \cdot \sum_{i = i_s + 1}^{\io - 1} F(\vx(j) \psum \vo)]
	\enspace.
\end{align*}

Combining the last inequality with Inequality~\eqref{eq:before_expansion_second}, and rearranging, yields
\[
	\delta \cdot \sum_{i = i_s + 1}^{\io - 1} F(\vx(j) \psum \vo)
	\leq
	- \eps \cdot F(\vo) + \delta(\io - 1 -i_s) \cdot F(\vz \psum \vo)
	\leq
	\delta(\io - 1 -i_s) \cdot [F(\vz \psum \vo) - \eps \cdot F(\vo)]
	\enspace,
\]
where the second inequality holds since $\delta(\io - 1 - i_s) \leq 1$. The lemma now follows from the last inequality by an averaging argument.
\end{proof}

\section{Concluding Remarks}

In this paper, we have improved the state-of-the-art approximation ratio for maximizing either a non-negative submodular set function or a non-negative DR-submodular function over a convex set constraint to $0.401$, which is the first approximation ratio for this problem exceeding the natural $0.4$ barrier. Our result significantly reduces the gap between the state-of-the-art approximation ratio for this problem and the state-of-the-art inapproximability of $0.478$~\cite{gharan2011submodular}. Naturally, it will be interesting to further reduce this gap by either improving over our approximation ratio, or improving over the inapproximability of~\cite{gharan2011submodular} (which, at this point, is already over a decade old).

We would like to mention two points regarding the possibility of getting better approximation guarantees using our techniques. First, as written, {\FWMCG} uses a fixed $t_s$ value. Using a value of $t_s$ that depends on the ratios $F(\vz \psum \vo) / F(\vo)$ and $F(\vz \hprod \vo) / F(\vo)$ leads to an improvement over the guarantee stated in Theorem~\ref{thm:main}. However, this improvement is minimal (our calculations show that only the fourth digit to the right of the decimal point is affected), and thus, we chose to avoid introducing this extra complication into our algorithm. Second, in the context of some convex bodies, it is possible to allow {\FWMCG} to run until some time beyond $1$ (or $\delta^{-1}$ in the discrete version of the algorithm, see~\cite{feldman2011unified} for a proof of a similar claim for the related algorithm Measured Continuous Greedy). This immediately yields a somewhat improved guarantee for such convex bodies.

\bibliographystyle{plain}
\bibliography{submodular}

\appendix
\toggletrue{appendix}
\section{Technical Details of Proving Theorem~\titleRef{thm:discrete}} \label{app:technical_discrete}

As mentioned in Section~\ref{sec:introduction}, Theorem~\ref{thm:discrete} is a consequence of the proof of our main result (Theorem~\ref{thm:main}). However, proving Theorem~\ref{thm:discrete} in this way requires handling a few technical issues, which we discuss below.

A set function $f\colon 2^\cN \to \bR$ is submodular if it obeys $f(A) + f(B) \geq f(A \cup B) + f(A \cap B)$ for every two sets $A, B \subseteq \cN$. A common way to extend set functions to continuous functions is the multilinear extension defined as follows. The name of the multilinear extension stems from the observation that it is a multilinear function of the coordinates of $\vx$.
\begin{definition}
Given a set function $f\colon 2^\cN \to \bR$, the \emph{multilinear extension} of $f$ is a function $F\colon [0, 1]^\cN \to \bR$ such that, for every $\vx \in [0, 1]^\cN$,
\[
	F(\vx)
	=
	\bE[f(\RSet(\vx))] =\sum_{S\subseteq \cN} \left[f(S) \cdot \prod_{u\in S} x_u \prod _{u\notin S}(1-x_u)\right]
	\enspace,
\]
where $\RSet(\vx)$ is a random subset of $\cN$ that includes every element $u \in \cN$ with probability $x_u$, independently.
\end{definition}

To prove Theorem~\ref{thm:discrete} using Theorem~\ref{thm:main}, we need to apply the algorithm from the last theorem to the multilinear extension $F$. However, this already raises the first technical issue that we need to discuss. Theorem~\ref{thm:main} assumes the ability to access both the DR-submodular function $F$ and its derivatives; but, in the setting of Theorem~\ref{thm:discrete}, we have access to the objective function $f$, but not to its multilinear extension $F$. Fortunately, it is possible to use sampling to approximate both $F$ and its derivatives up to a polynomially small error term with high probability (see~\cite{calinescu2011maximzing} for the details). Using such approximations, one can implement the algorithms from the proof of Theorem~\ref{thm:main} even without access to $F$ and its derivatives. Technically, the use of this method comes at the cost of introducing an additional error term whose value is equal to $E(|\cN|) \cdot \max_{S \subseteq \cN, \characteristic_S \in P} f(S)$ for a polynomially small function $E(n)$ of our choice. However, since the multiplicative constant in the approximation guarantee implied by the proof of Theorem~\ref{thm:main} is in fact slightly better than the stated value of $0.401$, this additional error term can be ignored.

The second technical issue that we need to discuss is that Theorem~\ref{thm:main} includes the error term $\delta D^2 L$, which does not appear in Theorem~\ref{thm:discrete}. For the reason explained in the previous paragraph, this error term can be ignored if its value is at most an arbitrarily small constant times $\max_{S \subseteq \cN, \characteristic_S \in P} f(S)$. Since $\delta$ is an arbitrarily small polynomial value that we can choose, it suffices to show two things: (i) $D$ is upper bounded by a polynomial function, (ii) $L$ is upper bounded by a polynomial function times $\max_{S \subseteq \cN, \characteristic_S \in P} f(S)$. The upper bound on $D$, the diameter of the convex body $P \subseteq [0, 1]^\cN$, is trivial since the diameter of the entire hypercube $[0, 1]^\cN$ is $\sqrt{|\cN|}$. Proving the necessary bound on $L$ is somewhat more involved. Recall that $L$ is the smoothness parameter of $F$. Thus, we need to upper bound the smoothness parameter of the multilinear extension of any non-negative submodular function. We are not aware of an explicit proof in the literature of such a bound, and therefore, we prove it below (as Proposition~\ref{prop:bound_L}). However, related proofs appeared, for example, in~\cite{feldman13maximization,mokhtari2020stochastic}.

We begin the proof of Proposition~\ref{prop:bound_L} with the following observation. Let $M = \max\{f(\varnothing), |\cN| \cdot \max_{u \in \cN} f(\{u\})\}$.
\begin{observation} \label{obs:trivial_bound}
If $f\colon 2^\cN \to \nnR$ is a non-negative submodular function, then for every set $S \subseteq \cN$, $f(S) \leq M$.
\end{observation}
\begin{proof}
If $S = \varnothing$, then the observation is trivial. Otherwise, by the submodularity of $f$
\[
	f(S)
	\leq
	f(\varnothing) + \sum_{u \in S} [f(\{u\}) - f(\varnothing)]
	\leq
	\sum_{u \in S} f(\{u\})
	\leq
	\sum_{u \in \cN} f(\{u\})
	\leq
	|\cN| \cdot \max_{u \in \cN} f(\{u\})
	\leq
	M
	\enspace,
\]
where the second and third inequalities follow from the non-negativity of $f$.
\end{proof}

Using the last observation, we can now provide a first bound on the smoothness of multilinear extensions.

\begin{lemma} \label{lem:derivative_change}
Let $F$ be the multilinear extension of a non-negative submodular function $f\colon 2^\cN \to \nnR$. Then, $F$ is $L$-smooth for $L = 2|\cN| \cdot M$.
\end{lemma}
\begin{proof}
Fix arbitrary two vectors $\vx, \vy \in [0, 1]^\cN$. We need to show that $\|\nabla F(\vx) - \nabla F(\vy)\|_2 \leq L \cdot \|\vx - \vy\|_2$. Let $A$ be the set of elements $u \in \cN$ with $x_u < y_u$, and let $B$ be the set of elements with $x_u > y_u$. Recall that $\RSet(x)$ is a random set containing every element $u \in \cN$ with probability $x_u$, independently. For the sake of the proof, we assume that $\RSet(\vy)$ is obtained from $\RSet(\vx)$ using the following process. Every element of $A \setminus \RSet(\vx)$ is added to a set $D$ with probability of $1 - (1 - y_u) / (1 - x_u)$, and every element of $B \cap \RSet(\vx)$ is added to $D$ with probability $1 - y_u / x_u$. Then, $\RSet(\vy)$ is chosen as $\RSet(\vx) \XOR D$.\footnote{We use $\XOR$ to indicate exclusive-OR since the symbol $\psum$ is already reserved for probabilistic sum in this paper.} Observe that every element $u \in \cN$ gets into $D$ with probability $|x_u - y_u|$, independently, and thus, $\RSet(\vx) \XOR D$ indeed has the distribution that $\RSet(\vy)$ should have.

Using the above definitions, we now get for every element $u \in \cN$,
\begin{align*}
	\frac{\partial F(\vy)}{\partial y_u}
	={} &
	\bE[f(\RSet(\vy) \cup \{u\}) - f(\RSet(\vy) \setminus \{u\})]
	=
	\bE[f((\RSet(\vx) \oplus D) \cup \{u\}) - f((\RSet(\vx) \oplus D) \setminus \{u\})]\\
	={} &
	\Pr[D = \varnothing] \cdot \bE[f(\RSet(\vx) \cup \{u\}) - f(\RSet(\vx) \setminus \{u\}) \mid D = \varnothing]\\
	&+ \Pr[D \neq \varnothing] \cdot \bE[f((\RSet(\vx) \oplus D) \cup \{u\}) - f((\RSet(\vx) \oplus D) \setminus \{u\}) \mid D \neq \varnothing]\\
	\leq{} &
	\Pr[D = \varnothing] \cdot \bE[f(\RSet(\vx) \cup \{u\}) - f(\RSet(\vx) \setminus \{u\}) \mid D = \varnothing] + \Pr[D \neq \varnothing] \cdot M
	\enspace,
\end{align*}
where the first equality holds since $F$ is a multilinear function, and the inequality holds by Observation~\ref{obs:trivial_bound} and the non-negativity of $f$.
Similarly, the non-negativity of $f$ and Observation~\ref{obs:trivial_bound} also yield
\begin{align*}
	\Pr[D &{}= \varnothing] \cdot \bE[f(\RSet(\vx) \cup \{u\}) - f(\RSet(\vx) \setminus \{u\}) \mid D = \varnothing]\\
	={} &
	\bE[f(\RSet(\vx) \cup \{u\}) - f(\RSet(\vx) \setminus \{u\}] - \Pr[D \neq \varnothing] \cdot \bE[f(\RSet(\vx) \cup \{u\}) - f(\RSet(\vx) \setminus \{u\}) \mid D \neq \varnothing]\\
	\leq{} &
	\bE[f(\RSet(\vx) \cup \{u\}) - f(\RSet(\vx) \setminus \{u\}] + \Pr[D \neq \varnothing] \cdot M
	=
	\frac{\partial F(\vx)}{\partial y_u} + \Pr[D \neq \varnothing] \cdot M
	\enspace.
\end{align*}
Combining this inequality with the previous one yields
\[
	\frac{\partial F(\vy)}{\partial y_u} - \frac{\partial F(\vx)}{\partial y_u}
	\leq
	2\Pr[D \neq \varnothing] \cdot M
	\leq
	2M \cdot \|\vy - \vx\|_1
	\leq
	2M \sqrt{\cN} \cdot \|\vy - \vx\|_2
	\enspace.
\]

By the symmetry between $\vx$ and $\vy$, the last inequality implies $\left|\frac{\partial F(\vy)}{\partial y_u} - \frac{\partial F(\vx)}{\partial y_u}\right| \leq 2M \sqrt{|\cN|} \cdot \|\vy - \vx\|_2$, which implies the lemma since every coordinate of $\nabla F(\vy) - \nabla F(\vx)$ is equal to $\frac{\partial F(\vy)}{\partial y_u} - \frac{\partial F(\vx)}{\partial y_u}$ for some element $u \in \cN$.
\end{proof}

To get a bound on $L$ in terms of $\max_{S \subseteq \cN, \characteristic_S \in P} f(S)$, it is necessary to assume that $\characteristic_u \in P$ for every $u \in \cN$. Fortunately, this assumption is without loss of generality since any element $u \in \cN$ violating the inclusion $\characteristic_u \in P$ can be dropped from $\cN$ without affecting the value of $\max_{S \subseteq \cN, \characteristic_S \in P} f(S)$.
\begin{proposition} \label{prop:bound_L}
Let $F$ be the multilinear extension of a non-negative submodular function $f\colon 2^\cN \to \nnR$. If $P \subseteq [0, 1]^\cN$ is a down-closed convex body such that $\characteristic_u \in P$ for every $u \in \cN$, then $L \leq 2|\cN|^2 \cdot \max_{S \subseteq \cN, \characteristic_S \in P} f(S)$.
\end{proposition}
\begin{proof}
In light of Lemma~\ref{lem:derivative_change}, to prove the current proposition it suffices to observe that the conditions of the proposition imply
\begin{align*}
	|\cN| \cdot M
	={} &
	|\cN| \cdot \max\{f(\varnothing), |\cN| \cdot \max_{u \in \cN} f(\{u\})\}\\
	\leq{} &
	|\cN|^2 \cdot \max\{f(\varnothing), \max_{u \in \cN} f(\{u\})\}
	\leq
	|\cN|^2 \cdot \max_{\substack{S \subseteq \cN \\ \characteristic_S \in P}} f(S)
	\enspace.
	\qedhere
\end{align*}
\end{proof}
\section{Frank-Wolfe Variant} \label{app:local_search}

In this section, we prove Theorem~\ref{thm:local_search}, which we repeat here for convenience.

\thmLocalSearch*

The algorithm that we use to prove Theorem~\ref{thm:local_search} is given as Algorithm~\ref{alg:local_search}. 
We assume in Algorithm~\ref{alg:local_search} that $\delta^{-1}$ is integral. If that is not the case, then $\delta$ can be replaced with $1 / \lceil \delta^{-1} \rceil$.

\begin{algorithm}
\caption{\texttt{Frank-Wolfe Variant}$(F, P, \delta)$} \label{alg:local_search}
Let $\vx(0)$ be an arbitrary vector in $P$.\\
\For{$i = 1$ \KwTo $\delta^{-2}$}
{
	Let $\vz(i) \in \arg\max_{\vy\in P}\inner{\vy}{\nabla F(\vx(i - 1))}$.\\
	Let $\vx(i) \gets (1 - \delta) \cdot \vx(i - 1) + \delta \cdot \vz(i)$.\\
}
Let $i^* \in \arg\min_{i=1}^{\delta^{-2}}\{\inner{\vz(i) - \vx(i - 1)}{\nabla F(\vx(i - 1))}$.\\
\Return $\vx(i^*-1)$.
\end{algorithm}

Clearly, Algorithm~\ref{alg:local_search} runs in polynomial time. The next observation shows that the output vector of Algorithm~\ref{alg:local_search} belongs to $P$.
\begin{observation} \label{obs:output_in_P}
For every $0 \leq i \leq \delta^{-2}$, $\vx(i) \in P$.
\end{observation}
\begin{proof}
We prove the observation by induction on $i$. For $i = 0$, the observation holds by the definition of $\vx(0)$. Assume now that the observation holds for $i - 1$. Since $\vz(i) \in P$ by definition and $\vx(i - 1) \in P$ by our assumption, the convexity of $P$ guarantees that $\vx(i)$ belongs to $P$ because $\vx(i)$ is defined as a convex combination of the vectors $\vx(i - 1)$ and $\vz(i)$. Thus, the observation holds for $i$ as well.
\end{proof}

Next, we would like to prove that
\begin{equation} \label{eq:contradiction_proof}
	\inner{\vy - \vx(i^* - 1)}{\nabla F(\vx(i^* - 1))}
	\leq
	\delta \cdot \left[\max_{\vy' \in P} F(\vy') + D^2 L / 2\right]
	\quad
	\forall\;	\vy \in P
	\enspace.
\end{equation}
Notice that, if $F$ is DR-submodular, then this inequality implies
\[
	F(\vx(i^* - 1))
	\geq
	\frac{1}{2}\cdot \max_{\vy \in P} [F(\vx(i^* - 1) \vee \vy)+ F(\vx(i^* - 1) \wedge \vy)]
	-
	\frac{\delta \cdot [2 \cdot \max_{\vy' \in P} F(\vy') + D^2 L]}{4}
\]
via the same arguments used in the proof of Lemma~\ref{lem:optimal_local_maximum}. Therefore, proving Inequality~\eqref{eq:contradiction_proof} suffices to complete the proof of Theorem~\ref{thm:local_search}.

Assume to the contrary that Inequality~\eqref{eq:contradiction_proof} does not hold. Then, by the chain rule, for every $i \in [\delta^{-2}]$,
\begin{align}
	F(\vx(i)) -{}& F(\vx(i - 1))
	=
	\int_0^{\delta} \inner{\vz(i) - \vx(i - 1)}{\nabla F\left(\vx(i - 1) + s(\vz(i) - \vx(i - 1))\right)} ds \nonumber \\
	={} &
	\delta \cdot \inner{\vz(i) - \vx(i - 1)}{\nabla F(\vx(i - 1))} \nonumber\\&\mspace{50mu}+ \int_0^{\delta} \inner{\vz(i) - \vx(i - 1)}{\nabla F(\vx(i - 1) + s(\vz(i) - \vx(i - 1))) - \nabla F(\vx(i - 1))} ds \nonumber \\
	\geq{} & \delta \cdot \inner{\vz(i) - \vx(i - 1)}{\nabla F(\vx(i - 1))}
	\label{ineq11}\\\nonumber&\mspace{50mu} - \int_0^{\delta} \|\vz(i) - \vx(i - 1)\|_2 \cdot \|\nabla F(\vx(i - 1) + s(\vz(i) - \vx(i - 1))) - \nabla F(\vx(i - 1))\|_2 ds\\
	\geq{} &
	\delta \cdot \inner{\vz(i) - \vx(i - 1)}{\nabla F(\vx(i - 1))} - \int_0^{\delta} Ls \cdot \|\vz(i) - \vx(i - 1)\|_2^2 ds \label{ineq12}\\
	\geq{} &
	\delta \cdot \inner{\vz(i) - \vx(i - 1)}{\nabla F(\vx(i - 1))} - \delta^2 D^2 L/2\enspace,\nonumber
\end{align}
where Inequality \eqref{ineq11} follows from the Cauchy-Schwarz inequality, and Inequality \eqref{ineq12} uses the $L$-smoothness of $F$. By the definition of $i^*$, the rightmost side of the last inequality is at least
\begin{align*}
  \delta \cdot \inner{\vz(i^*) - \vx(i^* - 1)}{\nabla F(\vx(i^* - 1))}&{} - \delta^2 D^2 L/2\\
 ={} &  \delta \cdot \max_{\vy\in P}\inner{\vy - \vx(i^* - 1)}{\nabla F(\vx(i^* - 1))} - \delta^2 D^2 L/2\\
  >{} & \delta^2 \cdot [\max_{\vy \in P} F(\vy) + D^2 L / 2] - \delta^2 D^2 L/2 = \delta^2 \cdot \max_{\vy \in P} F(\vy)\enspace, 
\end{align*}
where the equality holds due to the way $\vz(i^*)$ is chosen by Algorithm~\ref{alg:local_search}, and the inequality follows from our assumption that Inequality~\eqref{eq:contradiction_proof} does not hold. Thus, we have proved $F(\vx(i)) - F(\vx(i - 1)) > \delta^2 \cdot \max_{\vy \in P} F(\vy)$. Adding up this inequality for every integer $1 \leq i \leq \delta^{-2}$, and using the non-negativity of $F$, we now get
\[
	F(\vx(\delta^{-2}))
	\geq
	F(\vx(\delta^{-2}))- F(\vx(0))
	=
	\sum_{i = 1}^{\delta^{-2}} [F(\vx(i)) - F(\vx(i - 1))]
	>
	\max_{\vy \in P} F(\vy)
	\enspace.
\]
Since, by Observation~\ref{obs:output_in_P}, $\vx(\delta^{-2})\in P$, we get the contradiction that we were looking for.

\section{Algorithm for Unconstrained DR-Submodular Maximization} \label{app:unconstrained}

In this section, we prove Theorem~\ref{thm:unconstrained}, which we repeat here for convenience. In the proof of this theorem, we assume that $\vo$ maximizes $F(\vo)$ among all vectors in $[0, 1]^\cN$. Notice that proving Theorem~\ref{thm:unconstrained} under this assumption implies the stated version of the theorem.

\thmUnconstrained*

The algorithm that we use to prove Theorem~\ref{thm:unconstrained} appears as Algorithm~\ref{alg:double_greedy}. In the description and analysis of Algorithm~\ref{alg:double_greedy}, we assume that $\eps^{-1}$ is an even integer. If that is not the case, then we can replace $\eps$ with $1 / (2\lceil \eps^{-1} \rceil)$. 
Like all algorithms based on the \textbf{Double-Greedy} algorithm of~\cite{buchbinder2015tight}, Algorithm~\ref{alg:double_greedy} maintains two solutions $\vx$ and $\vy$ that start as $\vzero$ and $\vone$, respectively, and converge to a single solution in iterations. In each iteration, the algorithm makes one coordinate of the solutions agree. The new value assigned to this coordinate (in both $\vx$ and $\vy$) is some multiple of $\eps/n$, and the algorithm denotes by $V$ the set of all such multiples.

\begin{algorithm}
\caption{\textbf{Double-Greedy} \texttt{for DR-Submodular Functions}($F, \cN, \eps$)} \label{alg:double_greedy}
Let $V \gets \{\frac{\eps j}{n} \mid j \in \bZ, 0 \leq j \leq n\eps^{-1}\}\subseteq[0,1]$. \\
Denote the elements of $\cN$ by $u_1, u_2, \dotsc, u_n$ in an arbitrary order.\\
\BlankLine
Let $\vx \gets \vzero$ and $\vy \gets \vone$.\\
\For{$i = 1$ \KwTo $|\cN|$}
{
	Find $a_i \in \arg \max_{v \in V} F(\vx + v \cdot \characteristic_{u_i})$, and let $\Delta_{a, i} = F(\vx + a_i \cdot \characteristic_{u_i}) - F(\vx)$.\label{line:a_u}\\
	Find $b_i \in \arg \max_{v \in V} F(\vy - v \cdot \characteristic_{u_i})$, and let $\Delta_{b, i} = F(\vy - b_i \cdot \characteristic_{u_i}) - F(\vy)$.\label{line:b_u}\\
	Set both $\vx_{u_i}$ and $\vy_{u_i}$ to $\frac{\Delta_{a,i}}{\Delta_{a,i} + \Delta_{b,i}} \cdot a_i + \frac{\Delta_{b,i}}{\Delta_{a,i} + \Delta_{b,i}} \cdot (1 - b_i)$ (if $\Delta_{a,i} + \Delta_{b,i} = 0$, we assume $\frac{\Delta_{a,i}}{\Delta_{a,i} + \Delta_{b,i}} = 0$ and $\frac{\Delta_{b,i}}{\Delta_{a,i} + \Delta_{b,i}} = 1$).
}
\Return $\vx$ (or, equivalently, $\vy$).
\end{algorithm}

We begin the analysis of Algorithm~\ref{alg:double_greedy} by studying the number of accesses to $F$ necessary for implementing it.
\begin{observation}
Every iteration of Algorithm~\ref{alg:double_greedy} can be implemented using $O(\log (n\eps^{-1}))$ function evaluations, and therefore, the entire algorithms requires $O(n \log (n\eps^{-1}))$ function evaluations.
\end{observation}
\begin{proof}
Access to the function $F$ is necessary only for implementing Lines~\ref{line:a_u} and~\ref{line:b_u} of Algorithm~\ref{alg:double_greedy}. We explain below why Line~\ref{line:a_u} can be implemented using $O(\log (n\eps^{-1}))$ function evaluations. A similar explanation applies also to Line~\ref{line:b_u}. Line~\ref{line:a_u} asks us to find a value $v \in V$ maximizing $F(\vx + v \cdot \characteristic_{u_i})$.
Since $F$ is DR-submodular, the expression $F(\vx + v \cdot \characteristic_{u_i})$ is a concave function of $v\in[0,1]$. Thus, when restricted to the discrete values $v\in V$, $F(\vx + v \cdot \characteristic_{u_i})$ forms a series of three parts: an increasing part, a constant part, and a decreasing part (one or two of these parts might be missing). Thus, for a given value $\eps n^{-1}\cdot j \in V \setminus \{1\}$, it is to possible to get information by computing the sign of the difference
\[
	F(\vx + \eps n^{-1} \cdot (j + 1) \cdot \characteristic_{u_i}) - F(\vx + \eps n^{-1} \cdot j \cdot \characteristic_{u_i})
	\enspace.
\]
If this sign is positive, then the maximum of $F(\vx + v \cdot \characteristic_{u_i})$ must be obtained either at $v = \eps n^{-1}\cdot (j+1)$ or to its right; if this sign is negative, then the maximum of $F(\vx + v \cdot \characteristic_{u_i})$ must be obtained either at $v = \eps n^{-1}\cdot j$ or to its left; and if the difference is $0$, then the (non-unique) maximum of $F(\vx + v \cdot \characteristic_{u_i})$ is obtained at $v=\eps n^{-1}\cdot j$. Given these observations, one can use binary search to find in $O(\log |V|) = O(\log (n\eps^{-1}))$ time a value $\eps n^{-1} \cdot j \in V$ maximizing $F(\vx + v \cdot \characteristic_{u_i})$. 
\end{proof}

Our next goal is to analyze the value of the output vector of Algorithm~\ref{alg:double_greedy}. For that purpose, let us denote by $\vx^{(i)}$ and $\vy^{(i)}$ the vectors $\vx$ and $\vy$ after $i$ iterations of the main loop of Algorithm~\ref{alg:double_greedy}. The next lemma shows that even though $a_i$ and $b_i$ are defined as maxima only over the values in $V$, they are in fact close to be maxima over the entire range $[0, 1]$.
\begin{lemma} \label{lem:approx_max}
For every integer $1 \leq i \leq n$,
\[
	F(\vx^{(i - 1)} + a_i \cdot \characteristic_{u_i})
	\geq
	\max_{v \in [0, 1]} F(\vx^{(i - 1)} + v \cdot \characteristic_{u_i}) - (2\eps/n) \cdot F(\vo)
\]
and
\[
	F(\vy^{(i - 1)} - b_i \cdot \characteristic_{u_i})
	\geq
	\max_{v \in [0, 1]} F(\vy^{(i - 1)} - v \cdot \characteristic_{u_i}) - (2\eps/n) \cdot F(\vo)
	\enspace.
\]
\end{lemma}
\begin{proof}
We only prove the part of the lemma related to $a_i$. The proof of the other part of the lemma is analogous. Let $v^*$ be a value of $v$ maximizing $F(\vx^{(i - 1)} + v \cdot \characteristic_{u_i})$ There are two cases to consider depending on the value of $v^*$. If $v^* \geq 1/2$, then let $v$ be the largest value in $V$ that is not larger than $v^*$. Then,
\begin{align} \label{eq:marginal_bound}
	F(\vx^{(i - 1)} + v^* \cdot \characteristic_{u_i}) -{}& F(\vx^{(i - 1)} + a_i \cdot \characteristic_{u_i})
	\leq
	F(\vx^{(i - 1)} + v^* \cdot \characteristic_{u_i}) - F(\vx^{(i - 1)} + v \cdot \characteristic_{u_i})\\\nonumber
	\leq{} &
	(v^* - v) \cdot \left. \frac{\partial F(\vc)}{c_{u_i}} \right|_{c = \vx^{(i - 1)} + v \cdot \characteristic_{u_i}}
	\leq
	\frac{v^* - v}{v} \cdot [F(\vx^{(i - 1)} + v \cdot \characteristic_{u_i}) - F(\vx^{(i - 1)})]\\\nonumber
	\leq{} &
	\frac{v^* - v}{v} \cdot F(\vx^{(i - 1)} + v \cdot \characteristic_{u_i})
	\leq
	\frac{2\eps}{n} \cdot F(\vx^{(i - 1)} + v \cdot \characteristic_{u_i})
	\leq
	\frac{2\eps}{n} \cdot F(\vo)
	\enspace,
\end{align}
where the first inequality holds by the definition of $a_i$, the second and third inequalities follow from Properties~\ref{prop:dr_bound2_up} and~\ref{prop:dr_bound2_down} of Lemma~\ref{lem:DR_properties}, respectively, the fourth inequality holds due to $F$'s non-negativity, and the last inequality follows from our assumption that $F(\vo)$ maximizes $F$ within $[0, 1]^\cN$. The penultimate inequality holds since the definition of $v$ guarantees that $v^* - v \leq \eps/n$, and the case we consider and our assumption that $\eps^{-1}$ is an even integer imply together that $v \geq 1/2$. Using Inequality~\eqref{eq:marginal_bound}, we now get
\[
	\max_{v \in [0, 1]} F(\vx_{i - 1} + v \cdot \characteristic_{u_i})
	=
	F(\vx_{i - 1} + v^* \cdot \characteristic_{u_i})
	\leq
	F(\vx^{(i - 1)} + a_i \cdot \characteristic_{u_i}) + \frac{2\eps}{n} \cdot F(\vo)
	\enspace,
\]
which completes the proof of the lemma for the case of $v^* \geq 1/2$.

Consider now the case of $v^* \leq 1/2$. In this case, we choose $v$ as the minimal value in $V$ that is not smaller than $v^*$. Then,
\begin{align*}
	F(\vx^{(i - 1)} +{}& a_i \cdot \characteristic_{u_i}) - F(\vx^{(i - 1)} + v^* \cdot \characteristic_{u_i})
	\geq
	F(\vx^{(i - 1)} + v \cdot \characteristic_{u_i}) - F(\vx^{(i - 1)} + v^* \cdot \characteristic_{u_i})\\
	\geq{} &
	(v - v^*) \cdot \left. \frac{\partial F(\vc)}{c_{u_i}} \right|_{c = \vx^{(i - 1)} + v \cdot \characteristic_{u_i}}
	\geq
	\frac{v - v^*}{1 - v} \cdot [F(\vx^{(i - 1)} + \characteristic_{u_i}) - F(\vx^{(i - 1)} + v \cdot \characteristic_{u_i})]\\
	\geq{} &
	-\frac{v - v^*}{1 - v} \cdot F(\vx^{(i - 1)} + v \cdot \characteristic_{u_i})
	\geq
	- \frac{2\eps}{n} \cdot F(\vx^{(i - 1)} + v \cdot \characteristic_{u_i})
	\geq
	- \frac{2\eps}{n} \cdot F(\vo)
	\enspace,
\end{align*}
where the first inequality holds by the definition of $a_i$, the second and third inequalities follow from Properties~\ref{prop:dr_bound2_up} and~\ref{prop:dr_bound2_down} of Lemma~\ref{lem:DR_properties}, the fourth inequality holds due to $F$'s non-negativity, and the last inequality follows from our assumption that $F(\vo)$ maximizes $F$ within $[0, 1]^\cN$. The penultimate inequality holds since the definition of $v$ guarantees that $v - v^* \leq \eps/n$, and the case we consider and our assumption that $\eps^{-1}$ is an even integer imply together that $v \leq 1/2$. The rest of the proof of the current case is identical to the corresponding part of the proof of the other case, and is thus, omitted.
\end{proof}

Next, we would like to lower bound the increase in the values of $F(\vx)$ and $F(\vy)$ during each iteration.
\begin{lemma} \label{lem:gain}
For every integer $1 \leq i \leq n$,
\[
	F(\vx^{(i)}) - F(\vx^{(i - 1)})
	\geq
	\frac{\Delta_{a,i}^2}{\Delta_{a,i} + \Delta_{b,i}}
	\qquad
	\text{and}
	\qquad
	F(\vy^{(i)}) - F(\vy^{(i - 1)})
	\geq
	\frac{\Delta_{b,i}^2}{\Delta_{a,i} + \Delta_{b,i}}
	\enspace.
\]
\end{lemma}
\begin{proof}
We prove only the first part of the lemma. The proof of the other part is analogous. By Property~\ref{prop:dr_bound1} of Lemma~\ref{lem:DR_properties}, 
\begin{align*}
	\mspace{99mu}&\mspace{-99mu}
	F\left(\vx^{(i - 1)} + \frac{\Delta_{a,i}}{\Delta_{a,i} + \Delta_{b,i}} \cdot a_i\cdot \characteristic_{u_i} + \frac{\Delta_{b,i}}{\Delta_{a,i} + \Delta_{b,i}} \cdot (1 - b_i)\cdot \characteristic_{u_i}\right) - F(\vx^{(i - 1)})\\
 ={}& F\left(\frac{\Delta_{a,i}}{\Delta_{a,i} + \Delta_{b,i}} \cdot\left(\vx^{(i - 1)} + a_i \cdot \characteristic_{u_i}\right)  + \frac{\Delta_{b,i}}{\Delta_{a,i} + \Delta_{b,i}} \left(\vx^{(i - 1)} + (1 - b_i) \cdot \characteristic_{u_i}\right) \right) - F(\vx^{(i - 1)})\\
	\geq{} &
	\frac{\Delta_{a,i}}{\Delta_{a,i} + \Delta_{b,i}} \cdot [F(\vx^{(i - 1)} + a_i \cdot \characteristic_{u_i}) - F(\vx^{(i - 1)})] \\&\mspace{200mu}+ \frac{\Delta_{b,i}}{\Delta_{a,i} + \Delta_{b,i}} \cdot [F(\vx^{(i - 1)} + (1 - b_i) \cdot \characteristic_{u_i}) - F(\vx^{(i - 1)})]\\
	\geq{} &
	\frac{\Delta_{a,i}}{\Delta_{a,i} + \Delta_{b,i}} \cdot [F(\vx^{(i - 1)} + a_i \cdot \characteristic_{u_i}) - F(\vx^{(i - 1)})] \\&\mspace{200mu}+ \frac{\Delta_{b,i}}{\Delta_{a,i} + \Delta_{b,i}} \cdot [F(\vy^{(i - 1)} - b_i \cdot \characteristic_{u_i}) - F(\vy^{(i - 1)} - \characteristic_{u_i})]\\
	\geq{} &
	\frac{\Delta_{a,i}}{\Delta_{a,i} + \Delta_{b,i}} \cdot [F(\vx^{(i - 1)} + a_i \cdot \characteristic_{u_i}) - F(\vx^{(i - 1)})]
	=
	\frac{\Delta_{a,i}^2}{\Delta_{a,i} + \Delta_{b,i}}
	\enspace,
\end{align*}
where the second inequality follows from the DR-submodularity of $F$, and the last inequality holds by the definition of $b_i$ since $1 \in V$.
\end{proof}

As is standard in the analysis of algorithms based on \textbf{Double-Greedy}, we also need to define a hybrid solution that starts as $\vo$ before the first iteration of the algorithm, and becomes more and more similar to the output of the algorithm as the algorithm progresses. Formally, for every integer $0 \leq i \leq n$, we define $\vo^{(i)} = (\vo \vee \vx^{(i)}) \wedge \vy^{(i)}$. Notice that indeed $\vo^{(0)} = \vo$ and $\vo^{(n)} = \vx^{(n)}$ because $\vx^{(n)} = \vy^{(n)}$. The next lemma bounds the decrease in $F(\vo^{(i)})$ as a function of $i$.
\begin{lemma} \label{lem:loss}
For every integer $1 \leq i \leq n$,
\[
	F(\vo^{(i - 1)}) - F(\vo^{(i)})
	\leq
	\frac{\Delta_{a,i} \cdot \Delta_{b,i}}{\Delta_{a,i} + \Delta_{b,i}} - \frac{2\eps}{n} \cdot F(\vo)
	\enspace.
\]
\end{lemma}
\begin{proof}
There are two cases to consider based on the relationship between $o_{u_i}^{(i - 1)}$ and $o_{u_i}^{(i)}$. We consider here only the case of $o_{u_i}^{(i - 1)} \leq o_{u_i}^{(i)}$. The other case is analogous. By the DR-submoduarlity of $F$ and the fact that $\vy^{(i - 1)} \geq \vo^{(i - 1)}$ by the definition of $\vo^{(i - 1)}$,
\[
	F(\vo^{(i)}) - F(\vo^{(i - 1)})
	\geq
	F(\vy^{(i - 1)} - (1 - o_{u_i}^{(i)}) \cdot \characteristic_{u_i}) - F(\vy^{(i - 1)} - (1 - o_{u_i}^{(i - 1)}) \cdot \characteristic_{u_i})
	\enspace.
\]
Let us bound the two terms on the right hand side of the last inequality. First, by Lemma~\ref{lem:approx_max},
\[
	F(\vy^{(i - 1)} - (1 - o_{u_i}^{(i - 1)}) \cdot \characteristic_{u_i})
	\leq
	F(\vy^{(i - 1)} - b_i \cdot \characteristic_{u_i}) - (2\eps / n) \cdot F(\vo)
	\enspace.
\]
Second, by Property~\ref{prop:dr_bound1} of Lemma~\ref{lem:DR_properties},
\begin{align*}
	F(\vy^{(i - 1)} - (1 - o_{u_i}^{(i)}&) \cdot \characteristic_{u_i})\\
	={} &
	F\left(\vy^{(i - 1)} - \left(1 - \frac{\Delta_{a,i}}{\Delta_{a,i} + \Delta_{b,i}} \cdot a_i - \frac{\Delta_{b,i}}{\Delta_{a,i} + \Delta_{b,i}} \cdot (1 - b_i)\right) \cdot \characteristic_{u_i}\right)\\
	\geq{} &
	\frac{\Delta_{a,i}}{\Delta_{a,i} + \Delta_{b,i}} \cdot F(\vy^{(i - 1)} - (1 - a_i) \cdot \characteristic_{u_i}) + \frac{\Delta_{b,i}}{\Delta_{a,i} + \Delta_{b,i}} \cdot F(\vy^{(i - 1)} - b_i \cdot \characteristic_{u_i})
	\enspace.
\end{align*}
Combining all the above inequalities, we get
\begin{align*}
	F(&\vo^{(i)}) - F(\vo^{(i - 1)})\\
	\geq{} &
	\frac{\Delta_{a,i}}{\Delta_{a,i} + \Delta_{b,i}} \cdot F(\vy^{(i - 1)} - (1 - a_i) \cdot \characteristic_{u_i}) + \left(\frac{\Delta_{b,i}}{\Delta_{a,i} + \Delta_{b,i}} - 1\right) \cdot F(\vy^{(i - 1)} - b_i \cdot \characteristic_{u_i}) + \frac{2\eps}{n} \cdot F(\vo)\\
	={} &
	\frac{\Delta_{a,i}}{\Delta_{a,i} + \Delta_{b,i}} \cdot \left[F(\vy^{(i - 1)} - (1 - a_i) \cdot \characteristic_{u_i}) - F(\vy^{(i - 1)} - b_i \cdot \characteristic_{u_i})\right] + \frac{2\eps}{n} \cdot F(\vo)\\
	={} &
	\frac{\Delta_{a,i}}{\Delta_{a,i} + \Delta_{b,i}} \cdot \left[F(\vy^{(i - 1)} - (1 - a_i) \cdot \characteristic_{u_i}) - F(\vy^{(i - 1)}) - \Delta_{b, i}\right] + \frac{2\eps}{n} \cdot F(\vo)\\
	\geq{} &
	\frac{\Delta_{a,i}}{\Delta_{a,i} + \Delta_{b,i}} \cdot \left[F(\vx^{(i - 1)} + a_i \cdot \characteristic_{u_i}) - F(\vx^{(i - 1)} + \characteristic_{u_i}) - \Delta_{b, i}\right] + \frac{2\eps}{n} \cdot F(\vo)\\
	\geq{} &
	- \frac{\Delta_{a,i} \cdot \Delta_{b, i}}{\Delta_{a,i} + \Delta_{b,i}} + \frac{2\eps}{n} \cdot F(\vo)
	\enspace,
\end{align*}
where the second inequality holds by the DR-submodularity of $F$ since $\vx^{(i - 1)} + a_i \cdot \characteristic_{u_i} \leq \vy^{(i - 1)} - (1 - a_i) \cdot \characteristic_{u_i}$, and the last inequality follows from the definition of $a_i$ since $1 \in V$. The lemma now follows by rearranging the last inequality.
\end{proof}

\begin{corollary} \label{cor:relation}
For every $r \geq 0$ and integer $1 \leq i \leq n$,
\[
	\frac{1}{r} \cdot [F(\vx^{(i)}) - F(\vx^{(i - 1)})] + r \cdot [F(\vy^{(i)}) - F(\vy^{(i - 1)})]
	\geq
	2[F(\vo^{(i - 1)}) - F(\vo^{(i)})] - \frac{4\eps}{n} \cdot F(\vo)
	\enspace.
\]
\end{corollary}
\begin{proof}
By Lemmata~\ref{lem:gain} and~\ref{lem:loss},
\begin{align*}
	\frac{1}{r} \cdot [F(\vx^{(i)}) - F(\vx^{(i - 1)})] &{}+ r \cdot [F(\vy^{(i)}) - F(\vy^{(i - 1)})]
	\geq
	\frac{(1/r) \cdot \Delta_{a,i}^2}{\Delta_{a,i} + \Delta_{b,i}} + \frac{r \cdot \Delta_{b,i}^2}{\Delta_{a,i} + \Delta_{b,i}}\\
	={} &
	\frac{(\Delta_{a,i} / \sqrt{r} - \Delta_{b,i} \cdot \sqrt{r})^2 + 2 \cdot \Delta_{a,i} \cdot \Delta_{b,i}}{\Delta_{a,i} + \Delta_{b,i}}
	\geq
	2 \cdot \frac{\Delta_{a,i} \cdot \Delta_{b,i}}{\Delta_{a,i} + \Delta_{b,i}}\\
	\geq{} &
	2 \cdot [F(\vo^{(i - 1)}) - F(\vo^{(i)})] - \frac{4\eps}{n} \cdot F(\vo)
	\enspace.
	\qedhere
\end{align*}
\end{proof}

We are now ready to prove Theorem~\ref{thm:unconstrained}.
\begin{proof}[Proof of Theorem~\ref{thm:unconstrained}.]
For every $r \geq 0$, adding up Corollary~\ref{cor:relation} for all integers $1 \leq i \leq n$ yields
\begin{align*}
	\frac{1}{r} \cdot [F(\vx^{(n)}) - F(\vx^{(0)}&)] + r \cdot [F(\vy^{(n)}) - F(\vy^{(0)})]\\
	={} &
	\sum_{i = 1}^n\left\{\frac{1}{r} \cdot [F(\vx^{(i)}) - F(\vx^{(i - 1)})] + r \cdot [F(\vy^{(i)}) - F(\vy^{(i - 1)})]\right\}\\
	\geq{} &
	\sum_{i = 1}^n\left\{2[F(\vo^{(i - 1)}) - F(\vo^{(i)})] - \frac{4\eps}{n} \cdot F(\vo)\right\}
	=
	2[F(\vo^{(0)}) - F(\vo^{(n)})] - 4\eps \cdot F(\vo)
	\enspace.
\end{align*}
Recalling that $\vo^{(0)} = \vo$, $\vx^{(n)} = \vy^{(n)} = \vo^{(n)}$, $\vx^{(0)} = \vzero$ and $\vy^{(0)} = \vone$, the last inequality implies
\[
	\frac{1}{r} \cdot [F(\vx^{(n)}) - F(\vzero)] + r \cdot [F(\vx^{(n)}) - F(\vone)]
	\geq
	2[F(\vo) - F(\vx^{(n)})] - 4\eps \cdot F(\vo)
	\enspace,
\]
and rearranging this inequality yields
\begin{align*}
	F(\vx^{(n)})
	\geq{} &
	\frac{(2 - 4\eps) \cdot F(\vo) + \frac{1}{r} \cdot F(\vzero) + r \cdot F(\vone)}{r + 2 + 1/r}\\
	\geq{} &
	\left(\frac{2r}{(r + 1)^2} - \eps\right) \cdot F(\vo) + \frac{1}{(r + 1)^2} \cdot F(\vzero) + \frac{r^2}{(r + 1)^2} \cdot F(\vone)
	\enspace.
	\qedhere
\end{align*}
\end{proof}
\section{Guessing the Necessary Values} \label{sec:knowledge}

In this section, we prove Lemma~\ref{lem:guess_tau}, which we repeat here for convenience.

\lemGuess*

\begin{proof}
Recall that $F(\vo)$ is the optimal value for the problem $\max_{\vx \in P} F(\vx)$. As discussed in Section~\ref{sec:introduction}, many previous constant approximation ratio algorithms have been suggested for this problem, and by executing any of these algorithms, we can get a value $v$ obeying $c \cdot F(\vo) \leq v \leq F(\vo)$ for some absolute constant $c \in (0, 1)$. Consider now the set
\[
	G_{\vo}
	=
	\{(1 - \eps)^i \cdot v/c \mid i \in \bZ, 1 \leq i \leq \log_{1 - \eps} c + 1\}
	\enspace.
\]
One can observe that the size of this set is at most $\log_{1 - \eps} c + 1$, and therefore, depends only on $\eps$ and the absolute constant $c$.

Note now that, for $i = 0$, $(1 - \eps)^i \cdot v/c = v/c \geq F(\vo)$, and for $i = \lceil \log_{1 - \eps} c \rceil$,
\[
	(1 - \eps)^i \cdot v/c
	\leq
	(1 - \eps)^{\log_{1 - \eps} c} \cdot v/c
	=
	v
	\leq
	F(\vo)
	\enspace.
\]
Hence, there must exist an integer $i \in [\log_{1 - \eps} c + 1]$ such that $(1 - \eps)^i \cdot v/c \leq F(\vo)$, but $(1 - \eps)^{i - 1} \cdot v/c \geq F(\vo)$. This implies that the value $g \triangleq (1 - \eps)^i \cdot v/c$ belongs to the set $G_{\vo}$ and has the properties stated in the lemma.

Our next objective is to show that, given the above value $g$, it is possible to construct in polynomial time constant size (depending only on $\eps$) sets $G_{\hprod}(g)$ and $G_{\psum}(g)$ of non-negative values such that $G_{\hprod}(g)$ is guaranteed to include a value obeying the inequalities of the lemma involving $g_{\hprod}$, and $G_{\psum}(g)$ is guaranteed to include a value obeying the inequalities of the lemma involving $g_{\psum}$. Notice that once we prove that, the lemma will follow by simply choosing $G \triangleq \{\{g\} \times G_{\hprod}(g) \times G_{\psum}(g) \mid g \in G_{\vo}\}$.

We begin by defining
\[
	G_{\psum}(g)
	\triangleq
	\left\{\eps i \cdot g ~\middle|~ i \in \bZ,  1 \leq i \leq \frac{2}{\eps(1 - \eps)} + 1\right\}
	\enspace.
\]
Notice that the size of this set indeed depends only on $\eps$, and there is a value in this set that can be chosen as $g_{\hprod}$ because
\[
	\eps \cdot 0 \cdot g
	=
	0
	\leq
	F(\vz \psum \vo)
	\leq
	2 \cdot F(\vo)
	=
	\eps \cdot \frac{2}{\eps} \cdot F(\vo)
	\leq
	\eps \cdot \frac{2}{\eps(1 - \eps)} \cdot g
	\enspace,
\]
where the first inequality holds by the non-negativity of $F$, and the last inequality follows from the properties of $g$.
The second inequality follows since the DR-submodularity and non-negativity of $F$ imply together that
\begin{align} \label{eq:psum_bound}
	F(\vz \psum \vo)
	={} &
	F(\vo + (1 - \vo) \hprod \vz)
	\leq
	F(\vo) + F((1 - \vo) \hprod \vz) - F(\vzero)\\\nonumber
	\leq{} &
	F(\vo) + F((1 - \vo) \hprod \vz)
	\leq
	2 \cdot F(\vo)
	\enspace,
\end{align}
where the last inequality follows since $(1 - \vo) \hprod \vz\leq \vo$, and hence, belongs to $P$ by the down-closeness of $P$ (recall that $F(\vo)$ upper bounds $F(\vx)$ for any $\vx \in P$ by the definition of $\vo$).

The construction of $G_{\hprod}(g)$ is very similar to the above construction of $G_{\psum}(g)$; with the sole difference being that Inequality~\eqref{eq:psum_bound} is replaced with the inequality $F(\vz \hprod \vo) \leq F(\vo)$, which holds since the down-closeness of $P$ implies that $\vz \hprod \vo$ belongs $P$.
%
\end{proof}

\end{document}